\documentclass{scrartcl}

\usepackage[utf8]{inputenc}
\usepackage[T1]{fontenc}

\usepackage{hyperref}
\usepackage{amsmath}
\usepackage{amssymb}
\usepackage{amsthm}
\usepackage{csquotes}
\usepackage{stmaryrd}
\usepackage{tikz}
\usepackage{hhline}
\usepackage{enumitem}
\usepackage{algpseudocode}
\usepackage{varwidth}
\usepackage{thm-restate}
\usepackage{multicol}
\usepackage{wrapfig}

\usepackage{todonotes}

\usetikzlibrary{shapes, positioning, decorations.pathmorphing, decorations.pathreplacing,
  matrix, calc, automata, fit}

\declaretheorem[style=plain]{theorem}
\declaretheorem[style=plain]{lemma}
\declaretheorem[style=plain]{corollary}
\declaretheorem[style=plain]{fact}

\declaretheorem[style=definition]{definition}

\declaretheorem[style=remark]{example}

\newenvironment{case}[1]{\medskip\par{\normalfont\normalsize\itshape{Case #1.}}\;}{}

\newcommand*{\FO}{\mathsf{FO}}
\newcommand*{\NL}{\ensuremath{\textnormal{NL}}}

\newcommand*{\V}[1][V]{\boldsymbol{\mathrm{#1}}}
\newcommand*{\DA}{\V[DA]}
\newcommand*{\Rm}[1][m]{\V[R_{#1}]}
\newcommand*{\Lm}[1][m]{\V[L_{#1}]}
\newcommand*{\WI}{\textnormal{WI}}
\newcommand*{\R}{\mathrel{\mathcal{R}}}
\renewcommand*{\L}{\mathrel{\mathcal{L}}}
\newcommand*{\J}{\mathrel{\mathcal{J}}}

\newcommand*{\Nat}{\mathbb{N}}

\newcommand*{\dom}[1]{\operatorname{dom}(#1)}
\newcommand*{\alphabet}{\operatorname{alph}}
\newcommand*{\landO}{\mathcal{O}}
\newcommand{\subs}[2]{\ensuremath{\llbracket #1 \rrbracket_{#2}}}

\newcommand*{\var}[1]{\texttt{#1}}

\newcommand{\smallset}[1]{\left\{#1\right\}}
\newcommand{\set}[2]{\left\{#1\mathrel{\left|\vphantom{#1}\vphantom{#2}\right.}#2\right\}}
\newcommand*{\onto}{\twoheadrightarrow}

\newsavebox{\malcevbox}
\savebox{\malcevbox}{\tikz[baseline=-.75ex]{%
    \node [shape=circle,draw,inner sep=0.4pt,scale=0.85] (char) {\ensuremath{\hspace*{0.1mm}m}};}}

\newcommand{\malcev}{\makeatletter%
  \def\c@rcled{\usebox{\malcevbox}}%
  \mathbin{\smash{\mathchoice{\text{\small\ensuremath{\c@rcled}}}%
    {\text{\small\ensuremath{\c@rcled}}}%
    {\text{\scriptsize\ensuremath{\c@rcled}}}%
    {\text{\tiny\ensuremath{\c@rcled}}}}}%
  \makeatother}

\title{The Word Problem for Omega-Terms \\ over the Trotter-Weil Hierarchy\thanks{The final publication is available at Springer via \href{http://dx.doi.org/10.1007/s00224-017-9763-z}{http://dx.doi.org/10.1007/s00224-017-9763-z}.}}

\usepackage{authblk}
\author{Manfred Kuf\-leitner\thanks{The first author was supported by the German Research Foundation (DFG) under grants \mbox{DI 435/5-2} and \mbox{KU 2716/1-1}.}~}
\author{Jan Philipp Wächter}
\affil{\small%
  Institut für Formale Methoden der Informatik\\%
  University of Stuttgart, Germany\\%
  \texttt{$\{$kufleitner,jan-philipp.waechter$\}$@fmi.uni-stuttgart.de}%
}

\begin{document}

\maketitle

\begin{abstract}
  \noindent\textbf{Abstract.}
For two given $\omega$-terms $\alpha$ and $\beta$, the word problem for $\omega$-terms over a variety $\V$ asks whether $\alpha=\beta$ in all monoids in $\V$.
We show that the word problem for $\omega$-terms over each level of the
Trotter-Weil Hierarchy is decidable. More precisely, for every fixed variety  in the Trotter-Weil
Hierarchy, our approach yields an algorithm in nondeterministic logarithmic space ($\NL$). In
addition, we provide deterministic polynomial time algorithms which are more efficient than
straightforward translations of the NL-algorithms. As an application of our results, we show that
separability by the so-called corners of the Trotter-Weil Hierarchy is witnessed by $\omega$-terms
(this property is also known as $\omega$-reducibility). In particular, the separation problem for
the corners of the Trotter-Weil Hierarchy is decidable.
\end{abstract}

\section{Introduction}

Algebraic characterizations of classes of regular languages are interesting as they
often allow to decide the class's membership problem. For example, by Schützenberger's famous
theorem~\cite{schutzenberge1965finite}, one can decide whether a given regular language is star-free
by computing its syntactic monoid $M$ and checking its aperiodicity. The latter can be achieved
by verifying $x^{|M|!} = x^{|M|!}x$ for all $x \in M$. This equation is also
stated as $x^\omega = x^\omega x$ since this notation is independent of the monoid's size. More
formally, we can see the equation as a pair of $\omega$-terms: these are finite words built using
letters, which are interpreted as variables, concatenation and an additional formal $\omega$-power.
Checking an equation $\alpha = \beta$ in a finite monoid is easy: one can simply substitute each
variable by all elements of the monoid. For each substitution, this yields a monoid element on the
left hand side and one on the right hand side. The equation holds if and only if they are always
equal.

Often, the question whether an equation holds is not only interesting for a single finite monoid but
for a (possibly infinite) class of such monoids. For example, one may ask whether all monoids in a
certain class are aperiodic. This is trivially decidable if the class is finite. But what if it is
infinite? If the class forms a variety (of finite monoids, sometimes also referred to as a
pseudo-variety), i.\,e.\ a class of finite monoids closed under (possibly empty) direct products,
submonoids and homomorphic images, then this problem is called the variety's \emph{word problem for
$\omega$-terms}. Usually, the study of a variety's word problem for $\omega$-terms also gives more
insight into the variety's structure, which is interesting in its own right.
McCammond showed that the word problem for $\omega$-terms of the variety $\V[A]$ of
aperiodic finite monoids is decidable \cite{mccammond2001normal}. The problem was shown to be
decidable in linear time for $\V[J]$, the class of $\mathcal{J}$-trivial finite monoids, by Almeida
\cite{almeida2002finite}  and for $\V[R]$, the class of $\mathcal{R}$-trivial monoids, by Almeida
and Zeitoun \cite{almeida2007automata}. For the variety $\DA$, Moura adapted and expanded those ideas to show decidability in time $\mathcal{O}((nk)^{5})$ where $n$ is the length of the input
$\omega$-terms and $k$ is the maximal nesting
depth of the $\omega$-power (which can be linear in $n$) \cite{moura2011word}. Remember that $\DA$ is the class of finite monoids whose regular $\mathcal{D}$-classes form aperiodic semigroups. This variety received a lot of attention due to its many different characterizations; see e.g.~\cite{DiekertGastinKufleitner08:short,tesson2002diamonds}. Most notably is its connection to two-variable first-order logic~\cite{therien1998over}. This logic is a natural restriction of
first-order logic over finite words, which in turn is the logic characterization of
$\V[A]$.

In this paper, we consider the word problem for $\omega$-terms over the varieties in the
\emph{Trotter-Weil Hierarchy}. It was introduced by Trotter and Weil~\cite{trotter1997lattice} with the idea of
using the good understanding of the band varieties (cf.~\cite{gerhard1989varieties}) for studying the
lattice of sub-varieties of $\DA$; bands are semigroups satisfying $x^2 = x$. The levels of this hierarchy exhaust $\DA$.
As it turns out, the Trotter-Weil Hierarchy has tight connections to the quantifier alternation hierarchy inside two-variable first-order logic~\cite{kufleitner2012alternation}. In addition, many
characterizations of $\DA$ admit natural restrictions which allow climbing up this hiearchy (see~\cite{kufleitner2012join}).

Please note that, in spite of this paper's title, we will refer to $\omega$-terms as $\pi$-terms for
most parts of the paper, the only exception being this introduction. We follow this notation introduced
by Perrin and Pin~\cite{perrin2004infinite} to avoid notational conflicts. Accordingly, we use $\pi$
for the formal power in $\pi$-terms and speak of the word problem for $\pi$-terms.

\paragraph*{Results.}
In this paper, we present the following results.\vspace{-0.5\baselineskip}
\begin{itemize}

  \item Our main tool for studying a variety $\V[V]$ of the Trotter-Weil Hierarchy
  is a family of finite index congruences $\equiv_{\V[V],n}$ for $n \in \Nat$. These congruences have the
  property that a monoid $M$ is in $\V$ if and only if there exists $n$ for which $M$ divides a quotient
  by $\equiv_{\V[V],n}$. The congruences are not new but they differ in some details from the ones usually found in the literature, where they are introduced in terms of rankers~\cite{kufleitner2012join,kufleitner2012alternation,kufleitner2012logical}. Unfortunately, these differences necessitate new proofs.

  \enlargethispage{\baselineskip}
  \item We lift the combinatorics from finite words to $\omega$-terms using the \enquote{linear
  order approach} introduced by Huschenbett and the first author~\cite{huschenbett2013ehrenfeucht}.
  They showed that, over varieties of aperiodic monoids, one can use the order $\mathbb{N} +
  \mathbb{Z} \cdot \mathbb{Q} + (-\mathbb{N})$ for the formal $\omega$-power. In this paper, we
  use the simpler order $\mathbb{N} + (-\mathbb{N})$. We show that two $\omega$-terms $\alpha$
  and $\beta$ are equal in some variety $\V[V]$ of the Trotter-Weil Hierarchy if and only if
  $\subs{\alpha}{\mathbb{N} + (-\mathbb{N})} \equiv_{\V[V],n} \subs{\beta}{\mathbb{N} +
  (-\mathbb{N})}$ for all $n \in \mathbb{N}$. Here, $\subs{\alpha}{\mathbb{N} + (-\mathbb{N})}$
  denotes the generalized word (i.\,e.\ the labeled linear order) obtained from replacing every occurrence of the formal $\omega$-power by the linear
  order $\mathbb{N} + (-\mathbb{N})$. Note that this order is tailor-made for the Trotter-Weil
  Hierarchy and does not result from simple arguments which work in any variety.

  \item We show that one can effectively check whether $\subs{\alpha}{\mathbb{N} + (-\mathbb{N})}
  \equiv_{\V,n} \subs{\beta}{\mathbb{N} + (-\mathbb{N})}$ for all $n \in \mathbb{N}$.

  \item We further improve the algorithms and show that, for every variety $\V$ of the Trotter-Weil
  Hierarchy, the word problem for $\omega$-terms over $\V$ is decidable in nondeterministic
  logarithmic space. The main difficulty is to avoid some blow-up which (naively) is caused by the
  nesting depth of the $\omega$-power. For $\V[R]$,
  which appears in the hierarchy, this result is incomparable to Almeida and Zeitoun's linear time
  algorithm~\cite{almeida2007automata}.

  \item We also introduce polynomial time algorithms, which are more efficient than the direct
  translation of these $\NL$ algorithms.

  \item As an application, we prove that the separation problems for the so-called \emph{corners}
  of the Trotter-Weil Hierarchy are decidable by showing $\omega$-reducibi\-li\-ty. For $\V[J]$, we adapt the proof of van Rooijen and Zeitoun
  \cite{roojien2013separation}.

  \item With little additional effort, we also obtain all of the above results for $\DA$, the limit
  of the Trotter-Weil Hierarchy. The decidability of the separation problem re-proves a result of
  Place, van Rooijen and Zeitoun \cite{place2013separating}. The algorithms for the word problem
  for $\omega$-terms are more efficient than those of Moura \cite{moura2011word}.

\end{itemize}

Separability of the join levels and the intersection levels is still open. We conjecture that these
problems can be solved with similar but more technical reductions.

\section{Preliminaries}

\paragraph*{Natural Numbers and Finite Words.}
Let $\Nat = \smallset{ 1, 2, \dots }$, $\Nat_0 = \smallset{ 0, 1, \dots }$ and $-\Nat =
\smallset{ -1, -2, \dots }$.
For the rest of this paper, we fix a finite alphabet $\Sigma$. By $\Sigma^*$, we denote the set of all finite
words over the alphabet $\Sigma$, including the empty word $\varepsilon$; $\Sigma^+$ denotes $\Sigma^*
\setminus \{ \varepsilon \}$.

\paragraph*{Order Types.}
A linearly ordered set $(P, \leq_P)$ consists of a (possibly infinite)
set $P$ and a linear ordering relation $\leq_P$ of $P$, i.\,e.\ a reflexive, anti-symmetric,
transitive and total binary relation ${\leq_P} \subseteq P \times P$. To simplify notation we
define two special objects $-\infty$ and $+\infty$. The former is always smaller with regard to
$\leq_P$ than any element in $P$ while the latter is always larger. We call two linearly ordered
sets $(P, \leq_P)$ and $(Q, \leq_Q)$ \emph{isomorphic} if there is an order-preserving bijection
$\varphi: P \to Q$. Isomorphism between linearly order sets is an equivalence relation; its classes
are called (linear) \emph{order types}.

The \emph{sum} of two linearly ordered sets $(P, \leq_P)$ and $(Q, \leq_Q)$ is $(P \uplus Q,
\leq_{P + Q})$ where $P \uplus Q$ is the disjoint union of $P$ and $Q$ and $\leq_{P + Q}$ orders
all elements of $P$ to be smaller than those of $Q$ while it behaves as $\leq_P$ and $\leq_Q$ on
elements from their respective sets. Similarly, the \emph{product} of $(P, \leq_P)$ and $(Q,
\leq_Q)$ is $(P \times Q, \leq_{P * Q})$ where $(p, q) \leq_{P * Q} (\tilde{p}, \tilde{q})$ holds
if and only if either $q \leq_Q \tilde{q}$ and $q \neq \tilde{q}$ or $q = \tilde{q}$ and $p \leq_P
\tilde{p}$ holds. Sum and product of linearly ordered sets are compatible with taking the order type.
This allows for writing $\mu + \nu$ and $\mu * \nu$ for order types $\mu$ and $\nu$.

We re-use $n \in \Nat_0$ to denote the order type of $(\{ 1, 2, \dots, n \}, \leq)$. One should
note that this use of natural numbers to denote order types does not result in contradictions
with sums and products: the usual calculation rules apply. Besides finite linear order types, we
need $\omega$, the order type of $(\Nat, \leq)$, and its dual $\omega^*$ the order type of
$(-\Nat, \leq)$. Another important order type in the scope of this paper is $\omega + \omega^*$,
whose underlying set is $\Nat \uplus (-\Nat)$. Note that, here, natural numbers and the (strictly)
negative
numbers are ordered as $1, 2, 3, \ldots, \enskip \ldots, -3, -2, -1$; therefore, in this order type,
we have for example $-1 \geq_{\omega + \omega^*} 1$.

\paragraph*{Generalized Words.}
Any finite word $w = a_1 a_2 \dots a_n$ of
length $n \in \Nat_0$ with $a_i \in \Sigma$ can be seen as a function which maps a \emph{position}
$i \in \{ 1, 2, \dots, n \}$ to the corresponding letter $a_i$ (or, possibly, the empty map). Therefore, it is natural to denote the positions in a word $w$ by $\dom{w}$. By relaxing the
requirement of $\dom{w}$ to be finite, one obtains the notion of \emph{generalized words}: a
(generalized) word $w$ over the alphabet $\Sigma$ of order type $\mu$ is a function $w: \dom{w}
\to \Sigma$, where $\dom{w}$ is a linearly ordered set in $\mu$. For $\dom{w}$, we usually choose
$(\mathbb{N}, \leq)$, $(-\mathbb{N}, \leq)$ and $(\Nat \uplus (-\Nat), \leq_{\omega + \omega^*})$ as
representative of $\omega$, $\omega^*$ and $\omega + \omega^*$, respectively. The order type of a
finite word of length $n$ is $n$.

Like finite words, generalized words can be concatenated, i.\,e.\ we write $u$ to the left of $v$
and obtain $uv$. In that case, the order type of $uv$ is the sum of the order types of $u$ and
$v$. Beside concatenation, we can also take powers of generalized words. Let $w$ be a generalized
word of order type $\mu$ which belongs to $(P_\mu, \leq_{\mu})$ and let $\nu$ be an arbitrary
order type belonging to $(P_\nu, \leq_{\nu})$. Then, $w^\nu$ is a generalized word of order type
$\mu * \nu$ which determines the ordering of its letters; $w$ maps $(p_1, p_2) \in P_{\mu}
\times P_{\nu}$ to $w(p_1)$. If $\nu = n$ for some $n \in \Nat$, then $w^\nu = w^n$ is equal
to the $n$-fold concatenation of $w$.

In this paper, the term \emph{word} refers to a generalized word. If it is important
for a word to be finite, it is referred to explicitly as a \emph{finite word}. As a counterpart to
the positions $\dom{w}$ in $w$, we define the set of \emph{letters appearing in a word}
$\alphabet(w)$ as the image of $w$ seen as a function. For example, for a finite word $w = a_1 a_2
\dots a_n$ of length $n \in \Nat_0$ (with $a_1, a_2, \dots, a_n \in \Sigma$), we have $\alphabet(w) =
\{ a_1, a_2, \dots, a_n \}$.

We also introduce notation for factors of words. For a pair $(l, r) \in ( \smallset{ -\infty }
\uplus \dom{w} ) \times ( \dom{w} \uplus \smallset{ +\infty } )$, define $w_{(l,
r)}$ as the restriction of the word $w$ (seen as a mapping) to the set of positions (strictly) larger
than $l$ and (strictly) smaller than $r$. Note that $w = w_{(-\infty, +\infty)}$ and $w_{(l, r)} =
\varepsilon$ for any pair $(l, r)$ with no position between $l$ and $r$.

\paragraph*{Monoids, Divisors, Congruences and Recognition.}
In this paper, the term \emph{mon\-oid} refers to a finite monoid (except when stated otherwise).
It is well known that, for any
monoid $M$, there is a smallest number $n \in \Nat$ such that $m^n$ is idempotent (i.\,e.\ $m^{2n} = m^n$) for every element
$m \in M$; this number is called the \emph{exponent} of $M$ and shall be denoted by $M! =
n$.\footnote{Note that all statements remain valid if one assumes that $M!$ is used to denote
$|M|!$.} A monoid $N$ is a \emph{divisor} of (another) monoid $M$, written as $N \prec M$, if $N$
is a homomorphic image of a submonoid of $M$.

A \emph{congruence} (relation) in a (not necessarily finite) monoid $M$ is an
equivalence relation $\mathcal{C} \subseteq M \times M$ such that $x_1 \mathrel{\mathcal{C}} x_2$
and $y_1 \mathrel{\mathcal{C}} y_2$ implies $x_1 y_1 \mathrel{\mathcal{C}} x_2 y_2$ for all $x_1,
x_2, y_1, y_2 \in M$. If $M$ is a (possibly infinite) monoid and $\mathcal{C} \subseteq M \times M$
is a congruence, then the set of equivalence classes of $\mathcal{C}$, denoted by $M/\mathcal{C}$,
is a well-defined monoid (which might still be infinite), whose size is called the \emph{index} of
$\mathcal{C}$. For any two congruences $\mathcal{C}_1$ and $\mathcal{C}_2$, one
can define their \emph{join} $\mathcal{C}_1 \vee \mathcal{C}_2$ as the smallest congruence which
includes $\mathcal{C}_1$ and $\mathcal{C}_2$; its index is at most as large as the index of
$\mathcal{C}_1$ and the index of $\mathcal{C}_2$.

A (possibly infinite) monoid $M$ \emph{recognizes} a language (of finite words) $L \subseteq \Sigma^*$
if there is a homomorphism $\varphi: \Sigma^* \to M$ with $L = \varphi^{-1} \left( \varphi(L)
\right)$. A language is \emph{regular} if and only if it is recognized by a finite monoid. It is
well known that there is a unique smallest monoid which recognizes a given regular language: the
\emph{syntactic monoid}.

\paragraph*{Green's Relations.}
Among the most important tools for studying monoids are Green's Relations. Let $x$ and $y$ be
elements of a monoid $M$. Define
\begin{align*}
  x \R y &\iff xM = yM,\\
  x \L y &\iff Mx = My \qquad \text{and}\\
  x \J y &\iff MxM = MyM
\end{align*}
where $xM = \set{ x m }{m \in M}$ is the \emph{right-ideal} of $x$, $Mx = \set{ m x }{m \in M}$
its \emph{left-ideal} and $MxM = \set{ m_1 x m_2 }{m_1, m_2 \in M}$ its (two-sided) \emph{ideal}.

By simple calculation, one can see that $x \R y$ holds if and only if there are $z, z' \in M$
such that $xz = y$ and $yz' = x$ and, symmetrically, that $x \L y$ holds if and only if there
are $z, z' \in M$ such that $zx = y$ and $z'y = x$.

\paragraph*{Varieties, $\pi$-Terms, Equations and Word Problem for $\pi$-terms.}
A \emph{variety} (of finite monoids) -- sometimes
also referred to as a \emph{pseudo-variety} -- is a class of monoids which is closed under
submonoids, homomorphic images and -- possibly empty -- finite direct products. For example, the class
$\V[R]$ of $\mathcal{R}$-trivial
monoids and the class $\V[L]$ of $\mathcal{L}$-trivial monoids both
form a variety, see e.g.~\cite{pin1986varieties}. Clearly, if $\V$ and $\V[W]$ are varieties, then so is $\V\cap \V[W]$. For example,
the class $\V[J]= \V[R] \cap \V[L]$ is a variety; in fact, it is the variety of all
$\mathcal{J}$-trivial monoids. For two varieties $\V$ and $\V[W]$, the smallest variety containing $\V \cup \V[W]$, the so called \emph{join}, is denoted by $\V \vee \V[W]$.

Many varieties can be defined in terms of \emph{equations} (or \emph{identities}). Because it will be useful later,
we take a more formal approach towards equations by using \emph{$\pi$-terms}\footnote{As mentioned in the introduction, $\pi$-terms are usually referred to as $\omega$-terms. In this paper, however, we use $\omega$ to denote the order type of the natural numbers. Therefore, we follow the approach of Perrin and Pin~\cite{perrin2004infinite} and use $\pi$ instead of $\omega$.}. A $\pi$-term is a finite word, built using letters, concatenation and an additional formal $\pi$-power
(and appropriate parentheses), whose $\pi$-exponents act as a placeholder for a substitution value.
Formally, every letter $a \in \Sigma$ is a $\pi$-term (over $\Sigma$). As a special case, also
$\varepsilon$ is a $\pi$-term. For two $\pi$-terms $\alpha$ and $\beta$, their concatenation $\alpha
\beta$ is a $\pi$-term as well and, if $\gamma$ is a $\pi$-term, then so is $(\gamma)^\pi$, where
$\pi$ is a formal exponent.

To state equations using $\pi$-terms, one needs to substitute the formal $\pi$-exponents by
actual values resulting in a word. We define $\subs{\gamma}{\mu}$ as the result of
substituting the $\pi$-exponents in $\gamma$ by an order type $\mu$, i.\,e.\ we have
$\subs{\varepsilon}{\mu} = \varepsilon$, $\subs{a}{\mu} = a$ for all $a \in \Sigma$,
$\subs{\alpha\beta}{\mu} = \subs{\alpha}{\mu} \subs{\beta}{\mu}$ and $\subs{(\gamma)^\pi}{\mu} =
\left( \subs{\gamma}{\mu} \right)^\mu$. For example\footnote{For a more elaborate example (involving
$\omega + \omega^*$) see \autoref{exmpl:positionInASubstPiTerm} on
\autopageref{exmpl:positionInASubstPiTerm}.}, we have
\[
  \subs{a\left( (b)^\pi c \right)^\pi}{3} = a \left( (bbb) c \right) \left( (bbb) c \right) \left( (bbb) c \right) \text{.}
\]

An equation $\alpha = \beta$
consists of two $\pi$-terms $\alpha$ and $\beta$ over the same alphabet $\Sigma$, which, here, can
be seen as a set of \emph{variables}. A homomorphism $\sigma: \Sigma^* \to M$ is called an
\emph{assignment of variables} in this context. An equation $\alpha = \beta$ \emph{holds} in a
monoid $M$ if for every assignment of variables $\sigma \left( \subs{\alpha}{M!} \right) =
\sigma \left( \subs{\beta}{M!} \right)$ is satisfied. It holds in a variety $\V$ if it holds
in all monoids in $\V$. The \emph{word problem for $\pi$-terms} over a variety $\V$ is the problem to decide
whether $\alpha = \beta$ holds in $\V$ for the input $\pi$-terms $\alpha$ and $\beta$.

\paragraph*{Mal'cev Products.}
Besides intersection and join, we need one more
constructions for varieties: the \emph{Mal'cev product}, which is often
defined using \emph{relational morphisms}. In this paper, we use a different, yet equivalent,
approach based on the congruences $\sim_K$ and $\sim_D$, see \cite{krohn1968homomorphisms} or
\cite[Corollary~4.3]{hall1999radical}. For their definition, let $x$ and $y$ be elements of a monoid
$M$ and define
\begin{align*}
  x \sim_K y &\iff \forall e \in E(S): ex \R e \text{ or } ey \R e \implies ex = ey\\
    \text{and }
  x \sim_D y &\iff \forall e \in E(S): xe \L e \text{ or } ye \L e \implies xe = ye\text{,}
\end{align*}
where $E(S)$ denotes the set of idempotents in $S$.

Obviously, $\sim_K$ and $\sim_D$ are of finite index in any (finite) monoid $M$. Thus, we have that
$M/{\sim_K}$ and $M/{\sim_D}$ are (finite) monoids and can define Mal'cev products of varieties.
Let $\V$ be a variety. The varieties $\V[K] \malcev \V$ and $\V[D] \malcev \V$ are defined by
\begin{align*}
  M \in \V[K] \malcev \V &\iff M/{\sim_K} \in \V \text{ and} \\
  M \in \V[D] \malcev \V &\iff M/{\sim_D} \in \V \text{,}
\end{align*}
where $M$ is a monoid. Note that, indeed, $\V[K] {\kern-0.1em} \malcev \V$ and $\V[D] \malcev \V$
are varieties for any variety $\V$ and that, furthermore, we have $\V \subseteq \V[K] \malcev \V$
and $\V \subseteq \V[D] \malcev \V$.

\section{The Trotter-Weil Hierarchy}

The main object of study in this paper, is the so-called \emph{Trotter-Weil Hierarchy}. We will
approach it primarily using certain combinatorial congruence. We will show that this approach is
equivalent to the more common algebraic approach, which we will use as the definition of the
hierarchy.

\paragraph*{The Trotter-Weil Hierarchy.}
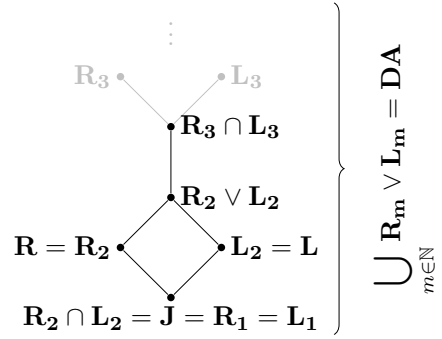
\begin{wrapfigure}{R}{0.4\textwidth}
  \vspace*{-2\baselineskip}
  \centering\resizebox{!}{0.3\textwidth}{%
    \begin{tikzpicture}[remember picture]
      \tikzstyle{every node}=[circle, fill, minimum size=3pt, inner sep=0pt];
      \tikzstyle{every label}=[shape=rectangle, draw=none, fill=none, outer sep=2pt];

      \node (J) [label={[name=Jlabel]below:{$\Rm[2] \cap \Lm[2] = \mathbf{J} = \Rm[1] = \Lm[1]$}}] {};
      \node (R1) [above left of=J, label=left:{$\mathbf{R} = \Rm[2]$}] {};
      \node (L1) [above right of=J, label={[name=L1label]right:{$\Lm[2] = \mathbf{L}$}}] {};
      \node (RvL1) [above right of=R1, label=right:{$\Rm[2] \vee \Lm[2]$}] {};
      \node (RcL2) [above of=RvL1, label=right:{$\Rm[3] \cap \Lm[3]$}] {};
      \node (R2) [opacity=0.25, above left of=RcL2, label={[opacity=0.25]left:{$\Rm[3]$}}] {};
      \node (L2) [opacity=0.25, above right of=RcL2, label={[opacity=0.25]right:{$\Lm[3]$}}] {};
      \node (dots) [opacity=0.25, above right of=R2, fill=none, rectangle] {$\vdots$};
      \draw (J) -- (R1)
            (J) -- (L1)
            (R1) -- (RvL1)
            (L1) -- (RvL1)
            (RvL1) -- (RcL2);
      \draw[draw opacity=0.25] (RcL2) -- (R2)
                               (RcL2) -- (L2);

      \draw [decorate, decoration={brace, amplitude=4pt}, xshift=10pt]
        ($(current bounding box.north east) + (6pt, 2pt)$) --
        ($(current bounding box.south east) + (0, -2pt)$)
        node [midway, anchor=mid, right, outer sep=6pt, fill=none, rectangle]
          {\rotatebox{90}{$\displaystyle\bigcup_{m \in \mathbb{N}}^{\phantom{m \in \mathbb{N}}} \Rm \vee \Lm = {\DA}$}};
    \end{tikzpicture}}\vspace*{-.5\baselineskip}
  \caption{\label{fig:trotterWeil}Trotter-Weil Hierarchy}\vspace*{-1.75\baselineskip}
\end{wrapfigure}
As the name implies, this hierarchy was first studied by Trotter and Weil \cite{trotter1997lattice}.
We will define it using Mal'cev products. Though this approach is different to the original one
used by Trotter and Weil, both are equivalent \cite{kufleitner2012logical}. We define:
\begin{align*}
  \Rm[1] &= \Lm[1] = \V[J],\\
  \Rm[m + 1] &= \V[K] \malcev \Lm \text{ and }\\
  \Lm[m + 1] &= \V[D] \malcev \Rm \text{.}
\end{align*}
These varieties form the so-called \emph{corners} of the hierarchy. Additionally, it contains the
\emph{join levels} $\Rm \vee \Lm$ and the \emph{intersection levels} $\Rm \cap \Lm$.
The term \enquote{hierarchy} is justified by the following inclusions: we have $\Rm \cap \Lm
\subseteq \Rm, \Lm \subseteq \Rm \vee \Lm$ and $\Rm \vee \Lm \subseteq \Rm[m + 1] \cap \Lm[m + 1]$;
the latter can be seen by induction.

Among the corners of the Trotter-Weil Hierarchy are some well known varieties: we have $\Rm[1] =
\Lm[1] = \V[J]$, $\Rm[2] = \V[R]$ and $\Lm[2] = \V[L]$ (for the last two, see
\cite{pin1986varieties}; the others are straightforward).

By taking the union of all varieties in the hierarchy, one gets the variety $\DA$
\cite{kufleitner2010lattice}, which is usually defined as the class of monoids whose regular
$\mathcal{D}$-classes form aperiodic semigroups\footnote{%
In finite monoids, $\mathcal{D}$-classes coincide with $\mathcal{J}$-classes; a $\mathcal{D}$-class
is called \emph{regular} if it contains an idempotent. A semigroup is called \emph{aperiodic} (or
\emph{group-free}) if it has no divisor which is a nontrivial group.}:
\begin{fact}\label{fct:trotterWeilIsDA}
  \hfill $\displaystyle \DA = \bigcup_{m \in \Nat} \Rm \vee \Lm = \bigcup_{m \in \Nat} \Rm =
    \bigcup_{m \in \Nat} \Lm$ \hfill\vspace*{0pt}
\end{fact}
These considerations yield the graphic representation given in \autoref{fig:trotterWeil}.

\paragraph*{Connections to Two-Variable Logic.}
The variety $\DA$ is closely connected to two-variable first-order logic. By $\FO^2[<]$, denote the
set of all first-order sentences over finite words which may only use the $<$ predicate (and
equality) and no more than two variables. A language $L \subseteq \Sigma^*$ of finite words is
definable by a sentence $\varphi \in \FO^2[<]$ if and only if its syntactic monoid is in $\DA$
\cite{therien1998over}, which it is if and only if it is in one of the Trotter-Weil Hierarchy's
varieties.

The intersection levels corresponds to the quantifier alternation hierarchy within $\FO^2[<]$
\cite{kufleitner2012alternation}: a first-order sentence using at most two variables belongs to
$\FO_m^2[<]$ if, on any path in its syntax tree, there is no quantifier after the first negation and
there are at most $m$ blocks of quantifiers. A language is definable by a sentence in $\FO_m^2[<]$
if and only if its syntactic monoid is in $\Rm[m + 1] \cap \Lm[m + 1]$.

\paragraph*{Equational Characterization.}
Besides its definition using Mal'cev products and its connections to logic, the Trotter-Weil
Hierarchy can also be characterized in terms of equations. For our proofs, we only need the
direction of this characterization stated in the lemma below. The other direction does hold as
well; we will see later on in \autoref{lem:TrotterWeilEquationsConverse} that it is an easy
consequence of the hierarchy's combinatorial characterization which we state below.

\begin{lemma}\label{lem:TrotterWeilEquations}
  Define the $\pi$-terms
  \[
    U_1 = (s x_1)^\pi s (y_1 t)^\pi \quad \text{and} \quad V_1 = (s x_1)^\pi t (y_1 t)^\pi
  \]
  over the alphabet $\Sigma_1 = \{ s, t, x_1 \}$. For $m \in \Nat$, let $x_{m + 1}$ and $y_{m + 1}$
  be new characters not in the alphabet $\Sigma_m$ and define the $\pi$-terms
  \[
    U_{m + 1} = (U_m x_{m + 1})^\pi U_m (y_{m + 1} U_m)^\pi \quad \text{and} \quad
    V_{m + 1} = (U_m x_{m + 1})^\pi V_m (y_{m + 1} U_m)^\pi
  \]
  over the alphabet $\Sigma_{m + 1} = \Sigma_m \uplus \{ x_{m + 1}, y_{m + 1} \}$.

  Then we have
  \begin{align*}
    M \in \Rm[1] = \Lm[1] = \V[J] &\impliedby U_1 = V_1 \text{ holds in } M \text{,}\\
    M \in \Rm[m + 1] &\impliedby (U_m x_{m + 1})^\pi U_m = (U_m x_{m + 1})^\pi V_m \text{ holds in } M
      \text{,}\\
    M \in \Lm[m + 1] &\impliedby U_m (y_{m + 1} U_m)^\pi = V_m (y_{m + 1} U_m)^\pi \text{ holds in } M
      \text{ and}\\
    M \in \Rm[m + 1] \cap \Lm[m + 1] &\impliedby U_m = V_m \text{ holds in } M
  \end{align*}
  for all $m \in \Nat$.
\end{lemma}
\begin{proof}
  The first implication is a well-known characterization of $\V[J]$ \cite{pin1986varieties}. We show
  the next two implications by induction over $m$ (see also \cite{kufleitner2012join}).

  First, assume $m = 1$ and consider a monoid $M$ in which $(U_1 x_2)^\pi U_1 = (U_1 x_2)^\pi V_1$
  holds. We need to show $M \in \Rm[2]$, which is equivalent to showing $M/{\sim_K} \in \Lm[1] =
  \V[J]$. For this, we show $u_1 = \sigma \left( \subs{U_1}{M!} \right) \sim_K
  \sigma \left( \subs{V_1}{M!} \right) = v_1$ for an arbitrary assignment of variables
  $\sigma: \Sigma_1^* \to M$. Afterwards, we are done by the first implication since the exponent of
  $M/{\sim_K}$ is a divisor of $M!$. Let $e$ be an arbitrary idempotent of $M$ and assume
  $e u_1 \R e$ (the other case from the definition of $\sim_K$ is symmetrical). Thus, there is an
  element $x_2 \in M$ such that $e u_1 x_2 = e$ holds. We can extend $\sigma$ by mapping the letter
  $x_2 \in \Sigma_2$ to the just defined monoid element $x_2$. We then have
  \begin{align*}
    e u_1 = e (u_1 x_2) u_1 = \dots &= e (u_1 x_2)^{M!} u_1\\
    &= e (u_1 x_2)^{M!} v_1 = \dots = e (u_1 x_2) v_1 = e v_1 \text{.}
  \end{align*}
  The equality in the middle holds because $(U_1 x_2)^\pi U_1 = (U_1 x_2)^\pi V_1$ holds in $M$ by
  assumption. This concludes the $m = 1$ case because the third implication is symmetrical.

  Now, assume $m > 1$ and consider a monoid $M$ in which $(U_m x_{m + 1})^\pi U_m = (U_m x_{m + 1})^\pi V_m$
  holds. We need to show $M \in \Rm[m + 1]$ and we do this by showing $M/{\sim_K} \in \Lm$. By
  induction, we only have to show that $U_{m - 1} (y_{m} U_{m - 1})^\pi = V_{m - 1} (y_{m} U_{m - 1})^\pi$
  holds in $M/{\sim_K}$. As before, let $\sigma: \left( \Sigma_{m - 1} \uplus \{ y_m \} \right)^* \to M$
  be an arbitrary assignment of variables. For convenience, let $u_{m - 1} =
  \sigma \left( \subs{U_{m - 1}}{M!} \right)$ and $v_{m - 1} =
  \sigma \left( \subs{V_{m - 1}}{M!} \right)$ and identify $y_m$ with $\sigma(y_m)$. Let $e$ be an
  arbitrary idempotent in $M$ such that $e u_{m - 1} (y_m u_{m - 1})^{M!} \R e$ holds. Clearly, this
  implies $e u_{m - 1} \R e$ and there are elements $x_{m}, x_{m + 1} \in M$ such that $e =
  e u_{m - 1} x_m$ and $e = e u_{m - 1} (y_m u_{m - 1})^{M!} x_{m + 1}$ holds. Extend $\sigma$ to
  map the letters $x_m$ and $x_{m + 1}$ to the respective monoid elements. Then, we have
  \begin{align*}
    e u_{m - 1} (y_m u_{m - 1})^{M!} &= e (u_{m - 1} x_m) u_{m - 1} (y_m u_{m - 1})^{M!}\\
    &= \dots = e (u_{m - 1} x_m)^{M!} u_{m - 1} (y_m u_{m - 1})^{M!} = e u_m
  \end{align*}
  with $u_m = \sigma \left( \subs{U_m}{M!} \right)$. This yields
  \[
    e = \underbrace{e u_{m - 1} (y_m u_{m - 1})^{M!}}_{e u_m} x_{m + 1} =
      e \left( u_m x_{m + 1} \right) = \dots = e \left( u_m x_{m + 1} \right)^{M!} \text{.}
  \]
  Using the fact that $(U_m x_{m + 1})^\pi U_m = (U_m x_{m + 1})^\pi V_m$ holds in $M$ (by assumption), we get
  \begin{align*}
    e u_m = e \left( u_m x_{m + 1} \right)^{M!} u_m = e \left( u_m x_{m + 1} \right)^{M!} v_m = e v_m
  \end{align*}
  where $v_m = \sigma \left( \subs{V_m}{M!} \right)$. In combination, we have
  \begin{align*}
    e u_{m - 1} (y_m u_{m - 1})^{M!} &= e u_m = e v_m = e (u_{m - 1} x_m)^{M!} v_{m - 1} (y_m u_{m - 1})^{M!}
  \end{align*}
  where the last equality holds due to the definition of $V_m$. Finally, we get
  \begin{align*}
    e u_{m - 1} (y_m u_{m - 1})^{M!} &= e (u_{m - 1} x_m)^{M!} v_{m - 1} (y_m u_{m - 1})^{M!}\\
    &= \dots = e (u_{m - 1} x_m) v_{m - 1} (y_m u_{m - 1})^{M!}\\
    &= e v_{m - 1} (y_m u_{m - 1})^{M!}
  \end{align*}
  Thus, we have shown $u_{m - 1} (y_m u_{m - 1})^{M!} \sim_K v_{m - 1} (y_m u_{m - 1})^{M!}$
  (the other case from the definition of $\sim_K$ is symmetrical) and are done. The implication for
  $\Lm[m + 1]$ in the case $m > 1$ follows by symmetry again.

  Finally, for the intersection levels, suppose that $U_m = V_m$
  holds in a monoid $M$. By the identities for the corners, we directly have $M \in \Rm[m + 1] \cap
  \Lm[m + 1]$.
\end{proof}

Besides the equational characterization of the individual varieties in the hierarchy, one can also
characterize their union $\DA$ in terms of an equation:
\begin{fact}\label{fct:DAByEquation}
  Let $M$ be a monoid. Then, we have
  \[
  M \in \DA \iff (xyz)^\pi y (xyz)^\pi = (xyz)^\pi \text{ holds in $M$.}
  \]
\end{fact}
\noindent A proof of this fact can be found in \cite{tesson2002diamonds}.

\section{Relations for the Trotter-Weil Hierarchy}

Although we define the Trotter-Weil Hierarchy algebraically using Mal'cev products, we will primarily
use a different characterization which is based on certain combinatorial congruences. Before we can
finally introduce these congruences, however, we need to give some definitions for factorizations of
words at the first or last $a$-position (i.\,e.\ an $a$-labeled position).

\paragraph*{Factorizations and accessible words.}
For a word $w$, a position
$p \in \dom{w} \uplus \{ -\infty \}$ and a letter $a \in \alphabet(w)$, let $X_a(w; p)$ denote
the first $a$-position (strictly) larger than $p$ (or the first $a$-position in $w$ if $p = -\infty$).
It is undefined if there is no such position (i.\,e.\ there is no $a$ to the right of $p$ in $w$) or if
the position is not well-defined. Define $Y_a(w; p)$ symmetrically as the first $a$-position from the
right which is (strictly) smaller than $p$. Notice that, with generalized words, the first $a$-position
to the right of a position $p$ is not necessarily well-defined even if there is an $a$-position larger
than $p$. For example, $a^{\omega^*} = \dots aa$ does not have a first $a$-position. We call words for
which this situation does not occur \emph{accessible}; i.\,e.\ a word $w$ is accessible if, for every
position $p$ in $w$ (including the special cases $p \in \{ \pm \infty \}$), $X_a(w; p)$ is defined if
and only if there is an $a$-position in $w_{(p, +\infty)}$ and $Y_a(w; p)$ is defined if and only if
there is an $a$-position in $w_{(-\infty, p)}$. Note that all finite words are accessible and that so
are all words of the form $\subs{\gamma}{\omega + \omega^*}$ for a $\pi$-term $\gamma$.

Let $w$ be an accessible word, define
\begin{align*}
  w \cdot X_a^L &= w_{(-\infty, X_a(w; -\infty))} \text{,} &
  w \cdot X_a^R &= w_{(X_a(w; -\infty), +\infty)} \text{,} \\
  w \cdot Y_a^L &= w_{(-\infty, Y_a(w; +\infty))} \text{ and } &
  w \cdot Y_a^R &= w_{(Y_a(w; +\infty), +\infty)}
\end{align*}
for all $a \in \alphabet(w)$. Additionally, define $C_{a, b}$ as a special form
of applying $X_a^L$ first and then $Y_b^R$ which is only defined if $X_a(w; -\infty)$ is strictly
larger than $Y_b(w; +\infty)$. For an example of $X_a^L$ and $X_a^R$ acting on a
word see \autoref{fig:XaExample}. Note that we have $w = (w \cdot X_a^L) a (w \cdot X_a^R) =
(w \cdot Y_a^L) a (w \cdot Y_a^R) = (w \cdot Y_b^L) b (w \cdot C_{a, b}) a (w \cdot X_a^R)$
(whenever these factors are defined). For example, we have $cbbcdcaca \cdot C_{a, b} = cdc$; see the
upper part of \autoref{fig:Cab} for a graphical representation of this example. If we apply a sequence
of factorizations, we omit the $\cdot$ between them, e.\,g.\ we write $w \cdot X_a^L Y_b^R = w \cdot
X_a^L \cdot Y_b^R = w \cdot C_{a, b}$.

\begin{figure}[h]
  \centering\resizebox{6cm}{!}{%
    \begin{tikzpicture}[text height=.8em, text depth=.2em]
      \node (w) {$w = {}$};
      \matrix [right=0cm of w, rectangle, draw, matrix of math nodes, ampersand replacement=\&, inner sep=2pt] (word) {
        a \& b \& b \& a \& b \& b \& b \& b \& b \& a \& b \& b \& b \& a \& b \& a \& b \& b \& a \& b \\
      };
      \node[above=0.5cm of word-1-4] (l) {$l$};
      \draw (l) edge[-latex, shorten >= 2pt] (word-1-4);
      \node[above=0.5cm of word-1-10] (Xa) {$X_a(w; l)$};
      \draw (Xa) edge[-latex, shorten >= 2pt] (word-1-10);
      \node[above=0.5cm of word-1-18] (r) {$r$};
      \draw (r) edge[-latex, shorten >= 2pt] (word-1-18);
      \path[decorate, decoration=brace, draw]
        ([yshift=-4pt]word-1-9.south east) -- node[below] {$w_{(l, r)} \cdot X_a^L$}
        ([yshift=-4pt]word-1-5.south west);
      \path[decorate, decoration=brace, draw]
        ([yshift=-4pt]word-1-17.south east) -- node[below] (label) {$w_{(l, r)} \cdot X_a^R$}
        ([yshift=-4pt]word-1-11.south west);
    \end{tikzpicture}}\vspace*{-.5\baselineskip}
    \caption{\label{fig:XaExample}Application of $X_a^L$ and $X_a^R$ to an example word.}\vspace*{-1.5\baselineskip}
\end{figure}

\paragraph*{Relations for the Trotter-Weil Hierarchy.}
With these definitions in place, we define for $m, n \in \Nat$ the relations $\equiv_{m, n}^X$, $\equiv_{m, n}^Y$ and
$\equiv_{m, n}^\WI$ of accessible words.\footnote{The presented relations could also be defined by (condensed)
\emph{rankers} (as it is done in \cite{kufleitner2012alternation} and \cite{kufleitner2012logical}).
Rankers were introduced by Weis and Immerman \cite{weis2009structure} (thus, the $\WI$ exponent in
$\equiv_{m, n}^\WI$) who reused the
\emph{turtle programs} by Schwentick, Thérien and Vollmer \cite{schwentick2002partially}. Another
concept related to condensed rankers is the \emph{unambiguous interval temporal logic} by Lodaya,
Pandya and Shah \cite{lodaya2008marking}.} The idea is that these relations hold on two words
$u$ and $v$ if both words allow for the same sequence of factorizations at the first or last
occurrence of a letter. For this, we use the following recursive definition.

\begin{definition}\label{def:relations}
  Let $m, n \in \Nat$ and let $u$ and $v$ be accessible words. Define recursively:
  \begin{enumerate}[leftmargin=*]
    \item $u \equiv_{0, 0}^Z v$, $u \equiv_{m, 0}^Z v$ and $u \equiv_{0, n}^Z v$ for
          $Z \in \smallset{X, Y, \WI}$ always hold.
    \item $\begin{aligned}[t]
            u \equiv_{m, n}^X v &\iff \alphabet(u) = \alphabet(v) \text{, }
            u \equiv_{m - 1, n - 1}^Y v \text{ and }\\
            &\phantom{{}\iff{}} \forall a \in \alphabet{(u)}:
              \begin{aligned}[t]
                &u \cdot X_a^L \equiv_{m - 1, n - 1}^Y v \cdot X_a^L \text{ and }\\
                &u \cdot X_a^R \equiv_{m, n - 1}^X v \cdot X_a^R
              \end{aligned}\\
            u \equiv_{m, n}^Y v &\iff \alphabet(u) = \alphabet(v) \text{, }
            u \equiv_{m - 1, n - 1}^X v \text{ and }\\
            &\phantom{{}\iff{}} \forall a \in \alphabet{(u)}:
              \begin{aligned}[t]
                &u \cdot Y_a^L \equiv_{m, n - 1}^Y v \cdot Y_a^L \text{ and }\\
                &u \cdot Y_a^R \equiv_{m - 1, n - 1}^X v \cdot Y_a^R
              \end{aligned}\\
            u \equiv_{m, n}^\WI v &\iff \alphabet(u) = \alphabet(v) \text{, }\\
            &\phantom{{}\iff{}} \forall a \in \alphabet{(u)}:
              \begin{aligned}[t]
                &u \cdot X_a^L \equiv_{m - 1, n - 1}^\WI v \cdot X_a^L \text{ and }\\
                &u \cdot X_a^R \equiv_{m, n - 1}^\WI v \cdot X_a^R \text{,}
              \end{aligned}\\
            &\phantom{{}\iff{}} \forall a \in \alphabet{(u)}:
              \begin{aligned}[t]
                &u \cdot Y_a^L \equiv_{m, n - 1}^\WI v \cdot Y_a^L \text{ and }\\
                &u \cdot Y_a^R \equiv_{m - 1, n - 1}^\WI v \cdot Y_a^R \text{ and}
              \end{aligned}\\
            &\phantom{{}\iff{}} \forall a, b \in \alphabet{(u)}:
              \parbox[t]{7cm}{%
                $u \cdot C_{a, b}$ and $v \cdot C_{a, b}$ are either both undefined or
                both defined and $u \cdot C_{a, b} \equiv_{m - 1, n - 1}^\WI
                  v \cdot C_{a, b}$ holds.}
          \end{aligned}$
  \end{enumerate}
  Additionally, define $u \equiv_{m, n}^{XY} v \iff u \equiv_{m, n}^X v \text{ and } u
  \equiv_{m, n}^Y v$ for all $m, n \in \Nat_0$.
\end{definition}

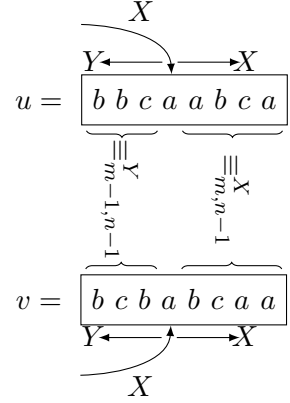
\begin{wrapfigure}{R}{0.3\textwidth}
  \vspace*{-\baselineskip}%
  \centering%
  \begin{tikzpicture}[text height=.8em, text depth=.2em]
    \node (u) {$u = {}$};
    \matrix [right=0cm of u, rectangle, draw, matrix of math nodes, ampersand replacement=\&, inner sep=2pt] (uword) {
      b \& b \& c \& a \& a \& b \& c \& a \\
    };
    \node[below=2cm of u] (v) {$v = {}$};
    \matrix [right=0cm of v, rectangle, draw, matrix of math nodes, ampersand replacement=\&, inner sep=2pt] (vword) {
      b \& c \& b \& a \& b \& c \& a \& a \\
    };
    \draw (uword.west)+(0, 1cm) edge[-latex, shorten >= 2pt, bend left, in=125] node[above] {$X$} (uword-1-4);
    \draw (vword.west)+(0, -1cm) edge[-latex, shorten >= 2pt, bend right, in=-125] node[below] {$X$} (vword-1-4);
    \draw ($(uword-1-4)+(-0.1cm, 0.5cm)$) edge[-latex, shorten >= 2pt] node[pos=1.0] {$Y$} ($(uword-1-4) + (-1cm, 0.5cm)$);
    \draw ($(uword-1-4)+(+0.1cm, 0.5cm)$) edge[-latex, shorten >= 2pt] node[pos=1.0] {$X$} ($(uword-1-4) + (+1cm, 0.5cm)$);
    \draw ($(vword-1-4)+(-0.1cm, -0.5cm)$) edge[-latex, shorten >= 2pt] node[pos=1.0] {$Y$} ($(vword-1-4) + (-1cm, -0.5cm)$);
    \draw ($(vword-1-4)+(+0.1cm, -0.5cm)$) edge[-latex, shorten >= 2pt] node[pos=1.0] {$X$} ($(vword-1-4) + (+1cm, -0.5cm)$);
    \path[decorate, decoration=brace, draw]
      ([yshift=-4pt]uword-1-3.south east) -- node[shape=coordinate, below, name=lefttop] {}
      ([yshift=-4pt]uword-1-1.south west);
    \path[decorate, decoration=brace, draw]
      ([yshift=-4pt]uword-1-8.south east) -- node[shape=coordinate, below, name=righttop] {}
      ([yshift=-4pt]uword-1-5.south west);
    \path[decorate, decoration=brace, draw]
      ([yshift=+4pt]vword-1-1.north west) -- node[shape=coordinate, above, name=leftbottom] {}
      ([yshift=+4pt]vword-1-3.north east);
    \path[decorate, decoration=brace, draw]
      ([yshift=+4pt]vword-1-5.north west) -- node[shape=coordinate, above, name=rightbottom] {}
      ([yshift=+4pt]vword-1-8.north east);
    \path (lefttop) -- node[pos=0.5, anchor=center, rotate=-90] {$\equiv_{m - 1, n - 1}^Y$} (leftbottom);
    \path (righttop) -- node[pos=0.5, anchor=center, rotate=-90] {$\equiv_{m, n - 1}^X$} (rightbottom);
  \end{tikzpicture}\vspace*{-\baselineskip}%
  \caption{\label{fig:relations}$\equiv_{m, n}^X$ illustrated.}\vspace*{-1.5\baselineskip}
\end{wrapfigure}

As this definition is vital for the understanding of the rest of this paper, we try to give an
intuitive understanding of how the relations work. In the parameter $m$, we remember the remaining
number of direction changes (which are
caused by an $X_a^L$ or $Y_a^R$ factorizations) in a factorization sequence and the parameter $n$
is the number of remaining factorization moves (independent of their direction). Thus, if $m$ or
$n$ is zero, then the relations shall be satisfied on all accessible word pairs. For $m$ and $n$
larger than zero, our first assertion is that both words have the same alphabet; otherwise, one of
them would admit a factorization at a letter while the other would not, as the letter is not in its
alphabet. For the letters of the common alphabet, we want to be able to perform further factorization
(until we reach $n = 0$). The $X$ or $Y$ exponent of the relation indicates whether we start
at the beginnings or at the ends of the two words; the $\WI$ exponent is a special case,
which we will discuss below. So, if we want $u \equiv_{m, n}^X v$ to hold and the alphabets of $u$ and
$v$ coincide, then we can continue factorizing at the first $a$ in $u$ and in $v$. In the
next step, we continue either in the two left parts or in the two right parts. If we continue in the
right parts (i.\,e.\ we do an $X_a^R$ factorization), then we require $u \cdot X_a^R \equiv_{m, n -
1}^X v \cdot X_a^R$. The $X$ in the exponent indicates that the last factorization position (the first
$a$ in $u$ and $v$, respectively) was at the beginning of the words. We have made one additional
factorization -- thus, we only have $n - 1$ remaining factorizations in this part -- but we did not
change the direction because we were at the beginnings of the words before and still are for the new
word pair -- thus, we still have $m$ such changes of direction. If, instead of taking the right parts
after the first $a$, we had taken the left parts, then we would still have made a single
factorization; thus we have $n - 1$ remaining factorizations. However, we would also have a direction
change: the factorization position (i.\,e.\ the first $a$) is to the right of the new words $u \cdot
X_a^L$ and $v \cdot X_a^L$. Therefore, we have $m - 1$ remaining direction changes and we change the
$X$ exponent into a $Y$ exponent. This is summed up by requiring
$u \cdot X_a^L \equiv_{m - 1, n - 1}^Y v \cdot X_a^L$. Both situations are illustrated in
\autoref{fig:relations}. Additionally, we also have the choice to switch from the beginning to the end
at the cost of one direction change and one factorization. This is reflected in the fact that we require
$u \equiv_{m - 1, n - 1}^Y v$ for $u \equiv_{m, n}^X v$.

\begin{wrapfigure}{l}{0.325\textwidth}
  \vspace{-\baselineskip}%
  \centering%
  \begin{tikzpicture}[text height=.8em, text depth=.2em]
    \node (u) {$u = {}$};
    \matrix [right=0cm of u, rectangle, draw, matrix of math nodes, ampersand replacement=\&, inner sep=2pt] (uword) {
      c \& b \& b \& c \& d \& c \& a \& c \& a \\
    };
    \node[below=2cm of u] (v) {$v = {}$};
    \matrix [right=0cm of v, rectangle, draw, matrix of math nodes, ampersand replacement=\&, inner sep=2pt] (vword) {
      b \& c \& b \& d \& c \& c \& a \& a \& c \\
    };
    \draw (uword.west)+(0, 1cm) edge[-latex, shorten >= 2pt, bend left, in=125] node[below, pos=0.1] {$X$} (uword-1-7);
    \draw (vword.west)+(0, -1cm) edge[-latex, shorten >= 2pt, bend right, in=-125] node[above, pos=0.1] {$X$} (vword-1-7);
    \draw (uword.east)+(0, 1cm) edge[-latex, shorten >= 2pt, bend right, in=-125] node[below, pos=0.1] {$Y$} (uword-1-3);
    \draw (vword.east)+(0, -1cm) edge[-latex, shorten >= 2pt, bend left, in=125] node[above, pos=0.1] {$Y$} (vword-1-3);
    \path[decorate, decoration=brace, draw]
      ([yshift=-4pt]uword-1-6.south east) -- node[shape=coordinate, below, name=lefttop] {}
      ([yshift=-4pt]uword-1-4.south west);
    \path[decorate, decoration=brace, draw]
      ([yshift=+4pt]vword-1-4.north west) -- node[shape=coordinate, above, name=leftbottom] {}
      ([yshift=+4pt]vword-1-6.north east);
    \path (lefttop) -- node[pos=0.5, anchor=center, rotate=-90] {$\equiv_{m - 1, n - 1}^\WI$} (leftbottom);
  \end{tikzpicture}\vspace*{-.5\baselineskip} %
  \caption{\label{fig:Cab}$\equiv_{m, n}^\WI$ illustrated.}\vspace*{-\baselineskip}
\end{wrapfigure}
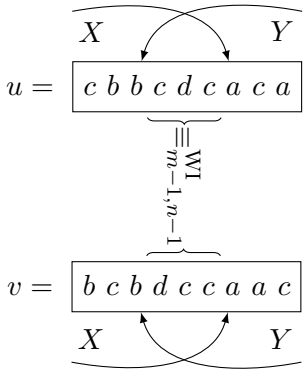
The relation $\equiv_{m, n}^Y$ works symmetrically to $\equiv_{m, n}^X$. However, $\equiv_{m, n}^\WI$
is a bit different. Here, we still lose one direction change after an $X_a^L$ or $Y_a^R$ factorization
while the number keeps the same for $X_a^R$ and $Y_a^L$ factorizations but we also have the special
$C_{a, b}$ factorizations, which start simultaneously at the beginning and the end. Consider the situation
given in \autoref{fig:Cab}. We go to the first $a$ and the last $b$; the former is to the right of the
latter in both words (or the $\equiv_{m, n}^\WI$ relation does not hold). For the pair of parts in the
middle, we require $u \cdot C_{a, b} \equiv_{m - 1, n - 1}^\WI v \cdot C_{a, b}$, i.\,e.\ we count this
as a single factorization move which loses a single direction change.
Finally, there is $\equiv_{m, n}^{XY}$ where we can start factorizing at the beginning or at the end.

By simple inductions, one can see that the relations are congruences of finite index over $\Sigma^*$.
Also note that $u \equiv_{m, n}^Z v$ implies $u \equiv_{m, k}^Z v$ and, if $m > 0$, also
$u \equiv_{m - 1, k}^Z v$ for all $k \leq n$ and $Z \in \smallset{X, Y, XY, \WI}$.

The importance of the just defined relations for this paper yields from their connection to the
Trotter-Weil Hierarchy (and, thus, to $\DA$), which we state in the next theorem.
\begin{theorem}\label{thm:TWcharacterization}
  Let $M$ be a monoid and $m \in \Nat$. Then\
  \begin{itemize}[itemsep=4.5pt plus 2.0pt minus 2.5pt]\enlargethispage*{1\baselineskip}
    \item $M \prec \Sigma^* / {\equiv_{m, n}^X} \text{ for some } n \in \Nat_0 \iff M \in \Rm$,
    \item $M \prec \Sigma^* / {\equiv_{m, n}^Y} \text{ for some } n \in \Nat_0 \iff M \in \Lm$,
    \item $M \prec \Sigma^* / {\equiv_{m, n}^{XY}} \text{ for some } n \in \Nat_0 \iff M \in \Rm \vee \Lm$ and
    \item $M \prec \Sigma^* / {\equiv_{m, n}^\WI} \text{ for some } n \in \Nat_0 \iff M \in \Rm \cap \Lm$ hold.
  \end{itemize}
\end{theorem}

The combinatorial nature of the relations will turn out to be useful in the remainder of this
paper since it allow us, for example, to obtain efficient algorithms for the problems we consider.
Unfortunately, on the other hand, it makes proving \autoref{thm:TWcharacterization} quite technical.
In fact, we will dedicate the next section to this proof.

\section{A Proof for \autoref{thm:TWcharacterization}}

We will prove both directions of \autoref{thm:TWcharacterization} individually. Before we can do this, however, we need to introduce some more concepts.

\paragraph{$\mathcal{R}$- and $\mathcal{L}$-factorizations.}
Let $\varphi: \Sigma^* \to M$ be a (monoid) homomorphism into a monoid $M$. The
\emph{$\mathcal{R}$-factorization} of a finite word $w$ is the (unique) factorization $w = w_0 a_1 w_1 a_2
w_2 \allowbreak \dots a_k w_k$ with $w_0, w_1 \dots, w_k \in \Sigma^*$ and $a_1, a_2, \dots, a_k \in \Sigma$
such that, on the one hand,
\begin{align*}
  \varphi(\varepsilon) &\R \varphi(w_0) \text{ and}\\
  \varphi(w_0 a_1 w_1 a_2 w_2 \dots a_i) &\R \varphi(w_0 a_1 w_1 a_2 w_2 \dots a_i w_i)
\end{align*}
hold for $i = 1, 2, \dots, k$ and, on the other hand,
\begin{gather*}
  \varphi(w_0 a_1 w_1 a_2 w_2 \dots a_i w_i) \not\R \varphi(w_0 a_1 w_1 a_2 w_2 \dots a_i w_i a_{i + 1})
\end{gather*}
holds for $i = 0, 1, \dots, k - 1$. Symmetrically, the \emph{$\mathcal{L}$-factorization} of $w$ is
the factorization $w = w_0 a_1 w_1 a_2 w_2 \dots a_k w_k$ with $w_0, w_1 \dots, w_k \in \Sigma^*$
and $a_1, a_2, \allowbreak \dots, \allowbreak a_k \in \Sigma$ such that, on the one hand,
\begin{align*}
  \varphi(w_k) &\L \varphi(\varepsilon) \text{ and}\\
  \varphi(w_{i - 1} a_i w_i a_{i + 1} w_{i + 1} \dots a_k w_k) &\L \varphi(a_i w_i a_{i + 1} w_{i + 1} \dots a_k w_k)
\end{align*}
hold for $i = 1, 2, \dots, k$ and, on the other hand,
\begin{gather*}
  \varphi(a_i w_i a_{i + 1} w_{i + 1} a_{i + 2} w_{i + 2} \dots a_k w_k) \not\L
    \varphi(w_i a_{i + 1} w_{i + 1} a_{i + 2} w_{i + 2} \dots a_k w_k)
\end{gather*}
holds for $i = 1, 2, \dots, k$.

\paragraph*{$\DA$ and $\mathcal{R}$-classes.}
In $\DA$, getting into a new $\mathcal{R}$-class is strictly coupled to an element's alphabet, as
the following lemma shows\footnote{The curious reader might be interested in the fact that the
lemma's assertion also holds for monoids in $\V[DS]$, the variety of monoids whose regular
$\mathcal{D}$-classes form (arbitrary, but finite) semigroups. More on $\V[DS]$ can, for example, be
found in \cite{alm94:short}.}, where $a$ can be seen as one of the
monoids generators (i.\,e.\ a letter in its alphabet).
\begin{lemma}\label{lem:RDescendInDA}
  Let $M \in \DA$ be a monoid and let $s, t \in M$ such that $s \R t$. Then
  \[
    s \R sa \implies t \R ta
  \]
  holds for all $a \in M$.
\end{lemma}
\begin{proof}
  Since we have $t \R s \R sa$, there are $x, y \in M$ with $s = tx$ and $t = say$. We then have
  \[
    t = t x a y = t (x a y)^2 = \dots = t (x a y)^{M!} \text{,}
  \]
  which yields
  \[
    ta (xay)^{M!} = t (xay)^{M!} a (xay)^{M!} = t (xay)^{M!} = t \text{.}
  \]
  using the equation from \autoref{fct:DAByEquation}. Thus, we have $ta \R t$.
\end{proof}
One of the main applications of the previous lemma is the following. If we have a monoid
$M \in \DA$, a homomorphism $\varphi: \Sigma^* \to M$ and the $\mathcal{R}$-factorization $w = w_0
a_1 w_1 a_2 w_2 \dots a_k w_k$ of a finite word $w \in \Sigma^*$, then we know that $a_i \not\in
\alphabet(w_{i - 1})$ for $i = 1, 2, \dots, k$. If we had $a_i \in \alphabet(w_{i - 1})$, we could
factorize $w_{i - 1} = u a_i v$ and would have
\[
  \varphi(w_0 a_1 w_1 a_2 w_2 \dots a_{i - 1} u) \R
    \varphi(w_0 a_1 w_1 a_2 w_2 \dots a_{i - 1} u a_i)
\]
and, by the previous lemma, also
\[
  \varphi(w_0 a_1 w_1 a_2 w_2 \allowbreak \dots \allowbreak a_{i - 1} u a_i v) \R
    \varphi(w_0 a_1 w_1 a_2 w_2 \dots a_{i - 1} u a_i v\allowbreak a_i) \text{,}
\]
which results in a contradiction to the
definition of $\mathcal{R}$-fact\-or\-iza\-tions. Of course, we can apply a left-right dual of the lemma
to get an analogue statement for $\mathcal{L}$-factorizations.

Another application of \autoref{lem:RDescendInDA} is stated in the following lemma.
\begin{lemma}\label{lem:RDescendInDAWords}
  Let $M \in \DA$ be a monoid, $s \in M$ and let $\varphi: \Sigma^* \to M$ be a homomorphism. Then,
  for all finite words $u, v \in \Sigma^*$ with $\alphabet(u) = \alphabet(v)$, we have
  \[
    s \R s \varphi(u) \implies s \R s \varphi(v) \text{.}
  \]
\end{lemma}
\begin{proof}
  The case $u = v = \varepsilon$ is trivial. Therefore, assume $s \R s \varphi(u)$ but $s \not\R
  s \varphi(v' a)$ for some prefix $v' a$ of $v$ with $a \in \Sigma$. Without loss of generality,
  let $v' a$ be the shortest such prefix; thus, we have $s \R s \varphi(v')$. Since we have
  $\alphabet(u) = \alphabet(v)$, the is a prefix $u' a$ of $u$. Because of $s \R s \varphi(u)$, we
  have $s \R s \varphi(u') \R s \varphi(u' a)$. Thus, we have $s \varphi(v') \R s \R s \varphi(u')
  \R s \varphi(u') \varphi(a)$. Now, \autoref{lem:RDescendInDA} yields $s \varphi(v') \R
  s \varphi(v') \varphi(a)$, which is a contradiction to our assumption.
\end{proof}

Now, we are prepared to prove the characterization of the Trotter-Weil Hierarchy stated in
\autoref{thm:TWcharacterization}. This is done in the following two theorems (see also
\cite{kufleitner2012logical} for the corners and \cite{kufleitner2012alternation} for the
intersection levels). We use the notations $X_\Sigma^D = \{ X_a^L, X_a^R \mid a \in \Sigma \}$,
$Y_\Sigma^D = \{ Y_a^L, Y_a^R \mid a \in \Sigma \}$ and some natural variations of it.

\begin{restatable}{theorem}{relationsImplyEqualityUnderHom}\label{thm:relationsImplyEqualityUnderHom}
  Let $M$ be a finite monoid, $\varphi: \Sigma^* \to M$ a homomorphism and $m \in \Nat$. Then:
  \begin{itemize}[itemsep=4.5pt plus 2.0pt minus 2.5pt]\enlargethispage{1.5\baselineskip}
    \item $M \in \Rm \implies \left( \exists n \in \Nat \, \forall u, v \in \Sigma^*:
      u \equiv_{m, n}^X v \implies \varphi(u) = \varphi(v) \right)$
    \item $M \in \Lm \implies \left( \exists n \in \Nat \, \forall u, v \in \Sigma^*:
      u \equiv_{m, n}^Y v \implies \varphi(u) = \varphi(v) \right)$
    \item $M \in \Rm \vee \Lm \implies \left( \exists n \in \Nat \, \forall u, v \in \Sigma^*:
      u \equiv_{m, n}^{XY} v \implies \varphi(u) = \varphi(v) \right)$
    \item $M \in \Rm[m + 1] \cap \Lm[m + 1] \implies \left( \exists n \in \Nat \, \forall u, v \in
      \Sigma^*:
      u \equiv_{m, n}^\WI v \implies \varphi(u) = \varphi(v) \right)$
  \end{itemize}
\end{restatable}
\begin{proof}
  We fix a homomorphism $\varphi: \Sigma^* \to M$ and proceed by induction over $m$. For $m = 1$, we
  have $\Rm[1] = \Lm[1] = \Rm[1] \vee \Lm[1] = \Rm[2] \cap \Lm[2] = \V[J]$. Thus, the assertion
  follows from a result of Simon \cite{sim75:short}. However, we also include a full proof for
  completeness. Let $M \in \V[J]$ and
  $n = |M|$, which is the number of $\mathcal{J}$-classes in $M$ (and equal to the number of
  $\mathcal{R}$-classes and the number of $\mathcal{L}$-classes). Assume that $u \equiv_{1, n}^X v$
  for two finite words $u, v \in \Sigma^*$ and let $u = u_0 a_1 u_1 a_2 u_2 \dots a_k u_k$ be the
  $\mathcal{R}$-factorization of $u$. We have $k + 1 \leq n$ and, because $M$ is
  $\mathcal{R}$-trivial, $\varphi(u_0) = \varphi(u_1) = \dots = \varphi(u_k) = 1$ or, more precisely,
  that no letter in $u_0 u_1 \dots u_k$ can be mapped by $\varphi$ to a value different from $1$.
  Notice that this implies $\varphi(u) = \varphi(a_1 a_2 \dots a_k)$. Furthermore, we have
  $\varphi(a_1), \varphi(a_2), \dots, \varphi(a_k) \neq 1$ due to the definition
  of an $\mathcal{R}$-factorization.

  By definition of $\equiv_{m, n}^X$, we have $a_1 \in \alphabet(v)$ and $u \cdot X_{a_1}^R =
  a_2 u_2 a_3 u_3 \dots a_k u_k \allowbreak\equiv_{1, n - 1}^X v \cdot X_{a_1}^R$. Therefore, we can find $a_2$ in
  $v \cdot X_{a_1}^R$ and have $u \cdot X_{a_1}^R \cdot X_{a_2}^R = a_3 u_3 a_4 u_4 \dots a_k u_k
  \equiv_{m, n - 2}^X v  \cdot X_{a_1}^R \cdot X_{a_2}^R$. Iterating this approach yields that $a_1
  a_2 \dots a_k$ is a subword\footnote{A finite word $c_1 c_2 \dots c_s$ with
  $c_i \in \Sigma$ is a subword of a (not necessarily finite) word $w$ if we can write
  $w = w_0 c_1 w_1 c_2 w_2 \dots c_s w_s$ for some words $w_0, w_1, \dots, w_s$.} of $v$. Now, let
  $v = v_0 b_1 v_1 b_2 v_2 \dots b_l v_l$ be the $\mathcal{R}$-factorization of $v$. By symmetry, we
  get $\varphi(b_1), \varphi(b_2), \dots, \varphi(b_l) \neq 1$ and that $b_1 b_2 \dots b_l$ is a
  subword of $u$. However, for no $j \in \{ 1, 2, \dots, l \}$, the letter $b_j$ can occur in $u_0
  u_1 \dots u_k$ since all letters in that word must be mapped to $1$. Thus, $b_1 b_2 \dots b_l$
  must in fact be a subword of $a_1 a_2 \dots a_k$. Again by symmetry, $a_1 a_2 \dots a_k$ must be
  a subword of $b_1 b_2 \dots b_l$ and, thus, the two words must be equal. This implies $\varphi(u)
  = \varphi(a_1 a_2 \dots a_k) = \varphi(b_1 b_2 \dots b_l) = \varphi(v)$ where the last equality
  follows from $\varphi(v_0) = \varphi(v_1) = \dots = \varphi(v_l) = 1$, which holds due to the
  $\mathcal{R}$-triviality of $M$.

  The argumentation for $u \equiv_{1, n}^Y v$ is symmetric
  using the $\mathcal{L}$-factorization, the case for $u \equiv_{1, n}^{XY} v$ follows trivially and
  the case for $u \equiv_{1, n}^\WI v$ uses the same argumentation.

  Now, let $M \in \Rm$ for an $m > 1$. This implies $M / {\sim_K} \in \Lm[m - 1]$ and there is an
  $n' \in \Nat$ such that $u' \equiv_{m - 1, n'}^Y v' \implies \varphi(u') \sim_K \varphi(v')$ holds
  for all $u', v' \in \Sigma^*$. Let $r$ be the number of $\mathcal{R}$-classes in $M$ and let
  $n = n' + r$. Consider the $\mathcal{R}$-factorization $u = u_0 a_1 u_1 a_2 u_2 \dots a_k u_k$ of
  a finite word $u \in \Sigma^*$; note that $k + 1 \leq r$ must hold. We have
  \begin{align*}
    u_i &= u \cdot X_{a_1}^R X_{a_2}^R \dots X_{a_{i - 1}}^R X_{a_{i}}^L \text{ for } i = 0, 1, \dots, k - 1 \text{ and}\\
    u_k &= u \cdot X_{a_1}^R X_{a_2}^R \dots X_{a_k}^R \text{.}
  \end{align*}
  For a second finite word $v \in \Sigma^*$ with $u \equiv_{m, n}^X v$, we know that
  $\alphabet(u) = \alphabet(v)$. Thus, we can apply $X_{a_1}^L$ and $X_{a_1}^R$ to $v$ and obtain
  \[
    v_0 = v \cdot X_{a_1}^L \quad \text{ and } \quad v' = v \cdot X_{a_1}^R \text{.}
  \]
  By definition of $\equiv_{m, n}^X$, we have $v_0 \equiv_{m - 1, n - 1}^Y u_0$ and
  $v' \equiv_{m, n - 1}^X u_1 a_2 u_2 a_3 u_3 \dots \allowbreak a_k \allowbreak u_k$. Because of $k \leq r < n$, we can
  apply the same argument on $v'$ and, by iteration, get
  \begin{align*}
    v_i &= v \cdot X_{a_1}^R X_{a_2}^R \dots X_{a_{i - 1}}^R X_{a_{i}}^L \text{ for } i = 0, 1, \dots, k - 1 \text{ and}\\
    v_k &= v \cdot X_{a_1}^R X_{a_2}^R \dots X_{a_k}^R
  \end{align*}
  with $u_i \equiv_{m - 1, n - i - 1}^Y v_i$ for $i = 0, 1, \dots, k - 1$ and
  $u_k \equiv_{m, n - k}^X v_k$. Because of $i \leq k \leq r - 1$, we have $n - i - 1 =
  n' + r - i - 1 \geq n' + r - (r - 1) - 1 = n'$ and $u_i \equiv_{m - 1, n'}^Y v_i$ for
  $i = 0, 1, \dots, k - 1$. For $u_k$ and $v_k$, we
  have $u_k \equiv_{m - 1, n - k -1}^Y v_k$ by the definition of the congruences and, therefore,
  $u_k \equiv_{m - 1, n'}^Y v_k$ because of $n - k - 1 \geq n - i - 1 \geq n'$.
  Summing this up, we have $u_i \equiv_{m - 1, n'}^Y v_i$ and, thus,
  $\varphi(u_i) \sim_K \varphi(v_i)$ for all $i = 0, 1, \dots, k$.

  Since we have defined $u_i$ by the $\mathcal{R}$-factorization of $u$, there is an $s_i \in M$ for
  any $i \in \{ 0, 1, \dots, k \}$ such that $\varphi(u_0 a_1 u_1 a_2 u_2 \dots a_i u_i) s_i =
  \varphi(u_0 a_1 u_1 a_2 u_2 \dots a_i)$ holds. For these, we have
  \[
    \left(\varphi(u_i) s_i \right)^{M!} \varphi(u_i) \R \left(\varphi(u_i) s_i \right)^{M!}
  \]
  because of $\left(\varphi(u_i) s_i \right)^{M!} \varphi(u_i) s_i \left(\varphi(u_i) s_i \right)^{M! - 1} =
  \left(\varphi(u_i) s_i \right)^{M!}$, which yields
  \[
    \left(\varphi(u_i) s_i \right)^{M!} \varphi(u_i) = \left(\varphi(u_i) s_i \right)^{M!} \varphi(v_i)
  \]
  by $\varphi(u_i) \sim_K \varphi(v_i)$. Thus, we have
  \begin{align*}
    \varphi(u_0 a_1 u_1 a_2 u_2 \dots a_k u_k)
    &= \varphi(u_0 a_1 u_1 a_2 u_2 \dots a_k u_k) \left( s_k \varphi(u_k) \right)^{M!}\\
    &= \varphi(u_0 a_1 u_1 a_2 u_2 \dots a_k) \left( \varphi(u_k) s_k \right)^{M!} \varphi(u_k)\\
    &= \varphi(u_0 a_1 u_1 a_2 u_2 \dots a_k) \left( \varphi(u_k) s_k \right)^{M!} \varphi(v_k)\\
    &= \varphi(u_0 a_1 u_1 a_2 u_2 \dots a_k) \varphi(v_k)\\
    &= \varphi(u_0 a_1 u_1 a_2 u_2 \dots a_{k - 1} u_{k - 1}) \varphi(a_k v_k)\\
    &= \varphi(u_0 a_1 u_1 a_2 u_2 \dots a_{k - 1}) \left( \varphi(u_{k - 1}) s_{k - 1} \right)^{M!} \varphi(u_{k - 1}) \varphi(a_k v_k)\\
    &= \varphi(u_0 a_1 u_1 a_2 u_2 \dots a_{k - 1}) \left( \varphi(u_{k - 1}) s_{k - 1} \right)^{M!} \varphi(v_{k - 1}) \varphi(a_k v_k)\\
    &= \varphi(u_0 a_1 u_1 a_2 u_2 \dots a_{k - 2} u_{k - 2}) \varphi(a_{k - 1} v_{k - 1} a_k v_k)\\
    &= \dots\\
    &= \varphi(v_0 a_1 v_1 a_2 v_2 \dots a_k v_k) \text{,}
  \end{align*}
  which concludes the proof for $\Rm$.

  The proof for $\Lm$ is symmetrical. For $\Rm \vee \Lm$, we observe that a monoid is in the join
  $\V \vee \V[W]$ of two varieties $\V$ and $\V[W]$ if and only if it is a divisor (i.\,e.\ the
  homomorphic image of a submonoid) of a direct product $M_1 \times M_2$ such that $M_1 \in \V$ and
  $M_2 \in \V[W]$ \cite[Exercise~1.1]{eilenberg1976automata}. Therefore, if we have a monoid
  $M \in \Rm \vee \Lm$, there are monoids $M_1 \in \Rm$ and $M_2 \in \Lm$ such that $M$ is a divisor
  of $M_1 \times M_2$; i.\,e.\ there is a submonoid $N$ of $M_1 \times M_2$ and a surjective monoid
  homomorphism $\psi: N \onto M$. For every $a \in \Sigma$, we can find elements $m_{a, 1} \in M_1$
  and $m_{a, 2} \in M_2$ with $(m_{a, 1}, m_{a, 2}) \in N$ such that $\varphi(a) = \psi(m_{a, 1},
  m_{a, 2})$. Indeed, we can define the maps
  $\varphi_1: \Sigma \to M_1$ and $\varphi_2: \Sigma \to M_2$ by setting $\varphi_1(a) := m_{a, 1}$
  and $\varphi_2(a) := m_{a, 2}$. These maps can be lifted into homomorphisms
  $\varphi_1: \Sigma^* \to M_1$ and $\varphi_2: \Sigma^* \to M_2$. By induction, there are $n_1$
  and $n_2$ such that $u \equiv_{m, n_1}^X v$ implies $\varphi_1(u) = \varphi_2(v)$ and
  $u \equiv_{m, n_2}^Y v$ implies $\varphi_2(u) = \varphi_2(v)$ for any two finite words
  $u, v \in \Sigma^*$. By setting $n = \max \{ n_1, n_2 \}$,
  we have
  \[
    u \equiv_{m, n}^{XY} v \implies \varphi_1(u) = \varphi_1(v) \text{ and } \varphi_2(u) = \varphi_2(v)
  \]
  for all $u, v \in \Sigma^*$. For all $u, v \in \Sigma^*$ with $u \equiv_{m, n}^{XY} v$, this yields
  \begin{align*}
    \varphi(a_1 a_2 \dots a_k) &= \varphi(a_1) \varphi(a_2) \dots \varphi(a_k)\\
    &= \psi(m_{a_1, 1}, m_{a_1, 2}) \psi(m_{a_2, 1}, m_{a_2, 2}) \dots \psi(m_{a_k, 1}, m_{a_k, 2})\\
    &= \psi((m_{a_1, 1}, m_{a_1, 2}) (m_{a_2, 1}, m_{a_2, 2}) \dots (m_{a_k, 1}, m_{a_k, 2}))\\
    &= \psi(m_{a_1, 1} m_{a_2, 1} \dots m_{a_k, 1}, m_{a_1, 2} m_{a_2, 2} \dots m_{a_k, 2})\\
    &= \psi \left( \varphi_1(u), \varphi_2(u) \right)\\
    &= \psi \left( \varphi_1(v), \varphi_2(v) \right)\\
    &= \varphi(b_1 b_2 \dots b_l)
  \end{align*}
  where $u = a_1 a_2 \dots a_k$, $v = b_1 b_2 \dots b_l$ and $a_1, a_2, \dots, a_k,
  b_1, b_2, \dots, b_l \in \Sigma$.

  Finally, let $M \in \Rm[m + 1] \cap \Lm[m + 1]$ with $m > 1$. Denote by $2^\Sigma$ the monoid of
  subsets of $\Sigma$ whose binary operation is the union of sets. It is easy to see that $2^\Sigma$
  is $\mathcal{J}$-trivial. Therefore, we have $M \times 2^\Sigma \in \Rm[m + 1] \cap \Lm[m + 1]$.
  Next, we lift $\varphi: \Sigma^* \to M$ into a homomorphism
  $\hat{\varphi}: \Sigma^* \to M \times 2^\Sigma$ by taking the word's alphabet as the entry in the
  second component. If we show $u \equiv_{m, n}^\WI v \implies \hat{\varphi}(u) = \hat{\varphi}(v)$ for
  a suitable $n \in \Nat$, we have, in particular, $u \equiv_{m, n}^\WI v \implies \varphi(u) = \varphi(v)$.
  The advantage of this approach is that we have
  $\hat{\varphi}(u) = \hat{\varphi}(v) \implies \alphabet(u) = \alphabet(v)$ for all $u, v \in
  \Sigma^*$ by the construction of $\hat{\varphi}$. Instead of continuing to write $\hat{\varphi}$,
  we simply substitute $M$ by $M \times 2^\Sigma$ and $\varphi$ by $\hat{\varphi}$.

  We have $M / {\sim_K} \in \Lm$ and $M / {\sim_D} \in \Rm$. By $\approx$, denote the join of $\sim_K$
  and $\sim_D$. Since it is a homomorphic image of both, $M / {\sim_K}$ and $M / {\sim_D}$, the monoid
  $M / {\approx}$ is in $\Rm \cap \Lm$ and we can apply induction, which yields an $n' \in \Nat$ such
  that $u \equiv_{m - 1, n'}^\WI v$ implies $\varphi(u) \approx \varphi(v)$ for all finite words
  $u, v \in \Sigma^*$. Let $c$ be the sum of the number of $\mathcal{R}$-classes and the number of
  $\mathcal{L}$-classes in $M$ and set $n = n' + c$. Suppose we have $u \equiv_{m, n}^\WI v$ for two
  finite words $u, v \in \Sigma^*$. Consider the $\mathcal{R}$-factorization
  $u = u_0' a_1 u_1' a_2 u_2' \dots a_r u_r'$ of $u$ and the $\mathcal{L}$-factorization
  $v = v_0' b_1 v_1' b_2 v_2' \dots b_l v_l'$ of $v$. Clearly, we have $r + 1 + l + 1 \leq c$.
  Define the positions $p_0^w = -\infty$, $p_{r + 1}^w = +\infty$ and $p_i^w = X_{a_i}(w; p_{i - 1}^w)$
  for $i = 1, 2, \dots, r$ and $w = u, v$. By \autoref{lem:RDescendInDA}, we know that $p_i^u$
  denotes the position of $a_i$ in the $\mathcal{R}$-factorization for $i = 1, 2, \dots, r$.
  Symmetrically, we can define $q_{l + 1}^w = + \infty$, $q_0^w = -\infty$ and
  $q_j^w = Y_{a_j}(w; q_{j + 1}^w)$ for $j = l, l - 1, \dots, 1$ and $w = u, v$. Again, we know that
  $q_j^v$ is the position of $b_j$ in the $\mathcal{L}$-factorization of $v$ for $j = 1, 2, \dots, l$.
  Additionally, we have
  \begin{align*}
    &p_0^w < p_1^w < \dots < p_r^w < p_{r + 1}^w \text{ and}\\
    &q_0^w < q_1^w < \dots < q_l^w < q_{l + 1}^w
  \end{align*}
  for $w = u$ and $w = v$ by their definition. We are going to show that we have
  $p_i^u \mathrel{\triangledown} q_j^u \iff p_i^v \mathrel{\triangledown} q_j^v$ for
  $\triangledown \in \{ <, =, > \}$ and all $i = 1, 2, \dots, r$ and $j = 1, 2, \dots, l$. Together,
  these results yield that the sequence which is obtained by ordering the $p_i$ and $q_j$ positions
  in $u$ is equal to the corresponding sequence in $v$. To prove this assertion, assume that we
  have $q_j^u \leq p_i^u$ but $q_j^v > p_i^v$ for an $i \in \{ 1, 2, \dots, r \}$ and a
  $j \in \{ 1, 2, \dots, l \}$ (all other cases are symmetric or analogous). Without loss of
  generality, we may assume that $p_{i - 1}^u < q_j^u \leq p_i^u$ holds since, otherwise, we can
  substitute $i$ by a smaller $i$ for which the former holds. Note that this substitution does not
  violate the condition $q_j^v > p_i^v$ as $p_i^v$ gets strictly smaller if $i$ decreases. Equally
  without loss of generality, we may assume $q_j^u \leq p_i^u < q_{j + 1}^u$ by a dual argumentation.
  The situation is presented in \autoref{fig:piqjLeftRight}. We have
  \begin{align*}
    u_{(p_{i - 1}^u, q_{j + 1}^u)} &= u \cdot X_{a_1}^R X_{a_2}^R \dots X_{a_{i - 1}}^R Y_{b_l}^L Y_{b_{l - 1}}^L \dots Y_{b_{j + 1}}^L \text{ and}\\
    v_{(p_{i - 1}^v, q_{j + 1}^v)} &= v \cdot X_{a_1}^R X_{a_2}^R \dots X_{a_{i - 1}}^R Y_{b_l}^L Y_{b_{l - 1}}^L \dots Y_{b_{j + 1}}^L
  \end{align*}
  and $u_{(p_{i - 1}^u, q_{j + 1}^u)} \equiv_{m, n - (i - 1) - (l - (j + 1) + 1)}^\WI
  v_{(p_{i - 1}^v, q_{j + 1}^v)}$, which yields $u_{(p_{i - 1}^u, q_{j + 1}^u)} \equiv_{m, 2}^\WI
  v_{(p_{i - 1}^v, q_{j + 1}^v)}$ because of
  \begin{align*}
    n - (i - 1) - (l - (j + 1) + 1) &= n' + c - i + 1 - l + j + 1 - 1\\
    &= n' + c - i - l + j + 1\\
    &\geq n' + c - (r + l) + 1\\
    &\geq n' + c - (c - 2) + 1 = n' + 3\\
    &> 2\text{.}
  \end{align*}
  If $q_j^u = p_i^u$, we have a contradiction since $u_{(p_{i - 1}^u, q_{j + 1}^u)} \cdot Y_{b_j}^L$
  contains no $a_i$ while $v_{(p_{i - 1}^v, q_{j + 1}^v)} \cdot Y_{b_j}^L$ does. For
  $q_j^u < p_i^u$, we can apply $C_{a_i, b_j}$ to $u_{(p_{i - 1}^u, q_{j + 1}^u)}$ while we cannot
  apply it to $v_{(p_{i - 1}^v, q_{j + 1}^v)}$ by its definition. Both situations constitute a
  contradiction.

  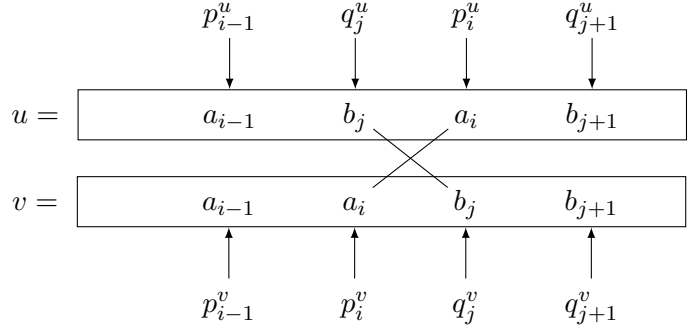
\begin{figure}[t]
    \begin{center}
      \begin{tikzpicture}[text height=.8em, text depth=.2em]
        \node (uLabel) {$u = {}$};
        \matrix [right=0cm of uLabel, rectangle, draw, matrix of math nodes, ampersand replacement=\&, inner sep=2pt] (u) {
          \&[1.5cm] a_{i - 1} \&[1cm] b_j \&[1cm] a_i \&[1cm] b_{j + 1} \&[0.75cm] \\
        };
        \node[above=0.75cm of u-1-2] (pi-1u) {$p_{i - 1}^u$};
        \draw (pi-1u) edge[-latex, shorten >= 2pt] (u-1-2);
        \node[above=0.75cm of u-1-3] (qju) {$q_j^u$};
        \draw (qju) edge[-latex, shorten >= 2pt] (u-1-3);
        \node[above=0.75cm of u-1-4] (piu) {$p_i^u$};
        \draw (piu) edge[-latex, shorten >= 2pt] (u-1-4);
        \node[above=0.75cm of u-1-5] (qj+1u) {$q_{j + 1}^u$};
        \draw (qj+1u) edge[-latex, shorten >= 2pt] (u-1-5);

        \node[below=0.5cm of uLabel] (vLabel) {$v = {}$};
        \matrix [right=0cm of vLabel, rectangle, draw, matrix of math nodes, ampersand replacement=\&, inner sep=2pt] (v) {
          \&[1.5cm] a_{i - 1} \&[1cm] a_i \&[1cm] b_j \&[1cm] b_{j + 1} \&[0.75cm] \\
        };
        \node[below=0.75cm of v-1-2] (pi-1v) {$p_{i - 1}^v$};
        \draw (pi-1v) edge[-latex, shorten >= 2pt] (v-1-2);
        \node[below=0.75cm of v-1-4] (qjv) {$q_j^v$};
        \draw (qjv) edge[-latex, shorten >= 2pt] (v-1-4);
        \node[below=0.75cm of v-1-3] (piv) {$p_i^v$};
        \draw (piv) edge[-latex, shorten >= 2pt] (v-1-3);
        \node[below=0.75cm of v-1-5] (qj+1v) {$q_{j + 1}^v$};
        \draw (qj+1v) edge[-latex, shorten >= 2pt] (v-1-5);

        \draw (u-1-3) -- (v-1-4);
        \draw (u-1-4) -- (v-1-3);
      \end{tikzpicture}
    \end{center}
    \caption{Contradiction: $p_i$ is to the right of $q_j$ in $u$ but to its left in $v$.}
    \label{fig:piqjLeftRight}
  \end{figure}

  We have proved that if we order the set $\{ p_i^u, q_j^u \mid i = 1, 2, \dots, r, j = 1, 2, \dots,\allowbreak l \}
  = \{ P_1^u, P_2^u,\allowbreak \dots,\allowbreak P_t^u \}$ (with $t \in \Nat_0$) such that
  \[
    P_1^u < P_2^u < \dots < P_t^u
  \]
  holds, then we can set
  \[
    P_s^v = \begin{cases}
      p_i^v & P_s^u = p_i^u \text{ for some } i \in \{ 1, 2, \dots, r \}\\
      q_j^v & P_s^u = q_j^u \text{ for some } j \in \{ 1, 2, \dots, l \}\\
    \end{cases}
  \]
  for $s = 1, 2, \dots, t$ and get
  \[
    P_1^v < P_2^v < \dots < P_t^v \text{.}
  \]
  These positions yield factorizations $u = u_0 c_1 u_1 c_2 u_2 \dots c_t u_t$ and
  $v = v_0 c_1 v_1 c_2\allowbreak v_2 \dots c_t v_t$ such that $c_s \in
  \{ a_i, b_j \mid i = 1, 2, \dots, r, j = 1, 2, \dots, l \}$ and $P_s^w$ denotes the position of
  $c_s$ in $w \in \{ u, v \}$ for $s = 1, 2, \dots, t$. To apply induction, we are going to show
  $u_s \equiv_{m - 1, n'}^\WI v_s$ for all $s = 1, 2, \dots, t$ next.

  To simplify notation, we say \enquote{$P_s$ is an $\mathcal{R}$-position} for any $s \in \{ 1, 2,
  \dots, t \}$ if $P_s^u = p_i^u$ for some $i \in \{ 1, 2, \dots, r \}$ (or, equivalently, if
  $P_s^v = p_i^v$ for some $i$) and we say \enquote{$P_s$ is an $\mathcal{L}$-position} if
  $P_s^u = q_j^u$ for some $j \in \{ 1, 2, \dots, l \}$ (or, equivalently again, if $P_s^v = q_j^v$
  for some $j$). Note that this definition is \emph{not} exclusive, i.\,e.\ there can be a position
  which is both, an $\mathcal{R}$-position and an $\mathcal{L}$-position.

  Next, we consider the corner cases of $u_0$/$v_0$ and $u_t$/$v_t$. If $P_1$ is an
  $\mathcal{R}$-position, we have $c_1 = a_1$ and
  \begin{align*}
    u_0 = u \cdot X_{a_1}^L \quad \text{ as well as } \quad
    v_0 = v \cdot X_{a_1}^L \text{,}
  \end{align*}
  which yields $u_0 \equiv_{m - 1, n'}^\WI v_0$ by definition of $\equiv_{m, n}^\WI$ and because of
  $c > 0$. If $P_1$ is an $\mathcal{L}$-position,
  we have $c_1 = b_1$ and
  \begin{align*}
    u_0 &= u \cdot Y_{b_l}^L Y_{b_{l - 1}}^L \dots Y_{b_1}^L \text{ as well as}\\
    v_0 &= v \cdot Y_{b_l}^L Y_{b_{l - 1}}^L \dots Y_{b_1}^L \text{.}
  \end{align*}
  Because $l < c$, $u_0 \equiv_{m - 1, n'}^\WI v_0$ holds also
  in this case. For $u_t$ and $v_t$, we can apply a symmetric argumentation.

  Finally, we distinguish four cases for a fixed $s \in \{ 1, 2, \dots, t - 1 \}$. If $P_s$ and
  $P_{s + 1}$ are both $\mathcal{R}$-positions, then we have $c_s = a_i$ and $c_{s + 1} = a_{i + 1}$
  for some $i \in \{ 1, 2, \dots, r \}$ and also
  \begin{align*}
    u_s &= u \cdot X_{a_1}^R X_{a_2}^R \dots X_{a_i}^R X_{a_{i + 1}}^L \text{ as well as}\\
    v_s &= v \cdot X_{a_1}^R X_{a_2}^R \dots X_{a_i}^R X_{a_{i + 1}}^L \text{.}
  \end{align*}
  By definition of $\equiv_{m, n}^\WI$, because of $i + 1 \leq c$, we thus have
  $u_s \equiv_{m - 1, n'}^\WI v_s$. A symmetric
  argument applies if both, $P_s$ and $P_{s + 1}$, are $\mathcal{L}$-positions. If $P_s$ is an
  $\mathcal{R}$-position but $P_{s + 1}$ is an $\mathcal{L}$-position, then $c_s = a_i$ for some
  $i \in \{ 1, 2, \dots, r \}$ and $c_{s + 1} = b_j$ for some $j \in \{ 1, 2, \dots, l \}$, which
  yields
  \begin{align*}
    u_s &= u \cdot X_{a_1}^R X_{a_2}^R \dots X_{a_i}^R Y_{b_l}^L Y_{b_{l - 1}}^L \dots Y_{b_j}^L \text{ as well as}\\
    v_s &= v \cdot X_{a_1}^R X_{a_2}^R \dots X_{a_i}^R Y_{b_l}^L Y_{b_{l - 1}}^L \dots Y_{b_j}^L \text{.}
  \end{align*}
  Therefore, we have $u_s \equiv_{m - 1, n'}^\WI v_s$ because of the definition of
  $\equiv_{m, n}^\WI$ and $n - i - (l - j + 1) =
  n' + c - (i + 1 + l) + j \geq n' + c - c + 0 = n'$. The fourth case is the most interesting: if
  $P_s$ is an $\mathcal{L}$ but not an $\mathcal{R}$-position while $P_{s + 1}$ is an $\mathcal{R}$
  but not an $\mathcal{L}$-position, then $c_s = b_j$ for some $j \in \{ 1, 2, \dots, l \}$ and
  $c_{s + 1} = a_i$ for some $i \in \{ 1, 2, \dots, r \}$. Additionally, we have
  $p_{i - 1}^w < P_s^w = q_j^w < P_{s + 1}^w = p_i^w < q_{j + 1}$ for $w = u$ and for $w = v$. We
  define
  \begin{align*}
    \tilde{u} &= u \cdot X_{a_1}^R X_{a_2}^R \dots X_{a_{i - 1}}^R Y_{b_l}^L Y_{b_{l - 1}}^L \dots Y_{b_{j + 1}}^L \text{ as well as}\\
    \tilde{v} &= v \cdot X_{a_1}^R X_{a_2}^R \dots X_{a_{i - 1}}^R Y_{b_l}^L Y_{b_{l - 1}}^L \dots Y_{b_{j + 1}}^L
  \end{align*}
  (we consider the $X$-blocks as empty -- meaning that we do not factorize -- if $i = 1$ and the
  $Y$-blocks as empty if $j = l$). We have $\tilde{u} \equiv_{m, n - (i - 1) - (l - j)}^\WI \tilde{v}$.
  Because of $n - (i - 1) - (l - j) = n' + c - (i + l) + j + 1 \geq n' + c - (r + l) + 1 \geq
  n' + 1$, $u_s = \tilde{u} \cdot C_{a_i, b_j}$, $v_s = \tilde{v} \cdot C_{a_i, b_j}$ and the
  definition of $\equiv_{m, n}^\WI$, we have
  $u_s \equiv_{m - 1, n'}^\WI v_s$.

  We have shown $u_s \equiv_{m - 1, n'}^\WI v_s$ for all $s = 1, 2, \dots, t$ and, by induction,
  therefore, know that $\varphi(u_s) \approx \varphi(v_s)$, i.\,e.\ for a fixed
  $s \in \{ 1, 2, \dots, t \}$, there are $w_1, w_2, \dots, w_k \in \Sigma^*$ such that
  \[
    \varphi(u_s) = \varphi(w_1) \sim_K \varphi(w_2) \sim_D \dots \sim_K \varphi(w_{k - 1}) \sim_D \varphi(w_k) = \varphi(v_s)
  \]
  holds.

  Remember that we substituted $M$ by $M \times 2^\Sigma$ so that we can assume $\varphi(u) =
  \varphi(v) \implies \alphabet(u) = \alphabet(v)$ for all $u, v \in \Sigma^*$. We can extend this
  implication: if we have $\varphi(u) \sim_K \varphi(v)$ for two $u, v \in \Sigma^*$, then, by
  definition of $\sim_K$, we also have $\varphi(u)^{M!} \varphi(u) = \varphi(u)^{M!} \varphi(v)$
  because of $\varphi(u)^{M!} \varphi(u) \R \varphi(u)^{M!}$. Therefore, we have
  $\alphabet(u) = \alphabet(u) \cup \alphabet(v)$ by the implication stated above. By symmetry, we,
  thus, have $\alphabet(u) = \alphabet(v)$. Since we can apply a similar argumentation for $\sim_D$,
  we have $\varphi(u) \sim_K \varphi(v) \text{ or } \varphi(u) \sim_D \varphi(v) \implies
  \alphabet(u) = \alphabet(v)$ for all $u, v \in \Sigma^*$. This yields
  $\alphabet(u_s) = \alphabet(w_1) = \alphabet(w_2) = \dots = \alphabet(w_k) = \alphabet(v_s)$.

  Since the factorizations $u = u_0 c_1 u_1 c_2 u_2 \dots c_t u_t$ and
  $v = v_0 c_1 v_1 c_2 v_2 \dots c_t v_t$ are subfactorizations from the $\mathcal{R}$-factorization
  of $u$ and the $\mathcal{L}$-factorization of $v$, there are $x_s, y_s \in M$ with
  \begin{align*}
    \varphi(u_0 c_1 u_1 c_2 u_2 \dots c_s) &= \varphi(u_0 c_1 u_1 c_2 u_2 \dots c_s u_s) x_s \text{ and}\\
    \varphi(c_{s + 1} v_{s + 1} c_{s + 2} v_{s + 2} \dots c_t v_t) &= y_s \varphi(v_s c_{s + 1} v_{s + 1} c_{s + 2} v_{s + 2} \dots c_t v_t) \text{.}
  \end{align*}
  Because of $\alphabet(u_s) = \alphabet(w_i)$ for all $i \in \{ 1, 2, \dots, k \}$ and by
  \autoref{lem:RDescendInDAWords},
  \[
    \left( \varphi(u_s) x_s \right)^{M!} \R \left( \varphi(u_s) x_s \right)^{M!} \varphi(u_s) \text{ implies }
      \left( \varphi(u_s) x_s \right)^{M!} \R \left( \varphi(u_s) x_s \right)^{M!} \varphi(w_i) \text{.}
  \]
  Similarly, we have
  \[
    \left( y_s \varphi(v_s) \right)^{M!} \L \varphi(w_i) \left( y_s \varphi(v_s) \right)^{M!}
  \]
  for all $i \in \{ 1, 2, \dots, k \}$. For $\varphi(w_i) \sim_K \varphi(w_{i + 1})$, this implies
  \[
    \left( \varphi(u_s) x_s \right)^{M!} \varphi(w_i) = \left( \varphi(u_s) x_s \right)^{M!} \varphi(w_{i + 1})
  \]
  and
  \[
    \varphi(w_i) \left( y_s \varphi(v_s) \right)^{M!} = \varphi(w_{i + 1}) \left( y_s \varphi(v_s) \right)^{M!}
  \]
  for $\varphi(w_i) \sim_D \varphi(w_{i + 1})$. In either case, we have
  \[
    \left( \varphi(u_s) x_s \right)^{M!} \varphi(w_i) \left( y_s \varphi(v_s) \right)^{M!} =
      \left( \varphi(u_s) x_s \right)^{M!} \varphi(w_{i + 1}) \left( y_s \varphi(v_s) \right)^{M!} \text{,}
  \]
  which yields for any $i \in \{ 1, 2, \dots, k - 1 \}$:
  \begin{align*}
    &\phantom{{}={}}\varphi(u_0 c_1 u_1 c_2 u_2 \dots c_s w_i c_{s + 1} v_{s + 1} c_{s + 2} v_{s + 2} \dots c_t v_t)\\
    &= \varphi(u_0 c_1 u_1 c_2 u_2 \dots c_s) \left( \varphi(u_s) x_s \right)^{M!} \varphi(w_i) \left( y_s \varphi(v_s) \right)^{M!} \varphi(c_{s + 1} v_{s + 1} c_{s + 2} v_{s + 2} \dots c_t v_t)\\
    &= \varphi(u_0 c_1 u_1 c_2 u_2 \dots c_s) \left( \varphi(u_s) x_s \right)^{M!} \varphi(w_{i + 1}) \left( y_s \varphi(v_s) \right)^{M!} \varphi(c_{s + 1} v_{s + 1} c_{s + 2} v_{s + 2} \dots c_t v_t)\\
    &= \varphi(u_0 c_1 u_1 c_2 u_2 \dots c_s w_{i + 1} c_{s + 1} v_{s + 1} c_{s + 2} v_{s + 2} \dots c_t v_t)
  \end{align*}
  So, we can substitute $w_i$ by $w_{i + 1}$ and, therefore, also $u_s$ by $v_s$, i.\,e.\ we have
  \begin{align*}
    &\phantom{{}={}} \varphi(u_0 c_1 u_1 c_2 u_2 \dots c_s u_s c_{s + 1} v_{s + 1} c_{s + 2} v_{s + 2} \dots c_t u_t)\\
    &= \varphi(u_0 c_1 u_1 c_2 u_2 \dots c_s v_s c_{s + 1} v_{s + 1} c_{s + 2} v_{s + 2} \dots c_t v_t) \text{.}
  \end{align*}
  Consecutively applying the former equation for $s = t$, then for $s = t - 1$ and so on yields
  \begin{align*}
    \varphi(u) &= \varphi(u_0 c_1 u_1 c_2 u_2 \dots c_{t - 1} u_{t - 1} c_t u_t)\\
    &= \varphi(u_0 c_1 u_1 c_2 u_2 \dots c_{t - 1} u_{t - 1} c_t v_t)\\
    &= \varphi(u_0 c_1 u_1 c_2 u_2 \dots c_{t - 1} v_{t - 1} c_t v_t)\\
    &\phantom{=} \vdots\\
    &= \varphi(v_0 c_1 v_1 c_2 v_2 \dots c_{t - 1} v_{t - 1} c_t v_t)\\
    &= \varphi(v) \text{,}
  \end{align*}
  which concludes the proof.
\end{proof}

It remains to show the other direction of \autoref{thm:TWcharacterization} (i.\,e.\ the converse of
\autoref{thm:relationsImplyEqualityUnderHom}). We will do this by using the equations from
\autoref{lem:TrotterWeilEquations}. Therefore, we begin with the following lemma.

\begin{lemma}\label{lem:UmVmAreInRelation}
  Let $\Sigma_m$, $U_m$ and $V_m$ be as in \autoref{lem:TrotterWeilEquations} and let $k \in \Nat$
  and $n \in \Nat_0$. Then, for every assignment of variables $\sigma: \Sigma_m^* \to \Sigma^*$ and
  all $Z \in \{ X, Y, XY, \WI \}$, we have
  \begin{align*}
    \sigma \left( \subs{U_m}{nk} \right) &\equiv_{m, n}^Z \sigma \left( \subs{V_1}{nk} \right)\\
    \sigma \left( \subs{ \left( U_m x_{m + 1} \right)^\pi U_m }{nk} \right) &\equiv_{m + 1, n}^X
      \sigma \left( \subs{ \left( U_m x_{m + 1} \right)^\pi V_m }{nk} \right)\\
    \sigma \left( \subs{ U_m \left( U_m x_{m + 1} \right)^\pi }{nk} \right) &\equiv_{m + 1, n}^Y
      \sigma \left( \subs{ V_m \left( U_m x_{m + 1} \right)^\pi }{nk} \right)
  \end{align*}
  for all $m \in \Nat$.
\end{lemma}
\begin{proof}
  We only show the first assertion by induction over $m$. The other assertions follow by similar
  arguments. For $m = 1$, we may assume $n > 0$ since, otherwise, there is nothing to show. We have
  \[
    U_1 = (s x_1)^\pi s (y_1 t)^\pi \quad \text{ and } \quad V_1 = (s x_1)^\pi t (y_1 t)^\pi \text{.}
  \]
  Let $u_1 = \sigma \left( \subs{U_1}{nk} \right)$ and $v_1 = \sigma \left(
  \subs{V_1}{nk} \right)$. First, assume $Z = X$. We are only interested in at most $n$
  consecutive simultaneous $X_\Sigma^R$ factorizations of $u_1$ and $v_1$ because, as soon as we apply
  at least one $X_\Sigma^L$ factorization, we know that $\equiv_{0, n}^X$ holds. As long as we apply
  only factorizations $X_a^R$ with $a \in \alphabet(\sigma(s x_1))$, the factorization position
  stays in the $(s x_1)^\pi$ part of $u_1$ and $v_1$. Since the number of remaining factorizations
  decreases, the right parts will eventually be in relation under $\equiv_{m, 0}^X$. If there is at
  least one $X_a^R$ factorization in the sequence where $a$ is in $\alphabet(\sigma(y_1 t))
  \setminus \alphabet(\sigma(s x_1))$, the right-hand side of $u_1$ belongs to the $(y_1 t)^{\pi}$ part and
  the right-hand side of $v_1$ belongs to the $t (y_1 t)^{\pi}$ part; but in both words, there are still at
  least $n - 1$ instances of $\sigma(y_1 t)$, which implies that the right-hand sides are equal under
  $\equiv_{m, n - 1}^X$. For $Z = Y$, the argumentation is symmetric, which also handles the $Z = {XY}$
  case. The additional $C_{a, b}$ of $Z = \WI$ needs no special handling since it decreases the
  first index of $\equiv_{m, n}^\WI$ to $m - 1 = 0$ anyway.

  To conclude the induction, we show
  \[
    u_{m + 1} = \sigma \left( \subs{U_{m + 1}}{nk} \right) \equiv_{m + 1, n}^Z\allowbreak
      \sigma \left( \subs{V_{m + 1}}{nk} \right) = v_{m + 1}
  \]
  next. For a schematic representation of $u_{m + 1}$ and $v_{m + 1}$, the reader is advised to
  refer to \autoref{fig:umAndvm}.
  \begin{figure}[t]
    \resizebox{\linewidth}{!}{%
    \begin{tikzpicture}[text height=.8em, text depth=.2em]
      \node (u) {$u_{m + 1} = {}$};
      \matrix [right=0cm of u, rectangle, draw, matrix of math nodes, ampersand replacement=\&, inner sep=2pt] (uword) {
        (u_m x_{m + 1}) \& (u_m x_{m + 1}) \& \dots \& (u_m x_{m + 1}) \& u_m \& (y_{m + 1} u_m) \& (y_{m + 1} u_m) \& \dots \& (y_{m + 1} u_m) \\
      };
      \node[below=0.5cm of u] (v) {$v_{m + 1} = {}$};
      \matrix [right=0cm of v, rectangle, draw, matrix of math nodes, ampersand replacement=\&, inner sep=2pt] (vword) {
        (u_m x_{m + 1}) \& (u_m x_{m + 1}) \& \dots \& (u_m x_{m + 1}) \& v_m \& (y_{m + 1} u_m) \& (y_{m + 1} u_m) \& \dots \& (y_{m + 1} u_m) \\
      };
      \path[decorate, decoration=brace, draw]
        ([yshift=-4pt]vword-1-4.south east) --
          node[anchor=north, yshift=-\baselineskip, align=center]
            {Repetition of the same word\\for $nk \geq n$ times.}
        ([yshift=-4pt]vword-1-1.south west);
      \path[decorate, decoration=brace, draw]
        ([yshift=-4pt]vword-1-9.south east) --
          node[anchor=north, yshift=-\baselineskip, align=center]
            {Repetition of the same word\\for $nk \geq n$ times.}
        ([yshift=-4pt]vword-1-6.south west);
      \path[decorate, decoration=brace, draw]
        ([yshift=+4pt]uword-1-5.north west) --
          node[anchor=south, align=center] {part containing the difference}
        ([yshift=+4pt]uword-1-5.north east);
    \end{tikzpicture}}
    \caption{Schematic representation of $u_{m + 1}$ and $v_{m + 1}$. Note that $x_{m + 1}$ and
    $y_{m + 1}$ are identified with $\sigma(x_{m + 1})$ and $\sigma(y_{m + 1})$, respectively.}\label{fig:umAndvm}
  \end{figure}
  For convenience, we identify $x_{m + 1}$ with $\sigma(x_{m + 1})$ and $y_{m + 1}$ with
  $\sigma(y_{m + 1})$. We will only show the case $Z = \WI$ since this is the most difficult one and
  the other ones are similar. In fact, we will proof the following claim by an inner induction on
  $l$ for $0 \leq l \leq n$:
  \[
    (u_m x_{m + 1})^l u_m (y_{m + 1} u_m)^l \equiv_{m + 1, l}^\WI (u_m x_{m + 1})^l v_m (y_{m + 1} u_m)^l
  \]
  The actual assertion then follows for the case $l = n$. For $l = 0$, there is nothing to show. So,
  let $l > 0$. First, consider an $X_a^L$ or $X_a^R$ factorizations. Only two cases can emerge: the
  factorization happens in the $(u_m x_{m + 1})^l$ part or the factorization happens in the
  $(y_{m + 1} u_m)^l$ part (in both words simultaneously). The factorization cannot happen in the
  central $u_m$ or $v_m$ part because we have $\alphabet(v_m) = \alphabet(u_m) \subseteq
  \alphabet(u_m x_{m + 1})$.

  First, consider the case in which the factorization happens at the beginning. Clearly, if this is
  the case, then the factorization must occur in the first $(u_m x_{m + 1})$ factor of both words at
  the same position:
  \begin{center}
    \begin{tikzpicture}[text height=.8em, text depth=.2em]
      \matrix [rectangle, draw, matrix of math nodes, ampersand replacement=\&, inner sep=2pt] (uword) {
        (u_m x_{m + 1}) \& (u_m x_{m + 1})^{l - 1} \& u_m \& (y_{m + 1} u_m)^{l - 1} \& (y_{m + 1} u_m)\\
      };
      \matrix [below=2cm of uword, rectangle, draw, matrix of math nodes, ampersand replacement=\&, inner sep=2pt] (vword) {
        (u_m x_{m + 1}) \& (u_m x_{m + 1})^{l - 1} \& v_m \& (y_{m + 1} u_m)^{l - 1} \& (y_{m + 1} u_m)\\
      };
      \draw (uword.west)+(0, +1cm) edge[-latex, shorten >= 2pt, bend left, in=+130] node[below, pos=0.1] {$X_a$} (uword-1-1.north);
      \draw (vword.west)+(0, -1cm) edge[-latex, shorten >= 2pt, bend right, in=-130] node[above, pos=0.1] {$X_a$} (vword-1-1.south);

      \path[decorate, decoration=brace, draw]
        ([yshift=-4pt]$(uword-1-1.south)+(-2pt, 0)$) -- node[shape=coordinate, below, name=1top] {}
        ([yshift=-4pt]$(uword-1-1.south west)+(+2pt, 0)$);
      \path[decorate, decoration=brace, draw]
        ([yshift=+4pt]$(vword-1-1.north west)+(+2pt, 0)$) -- node[shape=coordinate, above, name=1bottom] {}
        ([yshift=+4pt]$(vword-1-1.north)+(-2pt, 0)$);
      \path (1top) -- node[pos=0.5, anchor=center, rotate=-90] {$=$} (1bottom);

      \path[decorate, decoration=brace, draw]
        ([yshift=-4pt]$(uword-1-1.south east)+(-2pt, 0)$) -- node[shape=coordinate, below, name=2top] {}
        ([yshift=-4pt]$(uword-1-1.south)+(+2pt, 0)$);
      \path[decorate, decoration=brace, draw]
        ([yshift=+4pt]$(vword-1-1.north)+(+2pt, 0)$) -- node[shape=coordinate, above, name=2bottom] {}
        ([yshift=+4pt]$(vword-1-1.north east)+(-2pt, 0)$);
      \path (2top) -- node[pos=0.5, anchor=center, rotate=-90] {$=$} (2bottom);

      \path[decorate, decoration=brace, draw]
        ([yshift=-4pt]$(uword-1-4.south east)+(-2pt, 0)$) -- node[shape=coordinate, below, name=3top] {}
        ([yshift=-4pt]$(uword-1-2.south west)+(+2pt, 0)$);
      \path[decorate, decoration=brace, draw]
        ([yshift=+4pt]$(vword-1-2.north west)+(+2pt, 0)$) -- node[shape=coordinate, above, name=3bottom] {}
        ([yshift=+4pt]$(vword-1-4.north east)+(-2pt, 0)$);
      \path (3top) -- node[pos=0.5, anchor=center, rotate=-90] {$\equiv_{m + 1, l - 1}^\WI$} (3bottom);

      \path[decorate, decoration=brace, draw]
        ([yshift=-4pt]$(uword-1-5.south east)+(-2pt, 0)$) -- node[shape=coordinate, below, name=4top] {}
        ([yshift=-4pt]$(uword-1-5.south west)+(+2pt, 0)$);
      \path[decorate, decoration=brace, draw]
        ([yshift=+4pt]$(vword-1-5.north west)+(+2pt, 0)$) -- node[shape=coordinate, above, name=4bottom] {}
        ([yshift=+4pt]$(vword-1-5.north east)+(-2pt, 0)$);
      \path (4top) -- node[pos=0.5, anchor=center, rotate=-90] {$=$} (4bottom);
    \end{tikzpicture}
  \end{center}

  So, we have $(u_m x_{m + 1})^l u_m (y_{m + 1} u_m)^l \cdot X_a^L \equiv_{m, l - 1}^\WI
  (u_m x_{m + 1})^l v_m (y_{m + 1} u_m)^l \cdot X_a^L$ because the two words are equal. For an
  $X_a^R$ factorization, we have to apply induction on $l$ and get
  \[
    (u_m x_{m + 1})^{l - 1} u_m (y_{m + 1} u_m)^{l - 1} \equiv_{m + 1, l - 1}^\WI
      (u_m x_{m + 1})^{l - 1} v_m (y_{m + 1} u_m)^{l - 1}
  \]
  in the middle. Thus, we also get $(u_m x_{m + 1})^l u_m (y_{m + 1} u_m)^l \cdot X_a^R
  \equiv_{m + 1, l - 1}^\WI (u_m x_{m + 1})^l v_m (y_{m + 1} u_m)^l \cdot X_a^R$ (because the other
  word parts are equal and because $\equiv_{m + 1, l - 1}^\WI$ is a congruence).

  If the factorization happens at the end, the situation is similar. Clearly, the factorization can
  only happen at the same position in the first $(y_{m + 1} u_m)$ factor of each word, respectively
  (in fact, it can only happen within the first $y_{m + 1}$):
  
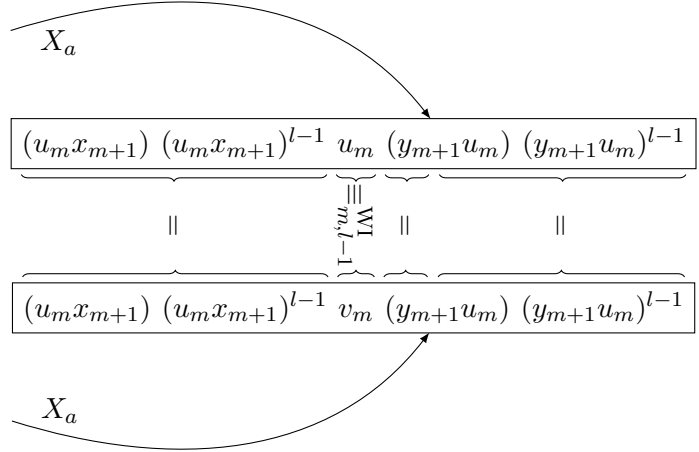
\begin{center}
    \begin{tikzpicture}[text height=.8em, text depth=.2em]
      \matrix [rectangle, draw, matrix of math nodes, ampersand replacement=\&, inner sep=2pt] (uword) {
        (u_m x_{m + 1}) \& (u_m x_{m + 1})^{l - 1} \& u_m \& (y_{m + 1} u_m) \& (y_{m + 1} u_m)^{l - 1}\\
      };
      \matrix [below=1.5cm of uword, rectangle, draw, matrix of math nodes, ampersand replacement=\&, inner sep=2pt] (vword) {
        (u_m x_{m + 1}) \& (u_m x_{m + 1})^{l - 1} \& v_m \& (y_{m + 1} u_m) \& (y_{m + 1} u_m)^{l - 1}\\
      };
      \draw (uword.west)+(0, +1.5cm) edge[-latex, shorten >= 2pt, bend left, in=+140] node[below, pos=0.1] {$X_a$} ($(uword-1-4.north)+(-0.5em, 0)$);
      \draw (vword.west)+(0, -1.5cm) edge[-latex, shorten >= 2pt, bend right, in=-140] node[above, pos=0.1] {$X_a$} ($(vword-1-4.south)+(-0.5em, 0)$);

      \path[decorate, decoration=brace, draw]
        ([yshift=-4pt]$(uword-1-2.south east)+(-2pt, 0)$) -- node[shape=coordinate, below, name=1top] {}
        ([yshift=-4pt]$(uword-1-1.south west)+(+2pt, 0)$);
      \path[decorate, decoration=brace, draw]
        ([yshift=+4pt]$(vword-1-1.north west)+(+2pt, 0)$) -- node[shape=coordinate, above, name=1bottom] {}
        ([yshift=+4pt]$(vword-1-2.north east)+(-2pt, 0)$);
      \path (1top) -- node[pos=0.5, anchor=center, rotate=-90] {$=$} (1bottom);

      \path[decorate, decoration=brace, draw]
        ([yshift=-4pt]$(uword-1-3.south east)+(-2pt, 0)$) -- node[shape=coordinate, below, name=2top] {}
        ([yshift=-4pt]$(uword-1-3.south west)+(+2pt, 0)$);
      \path[decorate, decoration=brace, draw]
        ([yshift=+4pt]$(vword-1-3.north west)+(+2pt, 0)$) -- node[shape=coordinate, above, name=2bottom] {}
        ([yshift=+4pt]$(vword-1-3.north east)+(-2pt, 0)$);
      \path (2top) -- node[pos=0.5, anchor=center, rotate=-90] {$\equiv_{m, l - 1}^\WI$} (2bottom);

      \path[decorate, decoration=brace, draw]
        ([yshift=-4pt]$(uword-1-4.south)+(-2pt-0.5em, 0)$) -- node[shape=coordinate, below, name=3top] {}
        ([yshift=-4pt]$(uword-1-4.south west)+(+2pt, 0)$);
      \path[decorate, decoration=brace, draw]
        ([yshift=+4pt]$(vword-1-4.north west)+(+2pt, 0)$) -- node[shape=coordinate, above, name=3bottom] {}
        ([yshift=+4pt]$(vword-1-4.north)+(-2pt-0.5em, 0)$);
      \path (3top) -- node[pos=0.5, anchor=center, rotate=-90] {$=$} (3bottom);

      \path[decorate, decoration=brace, draw]
        ([yshift=-4pt]$(uword-1-5.south east)+(-2pt, 0)$) -- node[shape=coordinate, below, name=4top] {}
        ([yshift=-4pt]$(uword-1-4.south)+(+2pt, 0)-(0.5em, 0)$);
      \path[decorate, decoration=brace, draw]
        ([yshift=+4pt]$(vword-1-4.north)+(+2pt, 0)-(0.5em, 0)$) -- node[shape=coordinate, above, name=4bottom] {}
        ([yshift=+4pt]$(vword-1-5.north east)+(-2pt, 0)$);
      \path (4top) -- node[pos=0.5, anchor=center, rotate=-90] {$=$} (4bottom);
    \end{tikzpicture}
  \end{center}
  We have $u_m \equiv_{m, n}^\WI v_m$ by induction on $m$ and (because $l \leq n$) also $u_m
  \equiv_{m, l - 1}^\WI v_m$. Because the other parts of the words are equal, respectively, and
  because $\equiv_{m, l - 1}^\WI$ is a congruence, we get $(u_m x_{m + 1})^l u_m (y_{m + 1} u_m)^l
  \cdot X_a^L \equiv_{m, l - 1}^\WI (u_m x_{m + 1})^l v_m (y_{m + 1} u_m)^l \cdot X_a^L$. We have
  $(u_m x_{m + 1})^l u_m (y_{m + 1} u_m)^l \cdot X_a^R \equiv_{m + 1, l - 1}^\WI
  (u_m x_{m + 1})^l v_m (y_{m + 1} u_m)^l \cdot X_a^R$ directly because the word parts on the right
  coincide.

  Because $Y_a^L$ and $Y_a^R$ factorizations are symmetrical, the only remaining case is a
  $C_{a, b}$ factorization. It is not possible that $C_{a, b}$ is defined on one of the words but
  not on the other because in each respective part the same letters appear. If $C_{a, b}$ is defined
  on $(u_m x_{m + 1})^l u_m (y_{m + 1} u_m)^l$ and on $(u_m x_{m + 1})^l v_m (y_{m + 1} u_m)^l$,
  then the only possible situation is that the $X_a^L$ factorization happens at the same position of
  the first $(y_{m + 1} u_m)$ part in each word and the $Y_b^R$ factorization happens at the same
  position of the last $(u_m x_{m + 1})$ factor in each word:
  
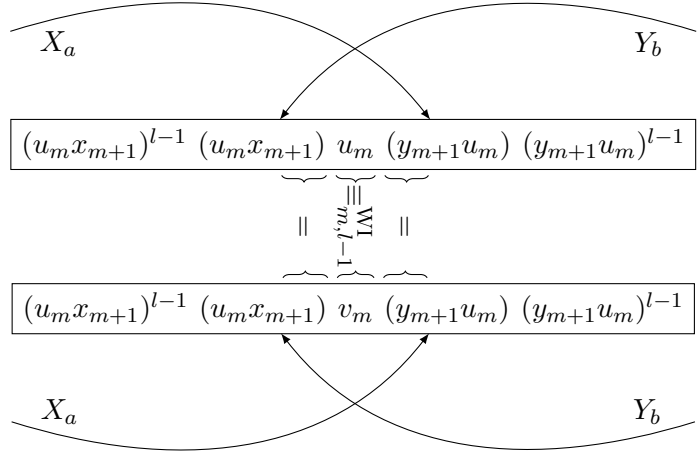
\begin{center}
    \begin{tikzpicture}[text height=.8em, text depth=.2em]
      \matrix [rectangle, draw, matrix of math nodes, ampersand replacement=\&, inner sep=2pt] (uword) {
        (u_m x_{m + 1})^{l - 1} \& (u_m x_{m + 1}) \& u_m \& (y_{m + 1} u_m) \& (y_{m + 1} u_m)^{l - 1}\\
      };
      \matrix [below=1.5cm of uword, rectangle, draw, matrix of math nodes, ampersand replacement=\&, inner sep=2pt] (vword) {
        (u_m x_{m + 1})^{l - 1} \& (u_m x_{m + 1}) \& v_m \& (y_{m + 1} u_m) \& (y_{m + 1} u_m)^{l - 1}\\
      };
      \draw (uword.west)+(0, +1.5cm) edge[-latex, shorten >= 2pt, bend left, in=+140] node[below, pos=0.1] {$X_a$} ($(uword-1-4.north)+(-0.5em, 0)$);
      \draw (vword.west)+(0, -1.5cm) edge[-latex, shorten >= 2pt, bend right, in=-140] node[above, pos=0.1] {$X_a$} ($(vword-1-4.south)+(-0.5em, 0)$);
      \draw (uword.east)+(0, +1.5cm) edge[-latex, shorten >= 2pt, bend right, in=-140] node[below, pos=0.1] {$Y_b$} ($(uword-1-2.north)+(+0.5em, 0)$);
      \draw (vword.east)+(0, -1.5cm) edge[-latex, shorten >= 2pt, bend left, in=+140] node[above, pos=0.1] {$Y_b$} ($(vword-1-2.south)+(+0.5em, 0)$);

      \path[decorate, decoration=brace, draw]
        ([yshift=-4pt]$(uword-1-2.south east)+(-2pt, 0)$) -- node[shape=coordinate, below, name=1top] {}
        ([yshift=-4pt]$(uword-1-2.south)+(+2pt, 0)+(0.5em, 0)$);
      \path[decorate, decoration=brace, draw]
        ([yshift=+4pt]$(vword-1-2.north)+(+2pt, 0)+(0.5em, 0)$) -- node[shape=coordinate, above, name=1bottom] {}
        ([yshift=+4pt]$(vword-1-2.north east)+(-2pt, 0)$);
      \path (1top) -- node[pos=0.5, anchor=center, rotate=-90] {$=$} (1bottom);

      \path[decorate, decoration=brace, draw]
        ([yshift=-4pt]$(uword-1-3.south east)+(-2pt, 0)$) -- node[shape=coordinate, below, name=2top] {}
        ([yshift=-4pt]$(uword-1-3.south west)+(+2pt, 0)$);
      \path[decorate, decoration=brace, draw]
        ([yshift=+4pt]$(vword-1-3.north west)+(+2pt, 0)$) -- node[shape=coordinate, above, name=2bottom] {}
        ([yshift=+4pt]$(vword-1-3.north east)+(-2pt, 0)$);
      \path (2top) -- node[pos=0.5, anchor=center, rotate=-90] {$\equiv_{m, l - 1}^\WI$} (2bottom);

      \path[decorate, decoration=brace, draw]
        ([yshift=-4pt]$(uword-1-4.south)+(-2pt-0.5em, 0)$) -- node[shape=coordinate, below, name=3top] {}
        ([yshift=-4pt]$(uword-1-4.south west)+(+2pt, 0)$);
      \path[decorate, decoration=brace, draw]
        ([yshift=+4pt]$(vword-1-4.north west)+(+2pt, 0)$) -- node[shape=coordinate, above, name=3bottom] {}
        ([yshift=+4pt]$(vword-1-4.north)+(-2pt-0.5em, 0)$);
      \path (3top) -- node[pos=0.5, anchor=center, rotate=-90] {$=$} (3bottom);
    \end{tikzpicture}
  \end{center}
  As in the previous case, we have $u_m \equiv_{m, l - 1}^\WI v_m$ and, because the other parts of
  the words coincide respectively, also $(u_m x_{m + 1})^l u_m (y_{m + 1} u_m)^l \cdot C_{a, b}
  \equiv_{m, l - 1}^\WI (u_m x_{m + 1})^l v_m (y_{m + 1} u_m)^l \cdot C_{a, b}$.
\end{proof}

\begin{restatable}{theorem}{relationsInTrotterWeil}
  \label{thm:relationsInTrotterWeil}
  Let $m, n \in \Nat$. Then:
  \begin{multicols}{2}
    \begin{itemize}
      \item $\Sigma^* / {\equiv_{m, n}^X} \in \Rm$
      \item $\Sigma^* / {\equiv_{m, n}^Y} \in \Lm$
      \item $\Sigma^* / {\equiv_{m, n}^{XY}} \in \Rm \vee \Lm$
      \item $\Sigma^* / {\equiv_{m, n}^\WI} \in \Rm[m + 1] \cap \Lm[m + 1]$
    \end{itemize}
  \end{multicols}
\end{restatable}
\begin{proof}
  To prove the theorem, one needs to show that the equations from
  \autoref{lem:TrotterWeilEquations} hold in the respective monoid. To do this, it is worthwhile to
  make an observation: choose $m, n \in \Nat$ and $Z \in \{ X, Y, XY, \WI \}$ arbitrarily and let
  $M = \Sigma^* / {\equiv_{m, n}^Z}$. The observation is that an equation $\alpha = \beta$ holds in
  $M$ if and only if $\sigma \left( \subs{\alpha}{n \cdot M!} \right) \equiv_{m, n}^Z \sigma
  \left( \subs{\beta}{n \cdot M!} \right)$ holds for all assignments $\sigma: \Gamma \to \Sigma^*$
  where $\Gamma$ is the alphabet of $\alpha$ and $\beta$ (i.\,e.\ the set of variables appearing in
  $\alpha$ and $\beta$).

  So, to show $\Sigma^* / {\equiv_{m, n}^\WI} \in \Rm[m + 1] \cap \Lm[m + 1]$, it suffices to show
  that $\sigma \left( \subs{U_m}{n \cdot M!} \right)\allowbreak \equiv_{m, n}^\WI
  \sigma \left( \subs{V_m}{n \cdot M!} \right)$ holds for all assignments of variables
  $\sigma: \Sigma_m^* \to \Sigma^*$. However, this follows from the first assertion of
  \autoref{lem:UmVmAreInRelation}. The same is true for $\Rm[1]$ and $\Lm[1]$. To show
  $\Sigma^* / {\equiv_{m + 1, n}^X} \in \Rm[m + 1]$ for $m \in \Nat_0$, we use the second assertion
  of \autoref{lem:UmVmAreInRelation} and, for $\Sigma^* / {\equiv_{m + 1, n}^Y} \in \Lm[m + 1]$,
  we use the third one.

  To prove that $M / {\equiv_{m, n}^{XY}}$ is in $\Rm \vee \Lm$, one can recycle an observation from the
  proof of \autoref{thm:relationsImplyEqualityUnderHom}: a monoid is in the join $\V \vee \V[W]$ of
  two varieties $\V$ and $\V[W]$ if and only if it is a divisor of a direct product $M_1 \times M_2$
  such that $M_1 \in \V$ and $M_2 \in \V[W]$. Indeed, for any two congruences $\mathcal{C}_1$ and
  $\mathcal{C}_2$ over a monoid $N$, $N / \left({\mathcal{C}_1 \cap \mathcal{C}_2} \right)$ is a divisor of
  $N / {\mathcal{C}_1} \times N / {\mathcal{C}_2}$ (as can be shown easily). Therefore,
  $M / {\equiv_{m, n}^{XY}}$ is a divisor of the direct product of $M / {\equiv_{m, n}^X} \in \Rm$ and
  $M / {\equiv_{m, n}^Y} \in \Lm$.
\end{proof}

To conclude this section, we finally state and prove the converse of
\autoref{lem:TrotterWeilEquations}. This gives us a full characterization of the Trotter-Weil
Hierarchy in terms of equations (see also \cite{kufleitner2012join}). Note that we include this proof
only for the sake of completeness. Strictly speaking, it is not necessary for the remainder of this
paper.
\begin{lemma}\label{lem:TrotterWeilEquationsConverse}
  Define the $\pi$-terms
  \[
    U_1 = (s x_1)^\pi s (y_1 t)^\pi \quad \text{and} \quad V_1 = (s x_1)^\pi t (y_1 t)^\pi
  \]
  over the alphabet $\Sigma_1 = \{ s, t, x_1 \}$. For $m \in \Nat$, let $x_{m + 1}$ and $y_{m + 1}$
  be new characters not in the alphabet $\Sigma_m$ and define the $\pi$-terms
  \[
    U_{m + 1} = (U_m x_{m + 1})^\pi U_m (y_{m + 1} U_m)^\pi \quad \text{and} \quad
    V_{m + 1} = (U_m x_{m + 1})^\pi V_m (y_{m + 1} U_m)^\pi
  \]
  over the alphabet $\Sigma_{m + 1} = \Sigma_m \uplus \{ x_{m + 1}, y_{m + 1} \}$.

  Then we have
  \begin{align*}
    M \in \Rm[1] = \Lm[1] = \V[J] &\implies U_1 = V_1 \text{ holds in } M \text{,}\\
    M \in \Rm[m + 1] &\implies (U_m x_{m + 1})^\pi U_m = (U_m x_{m + 1})^\pi V_m \text{ holds in } M
      \text{,}\\
    M \in \Lm[m + 1] &\implies U_m (y_{m + 1} U_m)^\pi = V_m (y_{m + 1} U_m)^\pi \text{ holds in } M
      \text{ and}\\
    M \in \Rm[m + 1] \cap \Lm[m + 1] &\implies U_m = V_m \text{ holds in } M
  \end{align*}
  for all $m \in \Nat$.
\end{lemma}
\begin{proof}
  As stated in the proof of \autoref{lem:TrotterWeilEquations}, the first implication is due to a
  well-known characterization in terms of equations of $\V[J]$ \cite{pin1986varieties}. The other
  implications can be proved by combining \autoref{lem:UmVmAreInRelation} and
  \autoref{thm:relationsImplyEqualityUnderHom}. For example, suppose $M \in \Rm[m + 1]$ and let
  $\sigma: (\Sigma_m \uplus \{ x_{m + 1} \})^* \to M$ be an arbitrary assignment of variables. We
  need to show $\sigma \left( \subs{(U_m x_{m + 1})^\pi U_m}{M!} \right) =
  \sigma \left( \subs{(U_m x_{m + 1})^\pi V_m}{M!} \right)$. By
  \autoref{thm:relationsImplyEqualityUnderHom}, there is an $n \in \mathbb{N}$ such that
  $u \equiv_{m + 1, n}^X v$ implies $\sigma(u) = \sigma(v)$ for all $u, v \in
  (\Sigma_m \uplus \{ x_{m + 1} \})^*$. Thus, for this $n$, we have
  \begin{align*}
    \sigma \left( \subs{(U_m x_{m + 1})^\pi U_m}{M!} \right)
      &= \sigma \left( \subs{(U_m x_{m + 1})^\pi U_m}{n \cdot M!} \right)\\
    = \sigma \left( \subs{(U_m x_{m + 1})^\pi V_m}{n \cdot M!} \right)
      &= \sigma \left( \subs{(U_m x_{m + 1})^\pi V_m}{M!} \right)
  \end{align*}
  since we have $\subs{(U_m x_{m + 1})^\pi U_m}{n \cdot M!} \equiv_{m + 1, n}^X
  \subs{(U_m x_{m + 1})^\pi V_m}{n \cdot M!}$ by \autoref{lem:UmVmAreInRelation}.
\end{proof}

\section{Relations and Equations}\label{sec:relationsAndEquations}

\paragraph*{Infinite Version of the Relations.}
So far, we have mainly considered finite words. In this section, which connects the relational
approach outlined above with equations, we finally need to consider generalized words. Thus,
we also define an infinite version of the above relations. Here, we allow an arbitrary number of
factorizations: for $m \in \Nat_0$ and $Z \in \smallset{X, Y, XY, \WI}$, define
\[
  u \equiv_m^Z v \iff \forall n \in \Nat: u \equiv_{m, n}^Z v \text{.}
\]
Notice that this definition is not useful for finite words as every one of them is in its own
class.

\paragraph*{Connecting the Relations and Equations.}
In order to solve the word problem for $\pi$-terms over the varieties in the Trotter-Weil Hierarchy,
one can use the following connection between the relations defined above
and equations in these varieties, which is straightforward if one makes the transition from
finite to infinite words.
Besides its use for the word problem for $\pi$-terms, this
connection is also interesting in its own right as it can be used to prove or disprove equations
in any of the varieties. As the class of monoids in which an equation $\alpha = \beta$ holds is a variety, one can
see the assertion for the join levels as an implication of the ones for the corners.
\begin{restatable}{theorem}{relationsAndEquations}
  \label{thm:relationsAndEquations}
  Let $\alpha$ and $\beta$ be two $\pi$-terms. For every $m \in \Nat$, we have:
  \begin{align*}
    \subs{\alpha}{\omega + \omega^*} \equiv_{m}^X \subs{\beta}{\omega + \omega^*} &\iff
      \alpha = \beta \text{ holds in } \Rm\\
    \subs{\alpha}{\omega + \omega^*} \equiv_{m}^Y \subs{\beta}{\omega + \omega^*} &\iff
      \alpha = \beta \text{ holds in } \Lm\\
    \subs{\alpha}{\omega + \omega^*} \equiv_{m}^{XY} \subs{\beta}{\omega + \omega^*} &\iff
      \alpha = \beta \text{ holds in } \Rm \vee \Lm\\
    \subs{\alpha}{\omega + \omega^*} \equiv_{m}^\WI \subs{\beta}{\omega + \omega^*} &\iff
      \alpha = \beta \text{ holds in } \Rm[m + 1] \cap \Lm[m + 1]
  \end{align*}
\end{restatable}
\noindent Using \autoref{fct:trotterWeilIsDA}, which states that $\DA$ is equal to the union of all varieties in the Trotter-Weil Hierarchy,
we immediately get the following corollary.
\begin{restatable}{corollary}{relationsAndDAEquations}
  \label{cor:relationsAndDAEquations}
  $
    \left( \forall m \in \Nat: \subs{\alpha}{\omega + \omega^*} \equiv_{m}^{XY} \subs{\beta}{\omega + \omega^*} \right)
      \iff \alpha = \beta \text{ holds in } \DA
  $
\end{restatable}

To prove \autoref{thm:relationsAndEquations}, we need two technical lemmas. The first one basically
says that a sufficiently large finite power is as good as an $\omega + \omega^*$ power.
\begin{lemma}\label{lem:UkVomega}
  Let $m \in \Nat_0$, $Z \in \smallset{X, Y, XY, \WI}$ and let $u$ and $v$ be accessible words.
  Then:
  \[
    u \equiv_{m, n}^Z v \implies \forall \, 0 \leq k \leq n: u^{k + 1} \equiv_{m, k}^Z v^{\omega + \omega^*}
  \]
\end{lemma}
\begin{proof}
  The case $m = 0$ is trivial. Therefore, let $m > 0$ and continue by induction over $k$. Again,
  the case $k = 0$ is trivial. To complete the induction, it remains to show that $u u^{k + 1}
  \equiv_{m, k + 1}^Z v^{\omega + \omega^*}$ holds for $k < n$. Obviously, $\alphabet(u^{k + 2}) =
  \alphabet \left( v^{\omega + \omega^*} \right)$ is satisfied by assumption. Now assume $Z = X$.
  The assumption $u \equiv_{m, n}^X v$ implies $u \equiv_{m - 1, n - 1}^Y v$. By induction on
  $m$, this yields $u^{k + 1} \equiv_{m - 1, k}^Y v^{\omega + \omega^*}$ and $u \equiv_{m - 1, k}^Y
  v$. Because $\equiv_{m - 1, k}^Y$ is a congruence, this shows $u u^{k + 1} \equiv_{m - 1, k}^Y v
  v^{\omega + \omega^*} = v^{\omega + \omega^*}$. Let $a \in \alphabet(u) = \alphabet(v)$. It
  remains to show $u u^{k + 1} \cdot X_a^L \equiv_{m - 1, k}^Y v^{\omega + \omega^*} \cdot X_a^L$
  and $u u^{k + 1} \cdot X_a^R \equiv_{m, k}^X v^{\omega + \omega^*} \cdot X_a^R$. For the former,
  note that $u \equiv_{m, n}^X v$ implies $u \cdot X_a^L \equiv_{m - 1, n - 1}^Y v \cdot X_a^L$,
  which in turn implies $u u^{k + 1} \cdot X_a^L = u \cdot X_a^L \equiv_{m - 1, k}^Y v \cdot X_a^L =
  vv^{\omega + \omega^*} \cdot X_a^L = v^{\omega + \omega^*} \cdot X_a^L$. For the latter, note that
  $u \equiv_{m, n}^X v$ implies $u \cdot X_a^R \equiv_{m, n - 1}^X v \cdot X_a^R$ and, thus,
  $u \cdot X_a^R \equiv_{m, k}^X v \cdot X_a^R$. By induction on $k$, we also have $u^{k + 1}
  \equiv_{m, k}^X v^{\omega + \omega^*}$. Together these yield $u u^{k + 1} \cdot X_a^R = (u \cdot
  X_a^R) u^{k + 1} \equiv_{m, k}^X (v \cdot X_a^R) v^{\omega + \omega^*} = v^{\omega + \omega^*}
  \cdot X_a^R$.

  The case for $Z = Y$ is symmetric and the case for $Z = XY$ follows directly. Finally, for
  $Z = \WI$ the argumentation is analogous because, for $k > 0$, neither $uu^{k + 1} \cdot C_{a, b}$ nor
  $v^{\omega + \omega^*} \cdot C_{a, b}$ is defined for any pair $(a, b)$ of letters.
\end{proof}

The second technical lemma states that, for finitely many factorization steps, only finitely many
positions of $\omega + \omega^*$ are relevant.
\begin{lemma}\label{lem:omegaAndKForPiTerms}
  Let $m, n \in \Nat_0$, $Z \in \smallset{X, Y, XY, \WI}$ and let $\gamma$ be a $\pi$-term. Then
  $$
    \subs{\gamma}{k} \equiv_{m, n}^Z \subs{\gamma}{\omega + \omega^*}
  $$
  holds for all $k \in \Nat_0$ with $k > n$.
\end{lemma}
\begin{proof}
  The cases for $m = 0$ or $n = 0$ are trivial. Thus, assume $m > 0$ and $n > 0$. If $\gamma =
  \varepsilon$ or $\gamma = a$ for an $a \in \Sigma$, then $\subs{\gamma}{k} = \gamma =
  \subs{\gamma}{\omega + \omega^*}$. If $\gamma = \alpha \beta$ for two $\pi$-terms $\alpha$ and
  $\beta$, then by induction $\subs{\alpha}{k} \equiv_{m, n}^Z \subs{\alpha}{\omega + \omega^*}$ and
  $\subs{\beta}{k} \equiv_{m, n}^Z \subs{\beta}{\omega + \omega^*}$ hold. As $\equiv_{m, n}^Z$ is
  a congruence, this implies $\subs{\gamma}{k} \equiv_{m, n}^Z \subs{\gamma}{\omega + \omega^*}$.

  Finally, let $\gamma = (\alpha)^\pi$ for a $\pi$-term $\alpha$. It remains to show that
  $\subs{\alpha}{k}^k \equiv_{m, n}^Z \subs{\alpha}{\omega + \omega^*}^{\omega + \omega^*}$.
  Clearly, the alphabetic condition is satisfied and, by induction, $\subs{\alpha}{k}\allowbreak \equiv_{m, n}^Z
  \subs{\alpha}{\omega + \omega^*}$ holds. For $Z = X$, let $a \in \alphabet(\alpha)$.
  By induction on $m$ and $n$, we have $\subs{\alpha}{k}^k = \subs{\gamma}{k} \equiv_{m - 1, n -
  1}^Y \subs{\gamma}{\omega + \omega^*} = \subs{\alpha}{\omega + \omega^*}^{\omega + \omega^*}$.
  We also have $\subs{\alpha}{k}^k \cdot X_a^L = \subs{\alpha}{k} \cdot X_a^L \equiv_{m - 1, n -
  1}^Y \subs{\alpha}{\omega + \omega^*} \cdot X_a^L = \subs{\alpha}{\omega + \omega^*}^{\omega +
  \omega^*} \cdot X_a^L$ by definition of $\subs{\alpha}{k} \equiv_{m, n}^X \subs{\alpha}{\omega +
  \omega^*}$. To show $\subs{\alpha}{k}^k \cdot X_a^R \equiv_{m, n - 1}^X \subs{\alpha}{\omega +
  \omega^*}^{\omega + \omega^*} \cdot X_a^R$, we use $\subs{\alpha}{k} \cdot X_a^R \equiv_{m, n - 1}^X
  \subs{\alpha}{\omega + \omega^*} \cdot X_a^R$, which holds because of $\subs{\alpha}{k}
  \equiv_{m, n}^X \subs{\alpha}{\omega + \omega^*}$, and show $\subs{\alpha}{k}^{k - 1}
  \equiv_{m, n - 1}^X \subs{\alpha}{\omega + \omega}^{\omega + \omega^*}$. This is sufficient
  because $\equiv_{m, n - 1}^X$ is a congruence and we have $uu^{\omega + \omega^*} = u^{\omega +
  \omega^*}$ for all words $u$. For showing $\subs{\alpha}{k}^{k - 1} \equiv_{m, n - 1}^X
  \subs{\alpha}{\omega + \omega}^{\omega + \omega^*}$, we write $k - 1 = k' + n$ for a $k' \in
  \Nat_0$ and
  \begin{align*}
    \subs{\alpha}{k}^{k - 1} &= \subs{\alpha}{k}^{k'} \subs{\alpha}{k}^{n} \text{ and }\\
    \subs{\alpha}{\omega + \omega^*}^{\omega + \omega^*} &= \subs{\alpha}{\omega + \omega^*}^{k'}
    \subs{\alpha}{\omega + \omega^*}^{\omega + \omega^*} \text{.}
  \end{align*}
  Setting $k = n - 1$ in \autoref{lem:UkVomega} yields $\subs{\alpha}{k}^n \equiv_{m, n - 1}^X
  \subs{\alpha}{\omega + \omega^*}^{\omega + \omega^*}$, which concludes the proof for $Z = X$
  since $\equiv_{m, n - 1}^X$ is a congruence.

  The case for $Z = Y$ is symmetric, which also shows the case for $Z = XY$. For $Z = \WI$,
  we note that $\subs{\alpha}{\omega + \omega^*}^{\omega + \omega^*} \cdot C_{a, b}$ is defined for
  no pair $a, b$ of letters. On the other hand, $\subs{\alpha}{k}^k \cdot C_{a, b}$ can only be
  defined for $k = 1$, in which case we are done because we have $n = 0$.
\end{proof}

Now, we are finally prepared to prove \autoref{thm:relationsAndEquations}.
\begin{proof}[Proof for \autoref{thm:relationsAndEquations}]
  The proof is structurally identical for all stated varieties. Therefore, we limit our discussion
  to $\Rm$.

  First, let $\subs{\alpha}{\omega + \omega^*} \equiv_m^X \subs{\beta}{\omega + \omega^*}$. Choose
  a monoid $M \in \Rm$ and an assignment of variables $\sigma: \Sigma^* \to M$. Because $M$ is in
  $\Rm$, there is an $n \in \Nat$ such that $u \equiv_{m, n}^X v$ implies $\sigma(u) = \sigma(v)$
  for any two words $u, v \in \Sigma^*$ (see also \autoref{thm:relationsImplyEqualityUnderHom}).
  Now, choose $c \in \Nat$ with $M! \cdot c > n$. Then by assumption and
  \autoref{lem:omegaAndKForPiTerms}, we have
  $$
    \Sigma^* \ni \subs{\alpha}{M! \cdot c} \equiv_{m, n}^X \subs{\alpha}{\omega + \omega^*}
      \equiv_{m, n}^X \subs{\beta}{\omega + \omega^*} \equiv_{m, n}^X \subs{\beta}{M! \cdot c}
      \in \Sigma^*
  $$
  and, therefore, $\sigma(\subs{\alpha}{M!}) = \sigma(\subs{\alpha}{M! \cdot c}) =
  \sigma(\subs{\beta}{M! \cdot c}) = \sigma(\subs{\beta}{M!})$, which is equivalent to $\alpha =
  \beta$ holding in $M$.

  Now, let $\subs{\alpha}{\omega + \omega^*} \not\equiv_m^X \subs{\beta}{\omega + \omega^*}$, which
  implies that there is an $n \in \Nat$ such that $\subs{\alpha}{\omega + \omega^*} \not\equiv_{m, n}^X
  \subs{\beta}{\omega + \omega^*}$. Define $M := \Sigma^* / {\equiv_{m, n}^X}$, which is in $\Rm$
  (by \autoref{thm:relationsInTrotterWeil}), and choose $c \in \Nat$ such that $M! \cdot
  c > n$. Then, by assumption and \autoref{lem:omegaAndKForPiTerms}, we have
  $$
    \Sigma^* \ni \subs{\alpha}{M! \cdot c} \equiv_{m, n}^X \subs{\alpha}{\omega + \omega^*}
      \not\equiv_{m, n}^X \subs{\beta}{\omega + \omega^*} \equiv_{m, n}^X \subs{\beta}{M! \cdot c}
      \in \Sigma^* \text{.}
  $$
  As assignment of variables $\sigma: \Sigma^* \to M$ choose the canonical projection. This yields
  $\sigma(\subs{\alpha}{M!}) = \sigma(\subs{\alpha}{M! \cdot c}) \neq \sigma(\subs{\beta}{M!
  \cdot c}) = \sigma(\subs{\beta}{M!})$, which means that $\alpha = \beta$ does \emph{not} hold in
  $M$.
\end{proof}

\section{Decidability}\label{sec:decidability}
In the previous section, we saw that checking whether $\alpha = \beta$ holds in a
variety of the Trotter-Weil Hierarchy boils down to checking $\subs{\alpha}{\omega + \omega^*}
\equiv_m^Z \subs{\beta}{\omega + \omega^*}$ (where $\equiv_m^Z$ depends on the variety in question).
In this section, we give an introduction on how to do this.
The presented approach works uniformly for all varieties in the
Trotter-Weil Hierarchy (in particular, it also works for the intersection levels, which tend to
be more complicated) and is designed to yield efficient algorithms.

The
definition of the relations which need to be tested is inherently recursive. One would
factorize $\subs{\alpha}{\omega + \omega^*}$ and $\subs{\beta}{\omega + \omega^*}$ on the first $a$
and/or last $b$ (for $a, b \in \Sigma$) and test the factors recursively. Therefore, the computation
is based on working with factors of words of the form $\subs{\gamma}{\omega + \omega^*}$ where
$\gamma$ is a $\pi$-term.
\begin{wrapfigure}{R}{0.4\textwidth}
  \vspace*{-.5\baselineskip}
  \centering\resizebox{!}{0.2\textwidth}{%
  \begin{tikzpicture}[scale=0.8,
  sibling distance=5ex,
  subtree/.style={shape border rotate=90, isosceles triangle, inner sep=2pt,
    edge from parent path={(\tikzparentnode.south) -- (\tikzchildnode.north)}},
  position label/.style={below = 3pt, text height = 1.5ex, text depth = 1ex, shape=rectangle}
  ]
  \node {}
  child[subtree]{ node[draw] (u1) {$u$} }
  child[subtree]{ node[draw] (u2) {$u$} }
  child[subtree]{ node[draw] (u3) {$u$} }
  child[edge from parent/.style={dashed, draw}]{ node (+dots) {$\dots$} }
  child[edge from parent/.style={draw=none}]{ node {} }
  child[edge from parent/.style={dashed, draw}]{ node (-dots) {$\dots$} }
  child[subtree]{ node[draw] (u-3) {$u$} }
  child[subtree]{ node[draw] (u-2) {$u$} }
  child[subtree]{ node[draw] (u-1) {$u$} }
  ;
  \node[below=3pt of u1, position label] (l1) {$1$};
  \node[below=3pt of u2, position label] {$2$};
  \node[below=3pt of u3, position label] {$3$};
  \node[below=3pt of u-3, position label] {$-3$};
  \node[below=3pt of u-2, position label] {$-2$};
  \node[below=3pt of u-1, position label] (l-1) {$-1$};
  \draw[decoration={brace, mirror}, decorate]
  (l1.south west) -- node[position label] {$\omega$-part} (l1.south -| +dots.south east);
  \draw[decoration={brace}, decorate]
  (l-1.south east) -- node[position label] {$\omega^*$-part} (l-1.south -| -dots.south west);
  \end{tikzpicture}}\vspace*{-.5\baselineskip}
  \caption{\label{fig:uToPi}Representation of $u^{\omega + \omega^*}$}\vspace*{-2\baselineskip}
\end{wrapfigure}
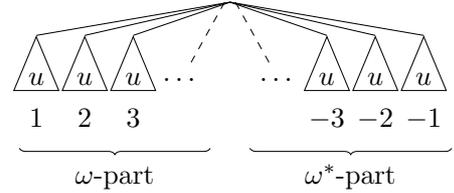

What happens if we consecutively factorize at a first/last $a$ is best understood if one considers
the structure of $\subs{(\alpha)^\pi}{\omega + \omega^*} = \subs{\alpha}{\omega + \omega^*}^{\omega +
\omega^*} = u^{\omega + \omega^*} = w$, which is schematically represented in \autoref{fig:uToPi}.

Suppose $u$ only contains a single $a$ and we start with the whole word $w_{(-\infty, +\infty)}$.
If we factorize on the first $a$ taking the part to the right, then we end up with the factor
$w_{(X_a(w; -\infty), +\infty)}$ with $X_a(w; -\infty) = (p, 1)$ where $p$ is the single
$a$-position in $u$. If we do this again, we obtain $w_{((p, 2), +\infty)}$. If we now factorize
on the next $a$ but take the part to the left, then we get $w_{((p, 2), (p, 3))}$. Notice that the
difference between $2$ and $3$ is $1$ and that there is no way of getting a (finite) difference
larger than one by factorizing on the respective first $a$. On the other hand, we can reach any
number in $\Nat$ as long as the right position is not in the $\omega$-part.

Notice that there is also no way of reaching $(p, -2)$ as left border without having $(q, -1)$ or
$(q, -2)$ as right border for a position $q \in \dom{u}$. These observations (and their
symmetrical duals) lead to the notion of \emph{normalizable} pairs of positions.
\begin{definition}
  Let $\gamma$ be a $\pi$-term and let $w = \subs{\gamma}{\omega + \omega^*}$.
  A pair $(l, r)$ of positions in $w$ such that $l$ is strictly smaller than $r$ is called
  \emph{normalizable} (with respect to $\gamma$) based on the following rules:
  \begin{itemize}
    \item Any pair is normalizable with respect to $\gamma = \varepsilon$ or $\gamma = a$ for an
          $a \in \Sigma$.
    \item $(-\infty, +\infty)$ is normalizable with respect to any $\pi$-term.
    \item If $\gamma = \alpha \beta$ for $\pi$-terms $\alpha$ and $\beta$, $l \in \dom{
          \subs{\alpha}{\omega + \omega^*}} \uplus \smallset{ -\infty }$ and $r \in \dom{\subs{\beta}{
          \omega + \omega^*}} \uplus \smallset{ +\infty }$, then $(l, r)$ is normalizable with respect
          to $\gamma$ if $(l, +\infty)$ is with respect to $\alpha$ and $(-\infty, r)$ is with respect
          to $\beta$.
    \item If $\gamma = \alpha \beta$ for $\pi$-terms $\alpha$ and $\beta$ and $l \in \dom{
          \subs{\alpha}{\omega + \omega^*}} \uplus \smallset{ -\infty }$ as well as $r \in \dom{
          \subs{\alpha}{\omega + \omega^*}}$ (or $l \in \dom{\subs{\beta}{\omega + \omega^*}}$
          as well as $r \in \dom{\subs{\beta}{\omega + \omega^*}} \uplus \smallset{ +\infty }$),
          then $(l, r)$ is normalizable with respect to $\gamma$ if it is with respect to $\alpha$
          (or $\beta$, respectively).
    \item If $\gamma = (\alpha)^\pi$ for a $\pi$-term $\alpha$, $l = (l', n)$ for $l' \in \dom{
          \subs{\alpha}{\omega + \omega^*}}$ and $n \in \Nat \uplus -\Nat$ and $r = +\infty$, then
          $(l, r)$ is normalizable with respect to $\gamma$ if $(l', +\infty)$ is with respect to
          $\alpha$ and $n$ is in $\Nat \uplus \smallset{ -1 }$.
    \item If $\gamma = (\alpha)^\pi$ for a $\pi$-term $\alpha$, $l = -\infty$, and $r = (r', m)$
          for $r' \in \dom{\subs{\alpha}{\omega + \omega^*}}$ and $m \in \Nat \uplus -\Nat$, then
          $(l, r)$ is normalizable with respect to $\gamma$ if $(-\infty, r')$ is with respect to
          $\alpha$ and $m$ is in $\smallset{ 1 } \uplus -\Nat$.
    \item If $\gamma = (\alpha)^\pi$ for a $\pi$-term $\alpha$, $l = (l', n)$ for $l' \in \dom{
          \subs{\alpha}{\omega + \omega^*}}$ and $n \in \Nat \uplus -\Nat$ and $r = (r', m)$ for $r'
          \in \dom{\subs{\alpha}{\omega + \omega^*}}$ and $m \in \Nat \uplus -\Nat$, then $(l, r)$ is
          normalizable with respect to $\gamma$ if
          \begin{itemize}
            \item $n \in \Nat$, $m \in -\Nat$ and $(l', +\infty)$ and $(-\infty, r')$ are
                  normalizable with respect to $\alpha$,
            \item $n, m \in \Nat$ or $n, m \in -\Nat$ and in both cases $m = n$ and
                  $(l', r')$ is normalizable with respect to $\alpha$, or
            \item $n, m \in \Nat$ or $n, m \in -\Nat$ and in both cases $m = n + 1$ and
                  $(l', +\infty)$ and $(-\infty, r')$ are normalizable with respect to $\alpha$.
          \end{itemize}
  \end{itemize}
\end{definition}
This definition looks cumbersome at first. All it does, however, is formalizing our previous
observations. This allows us to give a formal inductive proof that these observations hold for all positions
reachable by iterated first/last $a$ factorization. First, we need to extend our notation. For this,
we consider (abstract) positions $p \in \{ -\infty, +\infty \} \uplus \dom{w}$ in a word $w$ to be
implicitly linked to $w$. For a pair of positions $(l, r) \in \left( \{ -\infty \} \uplus \dom{w}
\right) \times \left( \dom{w} \uplus \{ +\infty \} \right)$ in an accessible word $w$, we write
\begin{align*}
  (l, r) \cdot X_a^L &= (l, X_a(w; l)) \text{,} &
  (l, r) \cdot X_a^R &= (X_a(w; l), r) \text{,} \\
  (l, r) \cdot Y_a^L &= (l, Y_a(w; r)) \text{ and } &
  (l, r) \cdot Y_a^R &= (Y_a(w; r), r)
\end{align*}
for all $a \in \alphabet{w_{(l, r)}}$. Note that we now have $w_{(l, r)} \cdot Z_a^D = w_{(l, r)
\cdot Z_a^D}$ for all $Z \in \{ X, Y \}$ and all $D \in \{ L, R \}$. In fact, we could have used this
as the definition previously. We also write $(l, r) \cdot C_{a, b}$ for $(l, r) \cdot X_a^L \cdot
Y_b^R$ if $X_a(w; l)$ is strictly larger than $Y_b(w; r)$. Therefore, we also have $w_{(l, r)} \cdot
C_{a, b} = w_{(l, r) \cdot C_{a, b}}$. As with words, we omit the $\cdot$ if we apply a sequence of
steps longer than one, e.\,g.\ we simply write $(l, r) \cdot X_a^L Y_b^R = (l, r) \cdot X_a^L \cdot
Y_b^R = (l, r) \cdot C_{a, b}$.

\begin{lemma}\label{lem:anyPairIsNormalizable}
  Let $\gamma$ be a $\pi$-term and let $w = \subs{\gamma}{\omega + \omega^*}$. Additionally, let
  $(l, r)$ be a normalizable (with respect to $\gamma$) pair of positions in $w$. Then the pairs
  \[
    (l, r) \cdot X^L_a, (l, r) \cdot X^R_a, (l, r) \cdot Y^L_a \text{ and } (l, r) \cdot Y^R_a
  \]
  are normalizable with respect to $\gamma$ for any $a \in \alphabet(w_{(l, r)})$.

  Therefore, $(-\infty, +\infty) \cdot F_1 F_2 \dots F_n$ is normalizable with respect to $\gamma$
  for any $F_1,\allowbreak F_2,\allowbreak \dots,\allowbreak F_n \in \{ X_a^L, X_a^R, Y_a^L, Y_a^R,
  C_{a, b} \mid a, b \in \Sigma \}$ (if it is defined).
\end{lemma}
\begin{proof}
  As the cases for $Y_a^L$ and $Y_a^R$ are symmetrical, we only show those for $X_a^L$ and $X_a^R$.
  Let $p = X_a(w; l)$ for an $a \in \alphabet(w_{(l, r)})$. Clearly, we have $l <_\mu p <_\mu r$,
  where $\mu$ is the order type of $w$, and we need to show that $(l, p)$ and $(p, r)$ are
  normalizable. For this, we proceed by induction on the structure of $\gamma$. The base case
  $\gamma = \varepsilon$ or $\gamma \in \Sigma$ is trivial.

  \begin{case}{$\gamma = \alpha \beta$}
  Define $u = \subs{\alpha}{\omega + \omega^*}$ and $v = \subs{\beta}{\omega + \omega^*}$. For $l \in
  \dom{u} \uplus \smallset{ -\infty }$ and $r \in \dom{u}$ we have $p \in \dom{u}$ as well.
  Additionally, $(l, r)$ needs to be normalizable with respect to $\alpha$ by the definition of
  normalizability and we can apply induction. The same argument, but on $\beta$, works for $l \in
  \dom{v}$ and $r \in \dom{v} \uplus \smallset{ +\infty }$. For $l \in \dom{u} \uplus \smallset{
  -\infty }$ and $r \in \dom{v} \uplus \smallset{ +\infty }$ we know that
  $(l, +\infty)$ is normalizable with respect to $\alpha$ and $(-\infty, r)$ is with respect to
  $\beta$ by the definition of normalizablity. If $p \in \dom{u}$, then $(p, +\infty) = (l, +\infty)
  \cdot X_a^R$ and $(l, p) = (l, +\infty) \cdot X_a^L$. Induction yields normalizability with
  respect to $\alpha$ for both and, by the definition of normalizability, we have that $(p, r)$ and
  $(l, p)$ are normalizable with respect to $\gamma$. For $p \in \dom{v}$, we can apply a similar
  argument, as then $(-\infty, p) = (-\infty, +\infty) \cdot X_a^L$ and $(p, r) = (-\infty, r) \cdot
  X_a^R$ are normalizable with respect to $\beta$.
  \end{case}

  \begin{case}{$\gamma = (\alpha)^\pi$}
  Define $u = \subs{\alpha}{\omega + \omega^*}$ and let $p = (p', k)$. If $l = (l', n)$ for an $n
  \in \Nat \uplus -\Nat$ and $r = +\infty$, then, by the definition of normalizability, we have
  that $(l', +\infty)$ is normalizable with respect to $\alpha$ and $n \in \Nat \uplus
  \smallset{ -1 }$. There are two cases: for $k = n \in \Nat \uplus \smallset{ -1 }$ we know that $p'
  = X_a(u; l')$ and, by induction, that $(l', p'), (p', +\infty)$ are normalizable with respect to
  $\alpha$. This yields the normalizability with respect to $\gamma$ of $(l, p)$ and $(p, +\infty)$.
  For $k = n + 1$ we know that $n \neq -1$ and, therefore, that $n, k \in \Nat$. We also have $p' =
  X_a(u; -\infty)$ and, thus, that $(-\infty, p')$ and $(p', +\infty)$ are normalizable with respect
  to $\alpha$ by induction. By definition, $(p, +\infty)$ and $(l, p)$ are normalizable with respect
  to $\gamma$ then. Note that $k$ cannot have any other value than $n$ or $n + 1$ since otherwise it
  could not be the smallest $a$-position to the right of $l$.

  If $l = -\infty$ and $r = (r', m)$, then $k = 1$, and $p' = X_a(u; -\infty)$, which yields
  $(-\infty, p') = (-\infty, +\infty) \cdot X_a^L$ and $(p', +\infty) = (-\infty, +\infty) \cdot
  X_a^R$. By induction, both of these pairs are normalizable with respect to $\alpha$ and, by
  definition of the normalizability, $(-\infty, p)$ is normalizable with respect to $\gamma$.
  Furthermore in this case, we know that $(-\infty, r')$ is normalizable with respect to $\alpha$
  and that $m$ is in $\smallset{ 1 } \uplus -\Nat$. For $m \in -\Nat$, this shows the
  normalizability with respect to $\gamma$ of $(p, r)$. For $m = 1$, we have $(p', r') = (-\infty,
  r') \cdot X_a^R$ and, by induction, its normalizability with respect to $\alpha$. This yields that
  $(p, r)$ is normalizable with respect to $\gamma$.

  If $l = (l', n)$ and $r = (r', m)$ for $n \in \Nat$ and $m \in -\Nat$, we know that
  $(l', +\infty)$ and $(-\infty, r')$ are normalizable with respect to $\alpha$. For $k = n \in
  \Nat$, we also know that $p' = X_a(u; l')$ and, therefore, that $(l', p') = (l', +\infty) \cdot
  X_a^L$ and $(p', +\infty) = (l', +\infty) \cdot X_a^R$ are normalizable with respect to $\alpha$ by
  induction. Then, by definition, $(l, p)$ and $(p, r)$ are normalizable with respect to $\gamma$.
  For $k = n + 1 \in \Nat$ we have that $p' = X_a(u; -\infty)$ and, therefore, the normalizability
  with respect to $\alpha$ of $(-\infty, p') = (-\infty, +\infty) \cdot X_a^L$ and $(p', +\infty) =
  (-\infty, +\infty) \cdot X_a^R$ by induction. This yields the normalizability with respect to
  $\gamma$ of $(l, p)$ and $(p, r)$.

  Finally, if $l = (l', n)$ and $r = (r', m)$ for $n, m \in \Nat$ or $n, m \in -\Nat$, we know that
  $0 \leq m - n \leq 1$. Because $p$ must be in between $l$ and $r$, $n = m$ also implies $n = m =
  k$ and that $p'$ is in between $l'$ and $r'$ as well as $p' = X_a(u; l')$. In that case, we have
  that $(l', r')$ and, by induction, also $(l', p') = (l', r') \cdot X_a^L$ and $(p', r') = (l', r')
  \cdot X_a^R$ are normalizable with respect to $\alpha$. This yields the normalizability with
  respect to $\gamma$ of $(l, p)$ and $(p, r)$. For $m = n + 1$, we know that $(l', +\infty)$ and
  $(-\infty, r')$ are normalizable with respect to $\alpha$. Moreover, there are only
  two cases: $k = n$ and $k = m$. In the former case, we have $p' = X_a(u; l')$ and the
  normalizability with respect to $\alpha$ of $(l', p') = (l', +\infty) \cdot X_a^L$ and $(p',
  +\infty) = (l', +\infty) \cdot X_a^R$ by induction, which yields the normalizability of $(l, p)$
  and $(p, r)$ with respect to $\gamma$. In the latter case, we have $p' = X_a(u; -\infty)$ and the
  normalizability with respect to $\alpha$ of $(-\infty, p') = (-\infty, +\infty) \cdot X_a^L$ and
  $(p', r') = (-\infty, r') \cdot X_a^R$, which yields the normalizability with respect to $\gamma$
  of $(l, p)$ and $(p, r)$.
  \end{case}
\end{proof}

The choice of words indicates that normalizability of a pair $(l, r)$ can be used to
define a normalization. Before we give a formal -- unfortunately, quite technical -- definition of
this, we describe its idea informally.
Let us refer back to the schematic representation of $\subs{(\alpha)^\pi}{\omega +
\omega^*} = w$ as given in \autoref{fig:uToPi}. Basically, there are three different cases for relative positions of the left
border $l$ and the right border $r$ which describe the factor $w_{(l, r)}$:
\begin{enumerate}
  \item $l$ is in the $\omega$-part and $r$ is in the $\omega^*$-part,
  \item $l$ and $r$ are either both in the $\omega$-part or both in the $\omega^*$-part and
        have the same value there, or
  \item $l$ and $r$ are either both in the $\omega$-part or both in the $\omega^*$-part but
        $r$ has a value exactly larger by one than $l$.
\end{enumerate}
This is ensured by the normalizability of $(l, r)$. Now, in the first case, we can safely move
$l$ to value $1$ (the first position) and $r$ to value $-1$ (the last position) without changing
the described factor. In the second and third case, we can move $l$ and $r$ to any value -- as long
as we retain the difference between the values -- without changing the described factor. Here, we
move them to the left-most values (which are $1, 1$ or $1, 2$). Afterwards, we go on recursively.

Unfortunately, things get a bit more complicated because $l$ might be $-\infty$ and $r$ might be
$+\infty$. In these cases, we normalize to the left-most or right-most value without changing the
factor.

For concatenations of $\pi$-terms, we have a similar situation: either $l$ and $r$ belong both to
the left or to the right factor, in which case we can continue by normalization with respect to
that, or $l$ belongs to the left factor and $r$ belongs to the right one. In this case, we have
to continue the normalization with $(l, +\infty)$ and $(-\infty, r)$ in the respective concatenation
parts, as this ensures that the described factor remains unchanged.

Formalizing these ideas results in the following inductive definition.
\begin{definition}
  Let $\gamma$ be a $\pi$-term, $w = \subs{\gamma}{\omega + \omega^*}$ and $(l, r)$ a normalizable
  pair of positions in $w$. The normalized pair $\overline{(l, r)}^\gamma = (\bar{l}, \bar{r})$ with
  respect to $\gamma$ is defined recursively:
  \begin{itemize}
    \item For $\gamma = \varepsilon$ or $\gamma = a \in \Sigma$ define $\bar{l} = l$ and $\bar{r} =
          r$.
    \item If $\gamma = \alpha \beta$ for $\pi$-terms $\alpha$ and $\beta$, $l \in \dom{
          \subs{\alpha}{\omega + \omega^*}} \uplus \smallset{ -\infty }$ and $r \in \dom{\subs{\beta}{
          \omega + \omega^*}} \uplus \smallset{ +\infty }$, then define $\bar{l}$ as the first component
          of $\overline{(l, +\infty)}^\alpha$ and $\bar{r}$ as the second component of
          $\overline{(-\infty, r)}^\beta$.
    \item If $\gamma = \alpha \beta$ for $\pi$-terms $\alpha$ and $\beta$ and $l \in \dom{
          \subs{\alpha}{\omega + \omega^*}} \uplus \smallset{ -\infty }$ as well as $r \in \dom{
          \subs{\alpha}{\omega + \omega^*}}$ (or $l \in \dom{\subs{\beta}{\omega + \omega^*}}$
          as well as $r \in \dom{\subs{\beta}{\omega + \omega^*}} \uplus \smallset{ +\infty }$),
          then define $(\bar{l}, \bar{r}) = \overline{(l, r)}^\alpha$ (or $(\bar{l}, \bar{r}) =
          \overline{(l, r)}^\beta$, respectively).
    \item If $\gamma = (\alpha)^\pi$ for a $\pi$-term $\alpha$, then:
          \begin{itemize}
            \item if $l = -\infty$, define $\bar{l} = -\infty$,
            \item if $r = +\infty$, define $\bar{r} = +\infty$,
            \item if $l = (l', n)$ and $r = +\infty$, define $\bar{l} = (\bar{l'}, \bar{n})$ with
                  $\bar{l'}$ given by the first component of $\overline{(l', +\infty)
                  }^\alpha$ and $\bar{n}$ given by
                  $$
                    \bar{n} = \begin{cases}
                      1 & \text{if } n \in \Nat\\
                      -1 & \text{if } n = -1 \text{,}
                    \end{cases}
                  $$
            \item if $l = -\infty$ and $r = (r', m)$, define $\bar{r} = (\bar{r'}, \bar{m})$ with
                  $\bar{r'}$ given by the second component of $\overline{(-\infty, r')
                  }^\alpha$ and $\bar{m}$ given by
                  $$
                    \bar{m} = \begin{cases}
                      1 & \text{if } n = 1\\
                      -1 & \text{if } n \in -\Nat \text{,}
                    \end{cases}
                  $$
            \item if $l = (l', n)$ and $r = (r', m)$ with $n \in \Nat$ and $m \in -\Nat$, define
                  $\bar{l} = (\bar{l'}, 1)$ with $\bar{l'}$ being by the first
                  component of $\overline{(l', +\infty)}^\alpha$ and define
                  $\bar{r} = (\bar{r'}, -1)$ with $\bar{r'}$ given by the second
                  component of $\overline{(-\infty, r')}^\alpha$,
            \item if $l = (l', n)$ and $r = (r', m)$ with $n = m$, define
                  $\bar{l} = (\bar{l'}, \bar{n})$ and $\bar{r} = (\bar{r'}, \bar{m})$ with
                  $(\bar{l'}, \bar{r'}) = \overline{(l', r')}^\alpha$ and $\bar{n} = \bar{m} = 1$,
                  and
            \item if $l = (l', n)$ and $r = (r', m)$ with $m = n + 1$, define
                  $\bar{l} = (\bar{l'}, \bar{n})$ and $\bar{r} = (\bar{r'}, \bar{m})$ with
                  $\bar{l'}$ given by the first component of $\overline{(l', +\infty)
                  }^\alpha$, $\bar{r'}$ given by the second component of $\overline{
                  (-\infty, r')}^\alpha$, $\bar{n} = 1$ and $\bar{m} = \bar{n} + 1 = 2$.
          \end{itemize}
  \end{itemize}
\end{definition}

One should note that if we normalize a normalizable pair $(l, r)$, then the resulting
pair is normalizable itself. Indeed, if we normalize an already normalized pair again, we do not
change any values. Formally, this can be proved by an induction on the structure of the $\pi$-term.
As an example for such an induction, we prove the following lemma which states that normalizing
a pair of positions does not change the described factor.

\begin{lemma}\label{lem:normalizationMaintainsFactor}
  Let $\gamma$ be a $\pi$-term and let $(l, r)$ be a normalizable pair of positions in $w =
  \subs{\gamma}{\omega + \omega^*}$.
  Then
  $$
    w_{(l, r)} = w_{\overline{(l, r)}^\gamma}
  $$
  holds.
\end{lemma}
\begin{proof}
  Define $\overline{(l, r)}^\gamma = (\bar{l}, \bar{r})$ and proceed by induction on the structure of
  $\gamma$. The base cases for $\gamma = \varepsilon$ and $\gamma \in \Sigma$ are trivial.

  If $\gamma = \alpha \beta$ for $\pi$-terms $\alpha$ and $\beta$, then define $u = \subs{\alpha}{
  \omega + \omega^*}$ and $v = \subs{\beta}{\omega + \omega^*}$. If $l \in \dom{u} \uplus \smallset{
  -\infty }$ and $r \in \dom{v} \uplus \smallset{ +\infty }$, then
  $$
    w_{(l, r)} = u_{(l, +\infty)} v_{(-\infty, r)} =
      u_{\overline{(l, +\infty)}^\alpha} v_{\overline{(-\infty, r)}^\beta} =
      w_{\overline{(l, r)}^\gamma} \text{.}
  $$
  If $l \in \dom{u} \uplus \smallset{ -\infty }$ and $r \in \dom{u}$, then
  $$
    w_{(l, r)} = u_{(l, r)} = u_{\overline{(l, r)}^\alpha} = w_{\overline{(l, r)}^\gamma}
      \text{.}
  $$
  The case $l \in \dom{v}$ and $r \in \dom{v} \uplus \smallset{ +\infty }$ is symmetrical.

  If $\gamma = (\alpha)^\pi$ for a $\pi$-term $\alpha$, then define $u = \subs{\alpha}{\omega +
  \omega^*}$. The case $l = -\infty$ and $r = +\infty$ is trivial. If $l = (l', n)$ for an $n \in
  \Nat \uplus -\Nat$ and $r = +\infty$, define $\bar{l'}$ by $\overline{(l', +\infty)}^\alpha =
  (\bar{l'}, +\infty)$. For $n \in \Nat$ we then have
  \begin{align*}
    w_{(l, r)} &= w_{((l', n), +\infty)} =
      \left( u^{\omega + \omega^*} \right)_{((l', n), +\infty)} = \left( u^{\omega + \omega^*} \right)_{((l', 1), +\infty)}
    \intertext{because of $u^{\omega + \omega^*} = u u^{\omega + \omega^*}$ and further}
    w_{(l, r)} &= u_{(l', +\infty)} u^{\omega + \omega^*} =
      u_{\overline{(l', +\infty)}^\alpha} u^{\omega + \omega^*} =
      u_{(\bar{l'}, +\infty)} u^{\omega + \omega^*}\\
    &= \left( u^{\omega + \omega^*} \right)_{((\bar{l'}, 1), +\infty)} =
      w_{\overline{((l', n), +\infty)}^\gamma} = w_{\overline{(l, r)}^\gamma}
  \end{align*}
  and for $n = -1$ -- the only remaining case -- we have
  \begin{align*}
    w_{(l, r)} &= w_{((l', -1), +\infty)} = u_{(l', +\infty)} =
      u_{\overline{(l', +\infty)}^\alpha} = u_{(\bar{l'}, +\infty)} =
      w_{((\bar{l'}, -1), +\infty)} =
      w_{\overline{(l, r)}^\gamma}
    \text{.}
  \end{align*}
  The case for $l = -\infty$ and $r = (r', m)$ is symmetrical.

  Therefore, we can assume $l = (l', n)$ and $r = (r', m)$. The case $n \in \Nat$ and $m \in -\Nat$
  is proved by a calculation similar to the one given above. For $n = m$ we have
  $$
    w_{(l, r)} = w_{((l', n), (r', n))} = u_{(l', r')} = u_{\overline{(l', r')}^\alpha} =
      w_{\overline{(l, r)}^\gamma}
  $$
  and for $m = n + 1$ we have
  \[
    w_{(l, r)} = u_{(l', +\infty)} u_{(-\infty, r')} =
      u_{\overline{(l', +\infty)}^\alpha} u_{\overline{(-\infty, r')}^\alpha} =
      w_{\overline{(l, r)}^\gamma} \text{.}\qedhere
  \]
\end{proof}

Another observation is crucial for the proof of the decidability: after normalizing a pair $(l, r)$
the values belonging to the $\omega + \omega^*$ parts for the two positions are all
in $\smallset{ 1, 2, -2, -1 }$. But: there are only finitely many such positions in any word $w =
\subs{\gamma}{\omega + \omega^*}$ for a $\pi$-term $\gamma$. Because the normalization preserves the
described factor, this means that there are only finitely many factors which can result from a
sequence of first/last $a$ factorizations.

Plugging all these ideas and observations together yields a proof for the next theorem.
\begin{theorem}\label{thm:RelationsAreDecidable}
  For two $\pi$-terms $\alpha$ and $\beta$ over the same alphabet $\Sigma$, it is decidable for all
  $m \in \mathbb{N}_0$ and all $Z \in \{ X, Y, XY, \WI \}$ whether
  \[
    \subs{\alpha}{\omega + \omega^*} \equiv_m^Z \subs{\beta}{\omega + \omega^*}
  \]
  holds.
  Furthermore, it is decidable whether
  \[
    \forall m \in \mathbb{N}: \subs{\alpha}{\omega + \omega^*} \equiv_m^{XY} \subs{\beta}{\omega + \omega^*}
  \]
  holds.
\end{theorem}
\begin{proof}
  We decide whether $u = \subs{\alpha}{\omega + \omega^*} \equiv_m^Z \subs{\beta}{\omega + \omega^*}
  = v$ holds by trying to find a factorization sequence $F_1 F_2 \dots F_n$ with $F_1, F_2, \dots,
  F_n \in \{ X_a^L, X_a^R, Y_a^L, Y_a^R,\allowbreak C_{a, b} \mid a, b \in \Sigma \}$ that can be
  applied to $u$ but not to $v$ (or vice versa), i.\,e.\ we rather try to decide $u \not\equiv_m^Z
  v$ instead.

  Which sequences need to be tested depends on the actual relation. For example, for
  $\equiv_m^X$, we can start with arbitrary many $X_a^R$ factorizations but, as soon as we apply
  an $X_a^L$, $Y_a^L$ or $Y_a^R$ factorization, we have changed the direction and we have to decrease
  $m$ by one. If we did this because of a $Y_a^D$ factorization (with $D \in \{ R, L \}$), we also
  need to switch to a $\equiv_{m - 1}^Y$ mode, in which case we can continue with arbitrary many $Y_a^L$
  factorizations while any other factorization decreases the remaining number of direction changes
  and might also switch back to a $\equiv_{m - 2}^X$ mode. If we want to test $\equiv_m^\WI$, we
  also need to allow $C_{a, b}$ factorizations and we need to count the number of direction changes
  appropriately. For testing $\equiv_m^{XY}$ for all $m \in \mathbb{N}$, the situation is simpler:
  here, we do not need to keep track of the remaining direction changes as there always is an
  arbitrary number of them left. Clearly, we can construct a deterministic finite automaton for any
  of the relations which accepts exactly those sequences $F_1 F_2 \dots F_n$ which need to be tested.

  To test whether a factorization sequence can be applied on $u$, we construct an additional
  deterministic finite automaton. The states of this automaton are the normalized pairs of positions
  in $u$ (of which there only finitely many, as discussed above). The initial state is $(-\infty,
  +\infty)$ and all states are final. This does not result in a trivial automaton because it will
  not be complete in general. We have an $F$-labeled transition with $F \in \{ X_a^D, Y_a^D, C_{a,
  b} \mid a, b \in \Sigma, D \in \{ L, R \} \}$ from $(l, r)$ to $\overline{(l, r) \cdot F}^\alpha$
  if $(l, r) \cdot F$ is defined. Note that normalization does not change the implicitly stored
  factor of $u$ by \autoref{lem:normalizationMaintainsFactor}. This automaton, by construction,
  accepts exactly those factorization sequences which can be applied to $u$.

  Using the same construction, we can also get such an automaton for $v$. We intersect both automata
  with the one which accepts the relevant factorization sequences. For the resulting pair of
  automata, we check the symmetric difference of the accepted languages for emptiness. It is empty
  if any only if $u$ and $v$ are in relation.
\end{proof}

Together with \autoref{thm:relationsAndEquations}, this gives the following decidability result.
Note that previous partial results exist: Almeida proved decidability for $\V[J]$ \cite{Almeida1991Implicit:short,almeida2002finite}, Almeida and Zeitoun proved it for $\V[R]$ \cite{almeida2007automata} and Moura for $\DA$ \cite{moura2011word}.
\begin{corollary}
  \label{cor:wordProblemsAreDecidable}
  The word problems for $\pi$-terms over $\Rm$, $\Lm$, $\Rm \vee \Lm$ and $\Rm \cap \Lm$ are
  decidable for any $m \in \Nat$. Moreover, the word problem for $\pi$-terms over $\DA$ is decidable.
\end{corollary}

\section{Nondeterministic Logarithmic Space}

In the presented algorithm, we construct for a $\pi$-term $\gamma$ over the alphabet $\Sigma$ a
deterministic finite automaton which accepts exactly those factorization sequences which can be
applied to $\subs{\gamma}{\omega + \omega^*} = w$. For this, we had to store normalized pairs $(l,
r)$ of positions in $w$ and we had to compute $\overline{(l, r) \cdot F}^\gamma$ for a factorization
$F \in \{ X_a^D, Y_a^D, C_{a, b} \mid a, b \in \Sigma, D \in \{ L, R \} \}$. In this section, we will
show that both can be done by a deterministic Turing machine in logarithmic space. Afterwards, we
will show that this yields membership of the word problems for $\pi$-terms to the class of
problems which can be solved by nondeterministic Turing machines within logarithmic space.

\paragraph*{Observations and Ideas.}
We will start by having a close look on how to store a position $p$ (and, therefore, a pair of
positions) in the word $w = \subs{\gamma}{\omega + \omega^*}$ on a Turing machine. To do this,
we first store the position in $\gamma$ which corresponds to $p$ in $w$; this, basically, is a
simple pointer. Additionally, we need to store to which value in $\Nat \uplus -\Nat$ the position $p$
belongs for all relevant $\pi$-exponents; this, we can do by storing a pointer to the
$\pi$-position in $\gamma$ together with the corresponding value. If $p$ is part of a normalized
pair of positions, the values can only be in $\smallset{ 1, 2, -2, -1 }$, which is a finite
set and, thus, needs only finite information.

\begin{example}\label{exmpl:positionInASubstPiTerm}
  Have a look at the $\pi$-term $\gamma = a \left( b (c)^\pi \right)^\pi a (b)^\pi c$. The word
  $\subs{\gamma}{\omega + \omega^*}$ can be represented by the following tree:
  \begin{center}
    \begin{tikzpicture}
      \tikzstyle{level 1}=[sibling distance=15ex]
      \tikzstyle{level 2}=[sibling distance=5ex]
      \tikzstyle{level 3}=[sibling distance=3ex]
      \node {}
        child{ node {$a$} }
        child[very thick, every child/.style=thin]{
          child{ node {$b$} }
          child{
            child{ node {$c$} }
            child[edge from parent/.style={dashed, draw}]{ node {$\dots$} }
            child[edge from parent/.style={draw=none}]{ node {} }
            child[edge from parent/.style={dashed, draw}]{ node {$\dots$} }
            child{ node {$c$} }
          }
          child[edge from parent/.style={dashed, draw}]{ node {$\dots$} }
          child[edge from parent/.style={draw=none}]{ node {} }
          child[edge from parent/.style={dashed, draw}]{ node {$\dots$} }
          child{ node {$b$} }
          child[very thick, every child/.style=thin]{
            child{ node {$c$} }
            child[very thick]{ node {$c$} }
            child[edge from parent/.append style={dashed}]{ node {$\dots$} }
            child[edge from parent/.style={draw=none}]{ node {} }
            child[edge from parent/.append style={dashed}]{ node {$\dots$} }
            child{ node {$c$} }
          }
        }
        child{ node {$a$} }
        child[every child/.style={sibling distance=3ex}]{
          child{ node {$b$} }
          child[edge from parent/.style={dashed, draw}]{ node {$\dots$} }
          child[edge from parent/.style={draw=none}]{ node {} }
          child[edge from parent/.style={dashed, draw}]{ node {$\dots$} }
          child{ node {$b$} }
        }
        child{ node {$c$} }
      ;
    \end{tikzpicture}
  \end{center}\vspace{1cm}
  The highlighted position can also be represented this way:\nopagebreak
  \begin{center}
    \begin{tikzpicture}[every node/.style={text height=.8em, text depth=.2em}]
      \matrix [rectangle, draw, matrix of math nodes, ampersand replacement=\&,
               inner sep=2pt] (gamma) {
        a \& ( \& b \& ( \& c \& ) \& {}^\pi \& ) \& {}^\pi \& a \& ( \& b \& ) \& {}^\pi \& c \\
      };
      \draw[-latex] ($(gamma-1-5.south) - (0, 0.75cm)$) -- ($(gamma-1-5.south) - (0, 0.2cm)$);
      \node[below=0.2cm of gamma-1-7] {$2$};
      \node[below=0.2cm of gamma-1-9] {$-1$};
    \end{tikzpicture}
  \end{center}
  Note that the last $\pi$-position does not have a value since it is not relevant for the
  position. Also note that storing the values on a second tape directly under the
  $\pi$-position would require linear space (for a position belonging to a normalized pair).
  We store them in a position/value list which will turn out to be more efficient after some
  modifications.
\end{example}

Unfortunately, this approach still requires at least linear space. But for normalized pairs, we can
optimize it further if we look at the definition of the normalization. At some point the two
positions $l$ and $r$ in the normalized pair $(l, r)$ branch since $l$ is a smaller position than
$r$. This can either happen because $l = -\infty$ or $r = +\infty$, because there are
sub-$\pi$-terms $\alpha$ and $\beta$ and $l$ belongs to $\subs{\alpha}{\omega + \omega^*}$ while
$r$ belongs to $\subs{\beta}{\omega + \omega^*}$, or because of a sub-$\pi$-term of the form
$(\alpha)^\pi$ where $l$ has a different value compared to $r$. In the former two cases, we call
the pair \emph{c-branching}\footnote{The \enquote{c} is for \emph{concatenation}.} and in the
latter case \emph{$\pi$-branching}. Whichever is the case, we know that
the values of $l$ and $r$ belonging to hierarchically higher $\pi$-positions are always equal
and, by definition of the normalization, are equal to $1$. Thus, we do not need to store these
values explicitly; instead, we are going to store the branching position.

If the branching position is a $\pi$-position (i.\,e.~the pair is $\pi$-branching), then, for
a normalized pair, the values of $l$ and $r$ at this position can, by definition, only be $1$ for
$l$ and $2$ for $r$ or $1$ for $l$ and $-1$ for $r$. This information can be stored alongside
the branching position in constant space.

To store the values hierarchically below the branching position, we need to have an even closer look
at normalized position pairs. Before we do this, however, it is convenient to define four position
\emph{forms}: a position is in
$+$-form if all its values for relevant $\pi$-positions are from $\Nat$ and it is in $-$-form
if they are from $-\Nat$. Positions in $\mp$-form may only have values from $\Nat$ for
$\pi$-positions which are hierarchically lower than the first $\pi$-position with a value
from $\Nat$ and positions in $\pm$-form are defined symmetrically. More formally, we define:
\begin{definition}
  A position $p \in \dom{\subs{\gamma}{\omega + \omega^*}}$ for a $\pi$-term $\gamma$ is in
  \begin{itemize}
    \item $+$-form (with respect to $\gamma$) if
          \begin{itemize}
            \item $\gamma = \varepsilon$ or $\gamma \in \Sigma$,
            \item $\gamma = \alpha \beta$ for $\pi$-terms $\alpha$ and $\beta$ and $p$ is in
                  $+$-form with respect to its respective sub-$\pi$-term, or
            \item $\gamma = (\alpha)^\pi$ for a $\pi$-term $\alpha$, $p = (p', n)$ with $n \in
                  \Nat$ and $p'$ is in $+$-form with respect to $\alpha$,
          \end{itemize}\pagebreak[2]
    \item $-$-form (with respect to $\gamma$) if
          \begin{itemize}
            \item $\gamma = \varepsilon$ or $\gamma \in \Sigma$,
            \item $\gamma = \alpha \beta$ for $\pi$-terms $\alpha$ and $\beta$ and $p$ is in
                  $-$-form with respect to its respective sub-$\pi$-term, or
            \item $\gamma = (\alpha)^\pi$ for a $\pi$-term $\alpha$, $p = (p', n)$ with $n \in
                  -\Nat$ and $p'$ is in $-$-form with respect to $\alpha$,
          \end{itemize}
    \item $\mp$-form (with respect to $\gamma$) if
          \begin{itemize}
            \item $\gamma = \varepsilon$ or $\gamma \in \Sigma$,
            \item $\gamma = \alpha \beta$ for $\pi$-terms $\alpha$ and $\beta$ and $p$ is in
                  $\mp$-form with respect to its respective sub-$\pi$-term,
            \item $\gamma = (\alpha)^\pi$ for a $\pi$-term $\alpha$, $p = (p', n)$ with $n \in
                  \Nat$ and $p'$ has $+$-form, or
            \item $\gamma = (\alpha)^\pi$ for a $\pi$-term $\alpha$, $p = (p', n)$ with $n \in
                  -\Nat$ and $p'$ has $\mp$-form, or
          \end{itemize}
    \item $\pm$-form (with respect to $\gamma$) if
          \begin{itemize}
            \item $\gamma = \varepsilon$ or $\gamma \in \Sigma$,
            \item $\gamma = \alpha \beta$ for $\pi$-terms $\alpha$ and $\beta$ and $p$ is in
                  $\pm$-form with respect to its respective sub-$\pi$-term,
            \item $\gamma = (\alpha)^\pi$ for a $\pi$-term $\alpha$, $p = (p', n)$ with $n \in
                  -\Nat$ and $p'$ has $-$-form, or
            \item $\gamma = (\alpha)^\pi$ for a $\pi$-term $\alpha$, $p = (p', n)$ with $n \in
                  \Nat$ and $p'$ has $\pm$-form.
          \end{itemize}
  \end{itemize}
\end{definition}
\begin{example}
  The position
  \begin{center}
    \begin{tikzpicture}[every node/.style={text height=.8em, text depth=.2em}]
      \matrix [rectangle, draw, matrix of math nodes, ampersand replacement=\&,
      inner sep=2pt] (gamma) {
        a \& ( \& b \& ( \& c \& ) \& {}^\pi \& ) \& {}^\pi \& a \& ( \& b \& ) \& {}^\pi \& c \\
      };
      \draw[-latex] ($(gamma-1-5.south) - (0, 0.75cm)$) -- ($(gamma-1-5.south) - (0, 0.2cm)$);
      \node[below=0.2cm of gamma-1-7] {$2$};
      \node[below=0.2cm of gamma-1-9] {$-1$};
    \end{tikzpicture}
  \end{center}
  from the previous example is in $\mp$-form but not in any of the other three forms.

  The similar position represented by
  \begin{center}
    \begin{tikzpicture}[every node/.style={text height=.8em, text depth=.2em}]
      \matrix [rectangle, draw, matrix of math nodes, ampersand replacement=\&,
      inner sep=2pt] (gamma) {
        a \& ( \& b \& ( \& c \& ) \& {}^\pi \& ) \& {}^\pi \& a \& ( \& b \& ) \& {}^\pi \& c \\
      };
      \draw[-latex] ($(gamma-1-5.south) - (0, 0.75cm)$) -- ($(gamma-1-5.south) - (0, 0.2cm)$);
      \node[below=0.2cm of gamma-1-7] {$2$};
      \node[below=0.2cm of gamma-1-9] {$3$};
    \end{tikzpicture}
  \end{center}
  is in $+$-form, in $\mp$-form and in $\pm$-form but it is not in $-$-form.
\end{example}

What use are these definitions for our goal of storing positions efficiently? If we know that a
position is in $+$-form or $-$-form, then we do not need to store whether a value is from $\Nat$ or
from $-\Nat$. Similarly, if a position is in $\mp$-form or in $\pm$-form, then we only need to store
one potential $\pi$-position at which the values switch form $-\Nat$ to $\Nat$ or vice versa.
While this does not seem to be a huge gain since we still need to store the actual value, it will
turn out to be crucial later on.

Next, we need to make some further observations. Consider a $\pi$-term $\gamma$ and define $w =
\subs{\gamma}{\omega + \omega^*}$. If we start with a position $p \in \dom{w}$ in $\mp$-form and
we go to the next $a$ on the right (i.\,e.~we compute $X_a(w; p)$), then the follow-up position
$p'$ is in $\mp$-form as well because strictly after the branching of $p$ and $p'$ all values for
relevant $\pi$-position have to be $1$ (since otherwise there would already have been an $a$
before). By symmetry, if we start in $\pm$-form and go to the previous $a$ on the left, the resulting
position will also be in $\pm$-form.
\begin{example}
  Look again at the $\pi$-term $\gamma = a \left( b (c)^\pi \right)^\pi a (b)^\pi c$. Suppose we
  are in the position $p$ in $w = \subs{\gamma}{\omega + \omega^*}$ and advance to $p' = X_b(w;
  p)$:
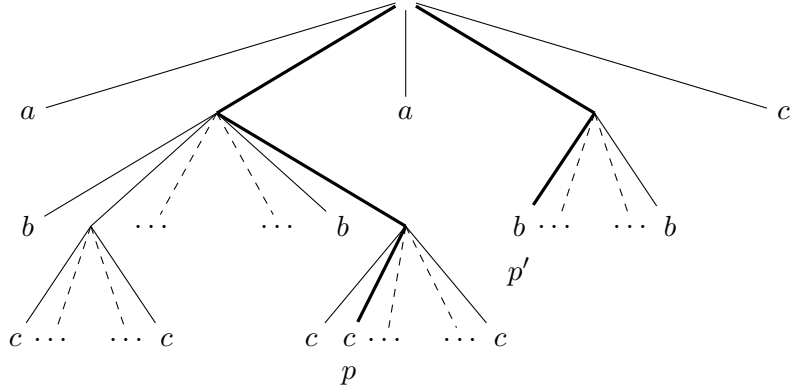
  \begin{center}
    \begin{tikzpicture}
      \tikzstyle{level 1}=[sibling distance=15ex]
      \tikzstyle{level 2}=[sibling distance=5ex]
      \tikzstyle{level 3}=[sibling distance=3ex]
      \node {}
      child{ node {$a$} }
      child[very thick, every child/.style=thin]{
        child{ node {$b$} }
        child{
          child{ node {$c$} }
          child[edge from parent/.style={dashed, draw}]{ node {$\dots$} }
          child[edge from parent/.style={draw=none}]{ node {} }
          child[edge from parent/.style={dashed, draw}]{ node {$\dots$} }
          child{ node {$c$} }
        }
        child[edge from parent/.style={dashed, draw}]{ node {$\dots$} }
        child[edge from parent/.style={draw=none}]{ node {} }
        child[edge from parent/.style={dashed, draw}]{ node {$\dots$} }
        child{ node {$b$} }
        child[very thick, every child/.style=thin]{
          child{ node {$c$} }
          child[very thick]{ node (p) {$c$} }
          child[edge from parent/.append style={dashed}]{ node {$\dots$} }
          child[edge from parent/.style={draw=none}]{ node {} }
          child[edge from parent/.append style={dashed}]{ node {$\dots$} }
          child{ node {$c$} }
        }
      }
      child{ node {$a$} }
      child[very thick, every child/.style={thin, sibling distance=3ex}]{
        child[very thick]{ node (Xb) {$b$} }
        child[edge from parent/.style={dashed, draw}]{ node {$\dots$} }
        child[edge from parent/.style={draw=none}]{ node {} }
        child[edge from parent/.style={dashed, draw}]{ node {$\dots$} }
        child{ node {$b$} }
      }
      child{ node {$c$} }
      ;

      \node[below=0cm of p] {$p$};
      \node[below=0cm of Xb] {$p'$};
    \end{tikzpicture}
  \end{center}
  Clearly, $p'$ is in $\mp$-form.

  If we advanced to $p'' = X_c(w; p)$ instead, the situation would be as follows:
  
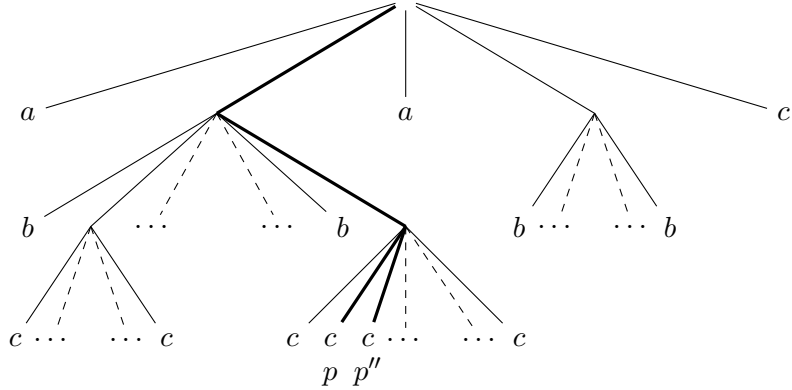
\begin{center}
    \begin{tikzpicture}
      \tikzstyle{level 1}=[sibling distance=15ex]
      \tikzstyle{level 2}=[sibling distance=5ex]
      \tikzstyle{level 3}=[sibling distance=3ex]
      \node {}
      child{ node {$a$} }
      child[very thick, every child/.style=thin]{
        child{ node {$b$} }
        child{
          child{ node {$c$} }
          child[edge from parent/.style={dashed, draw}]{ node {$\dots$} }
          child[edge from parent/.style={draw=none}]{ node {} }
          child[edge from parent/.style={dashed, draw}]{ node {$\dots$} }
          child{ node {$c$} }
        }
        child[edge from parent/.style={dashed, draw}]{ node {$\dots$} }
        child[edge from parent/.style={draw=none}]{ node {} }
        child[edge from parent/.style={dashed, draw}]{ node {$\dots$} }
        child{ node {$b$} }
        child[very thick, every child/.style=thin]{
          child{ node {$c$} }
          child[very thick]{ node (p) {$c$} }
          child[very thick]{ node (Xc) {$c$} }
          child[edge from parent/.append style={dashed}]{ node {$\dots$} }
          child[edge from parent/.style={draw=none}]{ node {} }
          child[edge from parent/.append style={dashed}]{ node {$\dots$} }
          child{ node {$c$} }
        }
      }
      child{ node {$a$} }
      child[every child/.style={sibling distance=3ex}]{
        child{ node {$b$} }
        child[edge from parent/.style={dashed, draw}]{ node {$\dots$} }
        child[edge from parent/.style={draw=none}]{ node {} }
        child[edge from parent/.style={dashed, draw}]{ node {$\dots$} }
        child{ node {$b$} }
      }
      child{ node {$c$} }
      ;

      \node[below=0cm of p, text height=1ex] (pLabel) {$p$};
      \node[below=0cm of Xc, text height=1ex] {$p''$};
    \end{tikzpicture}
  \end{center}
  Again, $p''$ clearly has $\mp$-form.
\end{example}

More formally, we can prove the following lemma:
\begin{lemma}\label{lem:forms}
  Let $\gamma$ be a $\pi$-term and $w = \subs{\gamma}{\omega + \omega^*}$. Then
  \begin{itemize}
    \item $X_a(w; -\infty)$ is in $+$-form for any $a \in \alphabet(w)$,
    \item if $l \in \dom{w}$ is in $+$-form, then so is $X_a(w; l)$ for all $a \in
          \alphabet(w_{(l, +\infty)})$,
    \item if $l \in \dom{w}$ is in $\mp$-form, then so is $X_a(w; l)$ for all $a \in
          \alphabet(w_{(l, +\infty)})$,
    \item $Y_a(w; +\infty)$ is in $-$-form for any $a \in \alphabet(w)$,
    \item if $r \in \dom{w}$ is in $-$-form, then so is $Y_a(w; r)$ for all $a \in
          \alphabet(w_{(-\infty, r)})$, and
    \item if $r \in \dom{w}$ is in $\mp$-form, then so is $Y_a(w; r)$ for all $a \in
          \alphabet(w_{(-\infty, r)})$.
  \end{itemize}
\end{lemma}
\begin{proof}
  By symmetry, we only have to prove the first three assertions.

  We start by proving the first one. The cases $\gamma = \varepsilon$ and $\gamma \in \Sigma$ are trivial.
  For the other cases, let $a \in \alphabet(w)$ and define $p = X_a(w; -\infty)$. If $\gamma =
  \alpha \beta$ for two $\pi$-terms $\alpha$ and $\beta$, we define $u = \subs{\alpha}{\omega +
  \omega^*}$ and $v = \subs{\beta}{\omega + \omega^*}$. Now, we have $p = X_a(u; -\infty)$
  or $p = X_a(v; -\infty)$. In either case, we can apply induction, which yields that $p$ is in
  $+$-form with respect to its sub-$\pi$-term, and we are done. If $\gamma = (\alpha)^\pi$ for a
  $\pi$-term $\alpha$, define $u = \subs{\alpha}{\omega + \omega^*}$. Clearly, $p$ has to be equal
  to $(p', 1)$ for $p' = X_a(u; -\infty)$. By induction, we have that $p'$ is in $+$-form and,
  since $1 \in \Nat$, we are done.

  The second assertion can be proved similarly. Let $l \in \dom{w}$ be in $+$-form. If $\gamma =
  \varepsilon$ or $\gamma \in \Sigma$, we do not have any remaining positions in $\dom{w}$ for $p =
  X_a(w; l)$. In the other cases, let $a \in \alphabet(w_{(l, +\infty)})$. If $\gamma =
  \alpha \beta$ for two $\pi$-terms $\alpha$ and $\beta$, we define $u = \subs{\alpha}{\omega +
  \omega^*}$ and $v = \subs{\beta}{\omega + \omega^*}$. We have to distinguish: if $l$ and $p$
  both are in $\dom{u}$ (or, symmetrically, in $\dom{v}$), then we have $p = X_a(u; l)$ and we
  can apply induction. Therefore, $p$ is in $+$-form with respect to $\alpha$, which yields that
  $p$ is also in $+$-form with respect to $\gamma$. If $l \in \dom{u}$ and $p \in \dom{v}$, then we
  know that $p = X_a(v; -\infty)$, which is in $+$-form with respect to $\beta$ by the first
  assertion and, therefore, also in $+$-form with respect to $\gamma$.

  If $\gamma = (\alpha)^\pi$ for a $\pi$-term $\alpha$, define $u = \subs{\alpha}{\omega +
  \omega^*}$. We can write $l = (l', m)$ and $p = (p', k)$ for some $m, k \in \Nat \uplus -\Nat$ and
  $l', p' \in \dom{u}$. By definition, $l'$ is in $+$-form because $l$ is so. For $k$, there are only
  two possible cases: $k = m$ and $k = m + 1$. In the former case, we know that $p' = X_a(u; l')$
  and can apply induction to get that $p'$ is in $+$-form with respect to $\alpha$. Since $l$ is
  in $+$-form, we also know that $k = m \in \Nat$. Together, this yields that $p$ is in $+$-form
  with respect to $\gamma$. In the latter case $k = m + 1$, we know that $p' = X_a(u; -\infty)$ is
  in $+$-form by the first assertion. Since $l$ is in $+$-form, we have $m \in \Nat$ and,
  therefore, also $k = m + 1 \in \Nat$. Thus, $p$ is in $+$-form with respect to $\gamma$.

  Now, we prove the third assertion. Let $l$ be in $\mp$-form with respect to $\gamma$. Again,
  for $\gamma = \varepsilon$ or $\gamma \in \Sigma$, there is nothing to show. The case for $\gamma =
  \alpha \beta$ for two $\pi$-terms $\alpha$ and $\beta$ can be proved analogously to the
  corresponding case in the proof for the second assertion.

  If $\gamma = (\alpha)^\pi$ for a $\pi$-term $\alpha$, define $u = \subs{\alpha}{\omega +
  \omega^*}$ and let $p = X_a(w; l)$ for an $a \in \alphabet(w_{(l, +\infty)})$. We can write $l
  = (l', m)$ and $p = (p', k)$ for some $m, k \in \Nat \uplus -\Nat$ and $l', p' \in \dom{u}$. Again,
  there are only two possible cases: $k = m$ and $k = m + 1$. In the former case we have $p' =
  X_a(u; l')$. If $k = m \in \Nat$, we know that $l'$ is in $+$-form since $l$ has to be in
  $\mp$-form. By the second assertion, this yields that $p'$ is in $+$-form as well. Therefore,
  we have that $p$ is in $\mp$-form. If $k = m \in -\Nat$, we can simply apply induction and get
  that $p'$ is in $\mp$-form with respect to $\alpha$. This yields that $p$ is in $\mp$-form with
  respect to $\gamma$. In the latter case $k = m + 1$, we know that $p' = X_a(u; -\infty)$ is in
  $+$-form by the first assertion and that $m \neq -1$. This yields $\mp$-form for $p$ if $k = m
  + 1 \in \Nat$ or $k = m + 1 \in -\Nat$.
\end{proof}

This observation is also important for pairs which arise by consecutive factorization at the
first/last $a$. If we start in $(-\infty, +\infty)$ and apply a sequence $X$ of elements from
$X_\Sigma^R =
\set{X_a^R}{a \in \Sigma}$, then $l'$ in the resulting pair $(l', +\infty) = (-\infty, +\infty)
\cdot X$ will be in $\mp$-form. Equally, for a sequence over $Y_\Sigma^L = \set{Y_a^L}{a \in
\Sigma}$, the right position will be in $\pm$-form. Can we assume that the left position is always in
$\mp$-form and the right one is always in $\pm$-form? Unfortunately, the answer to this question is
\enquote{no}. Suppose we start in the pair $(l, r)$ where $l$ is in $\mp$-form and $r$ is in
$\pm$-form. If we apply $X_a^L$ for some $a \in \Sigma$, then, obviously, we end up in a pair $(l,
r')$ whose right position is in $\mp$-form. But: the right position $r'$ branches form
$l$ at some point and strictly below that point values for relevant $\pi$-positions are equal to
$1$. So, we could say that for this lower part $r'$, indeed, is in $\pm$-form. In fact, we will prove
that, strictly below the branching, $l$ is always in $\mp$-form and $r$ is always in $\pm$-form
for any pair $(l, r) = (-\infty, +\infty) \cdot F$ where $F$ is a sequence of elements from
$F_\Sigma = \{ X_a^D, Y_a^D, C_{a, b} \mid a, b \in \Sigma, D \in \{ L, R \} \}$. The situation at
the branching point itself depends on
whether we have a c-branching or a $\pi$-branching pair. For a $\pi$-branching pair, we cannot
make an assumption on the value of $l$ and $r$ at the actual branching $\pi$-position. We
accommodate for this by having two definitions.
\begin{definition}
  Let $(l, r)$ be a pair of positions in $\subs{\gamma}{\omega + \omega^*}$ for a $\pi$-term
  $\gamma$ such that $l$ is strictly smaller than $r$.

  The pair is called \emph{well-c-shaped} (with respect to $\gamma$) if
  \begin{itemize}
    \item $\gamma = \varepsilon$ or $\gamma \in \Sigma$,
    \item $\gamma = \alpha \beta$ for $\pi$-terms $\alpha$ and $\beta$, $l \in
          \dom{\subs{\alpha}{\omega + \omega^*}} \uplus \smallset{-\infty}$, $r \in
          \dom{\subs{\alpha}{\omega + \omega^*}}$ and $(l, r)$ is well-c-shaped with
          respect to $\alpha$,
    \item $\gamma = \alpha \beta$ for $\pi$-terms $\alpha$ and $\beta$, $l \in
          \dom{\subs{\beta}{\omega + \omega^*}}$, $r \in \dom{\subs{\beta}{\omega + \omega^*}}
          \uplus \smallset{+\infty}$ and $(l, r)$ is well-c-shaped with respect to $\beta$,
    \item $\gamma = \alpha \beta$ for $\pi$-terms $\alpha$ and $\beta$, $l \in
          \dom{\subs{\alpha}{\omega + \omega^*}} \uplus \smallset{-\infty}$, $r \in
          \dom{\subs{\beta}{\omega + \omega^*}} \uplus \smallset{+\infty}$, and
          \begin{itemize}
            \item $l = -\infty$ or $l$ is in $\mp$-form with respect to $\alpha$, and
            \item $r = +\infty$ or $r$ is in $\pm$-form with respect to $\beta$,
          \end{itemize}
          or
    \item $\gamma = (\alpha)^\pi$ for a $\pi$-term $\alpha$ and
          \begin{itemize}
            \item $l = -\infty$ and $r = +\infty$,
            \item $l = -\infty$ and $r$ is in $\pm$-form,
            \item $l$ is in $\mp$-form and $r = +\infty$, or
            \item $l = (l', n)$ and $r = (r', n)$ for an $n \in \Nat \uplus -\Nat$ and $(l', r')$ is
                  well-c-shaped with respect to $\alpha$,
          \end{itemize}
  \end{itemize}
  It is called \emph{well-$\pi$-shaped} (with respect to $\gamma$) if $\gamma \neq \varepsilon$ and
  $\gamma \not\in \Sigma$ as well as
  \begin{itemize}
    \item $\gamma = \alpha \beta$ for $\pi$-terms $\alpha$ and $\beta$, $l, r \in
          \dom{\subs{\alpha}{\omega + \omega^*}}$ and $(l, r)$ is well-$\pi$-shaped with
          respect to $\alpha$,
    \item $\gamma = \alpha \beta$ for $\pi$-terms $\alpha$ and $\beta$, $l, r \in
          \dom{\subs{\beta}{\omega + \omega^*}}$ and $(l, r)$ is well-$\pi$-shaped with
          respect to $\beta$, or
    \item $\gamma = (\alpha)^\pi$ for a $\pi$-term $\alpha$, $l = (l', n)$, $r = (r',
          m)$ for $n, m \in \Nat \uplus -\Nat$ and $l', r' \in \dom{\subs{\alpha}{\omega + \omega^*}}$
          and
          \begin{itemize}
            \item $n = m$ and $(l', r')$ is well-$\pi$-shaped with respect to $\alpha$, or
            \item $n \neq m$, $l'$ is in $\mp$-form and $r'$ is in $\pm$-form.
          \end{itemize}
  \end{itemize}
  Finally, it is called \emph{well-shaped} (with respect to $\gamma$) if it is well-c-shaped or
  well-$\pi$-shaped (with respect to $\gamma$).
\end{definition}
\noindent The definition of well-c-shapedness is related to c-branching pairs and the definition of
well-$\pi$-shapedness is related to $\pi$-branching pairs. Note that any pair $(l, r)$ of positions
in $\subs{\gamma}{\omega + \omega^*}$ for a $\pi$-term is well-c-shaped if $l = -\infty$ or $r =
+\infty$. This results in some asymmetry in the definition.

We proceed by showing that any pair $(l, r)$ which arises by consecutive factorization at the
first/last $a$ is well-c-shaped or well-$\pi$-shaped.
\begin{lemma}
  Let $(l, r)$ be a pair of positions in $\subs{\gamma}{\omega + \omega^*}$ for a $\pi$-term
  $\gamma$. If $(l, r)$ is well-shaped, then so is $(l, r) \cdot F$ for any $F \in F_\Sigma$ (if it
  is defined).

  Therefore, $(-\infty, +\infty) \cdot F_1 F_2 \dots F_n$ is well-shaped with respect to $\gamma$
  for any $F_1,\allowbreak F_2,\allowbreak \dots,\allowbreak F_n \in F_\Sigma$ (if it is defined).
\end{lemma}
\begin{proof}
  The second part follows from the first since $(-\infty, +\infty)$ is well-c-shaped by definition.

  Let $(l, r)$ be well-shaped and let $p = X_a(w; l)$ be defined for an $a \in \Sigma$. Due to
  symmetry, it remains to show that $(l, p)$ and $(p, r)$ are well-shaped. Note that this also
  includes the $C_{a, b}$ factorization.

  First, consider $(l, p)$. If $l = -\infty$, then $p$ is in $+$-form by \autoref{lem:forms} and,
  therefore, also in $\pm$-form. This yields that $(-\infty, p)$ is well-c-shaped. Thus, we may
  safely assume that $l \neq -\infty$ and proceed by induction on the structure of $\gamma$.
  With this assumption, the cases $\gamma = \varepsilon$ and $\gamma \in \Sigma$ cannot occur.

  If $\gamma = \alpha \beta$ for two $\pi$-terms $\alpha$ and $\beta$, define $u =
  \subs{\alpha}{\omega + \omega^*}$ and $v = \subs{\beta}{\omega + \omega^*}$. If $l$ and $r$ are
  both in $\dom{u}$, then we know, by definition of well-shapedness, that $(l, r)$ is well-shaped
  with respect to $\alpha$. We also know that $p = X_a(u; l)$ since $p$ must be between $l$ and $r$.
  By induction, we have that $(l, p)$ is well-shaped with respect to $\alpha$, which yields that
  $(l, p)$ is also well-shaped with respect to $\gamma$. If $l \in \dom{v}$ and $r \in \dom{v} \uplus
  \smallset{+\infty}$, we can apply a similar argument. For $l \in \dom{u}$ and $r \in \dom{v} \uplus
  \smallset{+\infty}$, we know that $(l, r)$ is well-c-shaped and that $l$ is in $\mp$-form with
  respect to $\alpha$. This yields that $(l, +\infty)$ is well-c-shaped with respect to $\alpha$. By
  induction, we then have that $(l, p) = (l, X_a(u; l))$ is well-shaped with respect to $\alpha$
  if $p \in \dom{u}$. Well-shapedness with respect to $\gamma$ follows in both cases, i.\,e.\ $(l,
  p)$ is well-c-shaped or $(l, p)$ is well-$\pi$-shaped, by definition of well-shapedness. If $p \in
  \dom{v}$, then we have $p = X_a(v; -\infty)$, which is in $+$-form (and, thus, in $\pm$-form) by
  \autoref{lem:forms}. By definition, we have the well-c-shapedness of $(l, p)$ with respect to
  $\gamma$.

  If $\gamma = (\alpha)^\pi$ for a $\pi$-term $\alpha$, then define $u = \subs{\alpha}{\omega + \omega^*}$.
  We can write $l = (l', m)$ and $p = (p', k)$ for some $m, k \in \Nat \uplus -\Nat$ and
  $l', p' \in \dom{u}$. For $k = m$, we have $p' = X_a(u; l')$. This yields well-shapedness with
  respect to $\alpha$ of $(l', p')$ by induction. By definition, we also have well-shapedness with
  respect to $\gamma$. For $k = m + 1$ (the only other possible case), we know $p' = X_a(u;
  -\infty)$ which is in $+$-form with respect to $\alpha$ by \autoref{lem:forms} and, therefore,
  also in $\pm$-form. If $r = +\infty$, then $(l, r)$ had to be well-c-shaped and $l$ has to be in
  $\mp$-form. Then, $(l, p)$ is well-$\pi$-shaped. If $r = (r', n)$ for a $n \in \Nat \uplus -\Nat$
  and $r' \in \dom{u}$, then $n$ must be greater than $m$ (with respect to $\omega + \omega^*$),
  because it must be greater than or equal to $k$. Therefore, $(l, r)$ has to be well-$\pi$-shaped
  and $l$ has to be in $\mp$-form. Again, $(l, p)$ is well-$\pi$-shaped then, which concludes the
  proof that $(l, p)$ is always well-shaped.

  Next, consider $(p, r)$. If $r = +\infty$, then $(l, r)$ has to be well-c-shaped. This implies that
  $l = -\infty$ or $l$ is in $\mp$-form. In either case, we have that $p = X_a(w; l)$ is in
  $\mp$-form by \autoref{lem:forms}. Thus, $(p, +\infty)$ is well-c-shaped. Again, we can safely
  assume that $r \neq +\infty$ and proceed by induction on the structure of $\gamma$ where the
  cases $\gamma = \varepsilon$ and $\gamma \in \Sigma$ do not occur.

  If $\gamma = \alpha \beta$ for two $\pi$-terms $\alpha$ and $\beta$, then define $u =
  \subs{\alpha}{\omega + \omega^*}$ and $v = \subs{\beta}{\omega + \omega^*}$. The case $l \in
  \dom{u} \uplus \smallset{-\infty}$ and $r \in \dom{u}$ and the case $l, r \in \dom{v}$ are similar
  to the argumentation for $(l, p)$. For $l \in \dom{u} \uplus \smallset{-\infty}$ and $r \in
  \dom{v}$, $(l, r)$ has to be well-c-shaped. Therefore, $r$ has to be in $\pm$-form with respect to
  $\beta$ and $(-\infty, r)$ is well-c-shaped with respect to $\beta$. Induction yields the
  well-shapedness with respect to $\beta$, and, thus, also with respect to $\gamma$, of $(p, r)$ if
  $p \in \dom{v}$. If $p \in \dom{u}$, we observe that well-c-shapedness of $(l, r)$ yields $l =
  -\infty$ or $l$ in $\mp$-form. In either case, $p$ is in $\mp$-form by \autoref{lem:forms}. Thus,
  $(p, r)$ is well-c-shaped with respect to $\gamma$.

  If $\gamma = (\alpha)^\pi$ for a $\pi$-term $\alpha$, define $u = \subs{\alpha}{\omega +
  \omega^*}$. We can write $p = (p', k)$ and $r = (r', n)$ for $k, n \in \Nat \uplus -\Nat$. For
  $k = n$, distinguish: if $l = -\infty$, then $(l, r)$ is well-c-shaped and $r$ has to be in
  $\pm$-form with respect to $\gamma$. Therefore, $r'$ has to be in $\pm$-form with respect to
  $\alpha$. If $l = (l', m)$ for $l' \in \dom{u}$ and $m \in \Nat \uplus -\Nat$ with $m <_{\omega +
  \omega^*} k = n$, then $(l, r)$ is well-$\pi$-shaped and $r'$, again, is in $\pm$-form. In both
  cases, this yields that $(-\infty, r')$ is well-c-shaped with respect to $\alpha$. Induction yields
  that $(p', r') = (X_a(u; -\infty), r')$ is well-shaped with respect to $\alpha$. Since $k = n$,
  this implies that $(p, r)$ is well-shaped. If $l = (l', m)$ but $m = k = n$, then $(l', r')$ has to
  be well-shaped with respect to $\alpha$ and so has to be $(p', r')$ by induction, which again means
  that $(p, r)$ is well-shaped with respect to $\gamma$. If $k <_{\omega + \omega^*} n$ and $l =
  -\infty$, then $(l, r)$ is well-c-shaped, $r$ is in $\pm$-form and $p'$ is in $+$-form by
  \autoref{lem:forms}. Therefore, $(p, r)$ is well-$\pi$-shaped. If $l = (l', m)$ for $l' \in
  \dom{u}$ and $m \in \Nat \uplus -\Nat$ and $m \leq_{\omega + \omega^*} k <_{\omega + \omega^*} n$,
  then $l'$ and, thus, also $p'$ by \autoref{lem:forms} has to be in $\mp$-form while $r'$ has to
  be in $\pm$-form. This concludes the proof because, then, $(p, r)$ is well-$\pi$-shaped.
\end{proof}

By the previous lemma, we know that there is maximally one switching between values form $-\Nat$ to
values from $\Nat$ in $l$ below the branching position and that the same -- but in reverse -- is
true for $r$ when we consider a pair $(l, r)$ which arises from $(-\infty, +\infty)$ by applying a
sequence of factorizations from $F_\Sigma$. However, the algorithm described in the proof of
\autoref{thm:RelationsAreDecidable} performs a normalization after each factorization step. So far,
we have ignored this normalization but the next lemma states that normalization preserves
well-shapedness.

\begin{lemma}
  Let $(l, r)$ be a well-shaped, normalizable pair of positions in $\subs{\gamma}{\omega + \omega^*}$
  for a $\pi$-term $\gamma$ over $\Sigma$. Then, $\overline{(l, r)}^\gamma$ is well-shaped.
\end{lemma}
\begin{proof}
  Let $\overline{(l, r)}^\gamma = (\bar{l}, \bar{r})$. In the special case $l = -\infty$ and $r =
  +\infty$, we have $\overline{(l, r)}^\gamma = (l, r)$ and we are done. For the other cases, we
  proceed by induction over $\gamma$.

  For $\gamma = \varepsilon$ and $\gamma \in \Sigma$, we have that $\overline{(l, r)}^\gamma$ is
  well-c-shaped. In the case $\gamma = \alpha\beta$ for two $\pi$-terms $\alpha$ and $\beta$, we
  distinguish: if $l \in \{ -\infty \} \uplus \dom{\subs{\alpha}{\omega + \omega^*}}$ and $r \in
  \dom{\subs{\beta}{\omega + \omega^*}} \uplus \{ + \infty \}$, we have that $(l, r)$ must be
  well-c-shaped, that $\bar{l}$ is given by the first component of $\overline{(l, +\infty)}^\alpha$
  and that $\bar{r}$ is given by the second component of $\overline{(-\infty, r)}^\beta$. For $l =
  -\infty$, we also have $\bar{l} = -\infty$ and, for $r = +\infty$, we also have $\bar{r} =
  +\infty$. If $l \neq -\infty$, then $l$ must be in $\mp$-form. Note that normalization preserves
  the form of $l$ when normalizing $(l, +\infty)$, so $\bar{l}$ is in $\mp$-from as well.
  Symmetrically, $\bar{r}$ must be in $\pm$-form if $r \neq +\infty$. Thus, $\overline{(l,
  r)}^\gamma$ is well-c-shaped. For $l \in \{ -\infty \} \uplus \dom{\subs{\alpha}{\omega +
  \omega^*}}$ and also $r \in \dom{\subs{\alpha}{\omega + \omega^*}}$, we have $\overline{(l,
  r)}^\gamma = \overline{(l, r)}^\alpha$, which is well-shaped by induction. The same argument
  proves the symmetric case $l \in \dom{\subs{\beta}{\omega + \omega^*}}$ and $r \in
  \dom{\subs{\beta}{\omega + \omega^*}} \uplus \{ +\infty \}$.

  The remaining case is $\gamma = (\alpha)^\pi$ for a $\pi$-term $\alpha$. If we have $l = -\infty$
  or $r = +\infty$, then we know that $l$ is in $\mp$-form or that $r$ is in $\pm$-form,
  respectively, because $(l, r)$ has to be well-c-shaped. As before, normalizing $(-\infty, r)$ or
  $(l, +\infty)$ preserves this form, which makes the result well-c-shaped as well. Therefore, we
  have $l = (l', n)$ and $r = (r', m)$ for $n, m \in \Nat \uplus -\Nat$ and positions $l', r' \in
  \dom{\subs{\alpha}{\omega + \omega^*}}$. For $n = m$, we have $\bar{l} = (\bar{l'}, 1)$ and
  $\bar{r} = (\bar{r'}, 1)$ with $(\bar{l'}, \bar{r'}) = \overline{(l', r')}^\alpha$. By induction,
  we have that $\overline{(l', r')}^\alpha$ is well-c-shaped or well-$\pi$-shaped with respect to
  $\alpha$. In either case, $(\bar{l}, \bar{r})$ is also well-c-shaped or well-$\pi$-shaped with
  respect to $\gamma$. For $n \neq m$, the pair $(l, r)$ has to be well-$\pi$-shaped, $l'$ must be
  in $\mp$-form and $r'$ must be in $\pm$-form. Additionally, we have $\bar{l} = (\bar{l'},
  \bar{n})$ and $\bar{r} = (\bar{r'}, \bar{m})$ for $\bar{n}, \bar{m} \in \Nat \uplus -\Nat$ and
  positions $\bar{l'}, \bar{r'} \in \dom{\subs{\alpha}{\omega + \omega^*}}$. Note that we also have
  $\bar{n} \neq \bar{m}$ (in both of the cases which can occur when normalizing). Furthermore,
  normalization of $(l', +\infty)$ (with respect to $\alpha$) preserves the $\mp$-form of $l'$ and
  normalization of $(-\infty, r')$ preserves the $\pm$-form of $r'$, i.\,e.\ we have $\bar{l'}$ in
  $\mp$-form and $\bar{r'}$ in $\pm$-form, which makes $\overline{(l, r)}^\gamma$ well-$\pi$-shaped.
\end{proof}

Combining the previous two lemmas shows the following. Suppose we start with the position pair
$(-\infty, +\infty)$ and apply a single factorization $F_1 \in F_\Sigma$, then we get $(l'_1, r'_1) =
(-\infty, +\infty) \cdot F_1$, which is well-shaped. If we then normalize $(l'_1, r'_1)$, the
resulting pair $(l_1, r_1) = \overline{(l'_1, r'_1)}^\gamma$ is also well-shaped. We can continue
with another factorization and normalize again; the result $(l_1, r_2)$ will still be well-shaped.
Therefore, all the pairs of positions $(l, r)$ which are states in the deterministic finite
automaton constructed in the proof of \autoref{thm:RelationsAreDecidable} are well-shaped and
normalized.

One last small observation is necessary before we discuss the details of storing those pairs:
by definition of the normalization, $l$ and $r$ may only have values from $\smallset{1,
-1}$ for $\pi$-positions which are hierarchically lower than the branching position. Therefore, it is
sufficient to store the potential $\pi$-position at which the values switch (from $-1$ to $1$
for $l$ or form $1$ to $-1$ for $r$).

\paragraph*{Storing a Well-Shaped, Normalized Pair.}
Summing up all of our observations results in the situation which is schematically represented in
\autoref{fig:valuesOflAndr}.
\begin{figure}[h]
  \begin{center}
    \begin{tikzpicture}
       [sibling distance=7ex, level distance=10ex,
        every node/.style={fill, circle, minimum size=3pt, inner sep=0, outer sep=1mm},
        labellike/.style={shape=rectangle, draw=none, fill=none, inner sep=3pt},
        every label/.style={labellike},
        squiggly/.style={edge from parent/.style={decorate, decoration={coil, aspect=0}, draw},
        level distance=10ex},
        normal/.style={edge from parent/.style={draw}, level distance=10ex}]
      \node {}
      child[squiggly]{
        node[label=left:{Branching, possible values:},
             label=right:{$(1, 2)$ and $(1, -1)$}] {}
          child[squiggly]{ node[label=left:{Switching, value: $1$}] {}
            child[squiggly]{ node[label=below:$l$] {}
              edge from parent [->] node[left, labellike] {values: $1$} }
            child[draw=none, edge from parent/.style={draw=none}]{ {} }
            edge from parent [->] node[left, labellike] {values: $-1$}
          }
          child[squiggly]{ node[label=right:{Switching, value: $-1$}] {}
            child[draw=none, edge from parent/.style={draw=none}]{ {} }
            child[squiggly]{ node[label=below:$r$] {}
              edge from parent [->] node[right, labellike] {values: $-1$} }
            edge from parent [->] node[right, labellike] {values: $1$}
          }
          edge from parent [->] node[right, labellike] {values: $(1, 1)$}
      };
    \end{tikzpicture}
  \end{center}
  \caption{Schematic representation of the possible values of a normalized pair $(l,
    r)$}\label{fig:valuesOflAndr}
\end{figure}
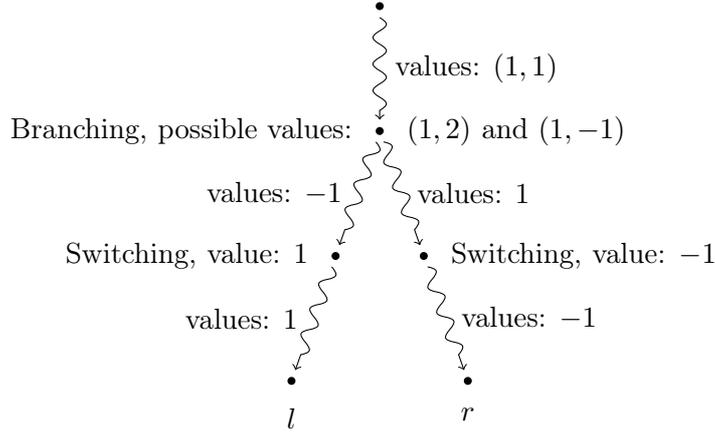
Therefore, we can construct the values of $l$ and $r$ for a normalized, well-shaped pair $(l, r)$ of positions in
$\subs{\gamma}{\omega + \omega^*}$ at all relevant $\pi$-positions in the $\pi$-term $\gamma$
from the variables specified in \autoref{tab:variablesForLAndR}.
\begin{table}
  \centering
  \begin{tabular}{|l|c|c|}
    \hline
    \textbf{Name} & \textbf{Values} & \textbf{Size} \\ \hhline{===}
    \texttt{ShapeType} & $\smallset{ \text{\enquote{well-c-shaped}},
      \text{\enquote{well-$\pi$-shaped}} }$ & $\landO( 1 )$\\ \hline
    \texttt{BranchPosition} & a $\pi$-position or $+\infty$ & $ \landO(\log n) $\\ \hline
    \texttt{BranchValues} & $\smallset{(1, 2), (1, -1)}$ & $\landO(1)$ \\ \hline
    \texttt{lSwitchPosition} & a $\pi$-position or $\bot$ & $ \landO(\log n) $\\ \hline
    \texttt{rSwitchPosition} & a $\pi$-position or $\bot$ & $ \landO(\log n) $\\ \hline
    \texttt{lEndPosition} & a $\Sigma$-position or $-\infty$ & $ \landO(\log n) $\\ \hline
    \texttt{rEndPosition} & a $\Sigma$-position or $+\infty$ & $ \landO(\log n) $\\ \hline
  \end{tabular}
  \caption{The variables required to compute the values at $\pi$-positions and their sizes.}
  \label{tab:variablesForLAndR}
\end{table}

We are going to explain the variables in more detail. Obviously, \texttt{ShapeType} gets the value
\enquote{well-c-shaped} if $(l, r)$ is well-c-shaped and the value \enquote{well-$\pi$-shaped}
if $(l, r)$ is well-$\pi$-shaped. For a well-c-shaped pair, we store the hierarchically lowest
$\pi$-position for which $l$ and $r$ share the same value in \texttt{BranchPosition}. If $l$ and
$r$ differ already in the hierarchically highest $\pi$-position, we store $+\infty$ in
\texttt{BranchPosition}. For well-$\pi$-shaped pairs, we store the $\pi$-position at which
$l$ and $r$ branch in \texttt{BranchPosition} and their respective values there in
\texttt{BranchValues}. The value of \texttt{BranchValues} is not relevant in other cases.
If there is a $\pi$-position at which the values of $l$ switch from $-1$ to $1$ (that position
is bound to be hierarchically lower than the branching position), then we store the position of the
first occurrence of $1$ in \texttt{lSwitchPosition}. \texttt{rSwitchPosition} stores the (potential)
corresponding position for $r$. Finally, in \texttt{lEndPosition} and \texttt{rEndPosition} we store
the position in the $\pi$-term $\gamma$ which corresponds to $l$ or $r$, respectively.

One may verify that all these values can be stored in the size specified in the table, where $n$ is
the length of the $\pi$-term $\gamma$ seen as a finite word in $\Sigma \uplus \smallset{\,{(},\, {)},\, {{}^\pi}\,}$.

The variables are sufficient to compute the values of $l$ and $r$ at all relevant
$\pi$-positions deterministically within logarithmic space bounds. We only discuss how to do this
for $r$, as this version can easily be adapted for $l$. Assume we want to compute the value of $r$
at a $\pi$-position given in \texttt{piPos}. Notice that \texttt{piPos} is relevant for $r$ if
and only if \texttt{rEndPosition} lies in between the opening and closing parenthesis belonging to
\texttt{piPos}. Because finding the matching left or right parenthesis is a simple matter of
counting the opening and closing parentheses this can be done deterministically in logarithmic
space. Thus, we may safely assume that our routine is only ever called for $\pi$-positions which
are relevant for $r$. To compute the value of $r$ at \texttt{piPos}, we only have to know to which
part of $r$ it belongs with respect to the schematic representation in \autoref{fig:valuesOflAndr}.
The first part where all values of $r$ are equal to $1$ consists of all $\pi$-positions which
are hierarchically higher than \texttt{BranchPosition} (assuming that it is not equal to $\bot$),
including \texttt{BranchPosition} if \texttt{ShapeType} is \enquote{well-c-shape} or excluding if
\texttt{BranchPosition} is \enquote{well-$\pi$-shaped}. Because we already know that
both, \texttt{BranchPosition} and \texttt{piPos} are relevant for $r$, checking whether one is
hierarchically higher than the other can be done by comparing their positions in $\gamma$, which
is possible in logarithmic space. The hierarchically higher one is to the right of the lower one.
If \texttt{piPos} belongs to the first part, we can return $1$ immediately. Next, we check whether
\texttt{BranchPosition} is equal to \texttt{piPos} and \texttt{ShapeType} is
\enquote{well-$\pi$-shaped}. If this is the case, we can return the value $2$ or $-1$ depending
on the value of \texttt{BranchValues}. The next part is from \texttt{BranchPosition} (excluding) to
\texttt{rSwitchPosition} (excluding). Again, we can check whether \texttt{piPos} belongs to this
part by comparing the position in the $\pi$-term and return $1$ immediately. If \texttt{piPos} was
not in any part so far, we know that it belongs to the last part from \texttt{rSwitchPosition}
(including) onwards and can return $-1$. If \texttt{BranchPosition} is $+\infty$ or
\texttt{rSwitchPosition} is $\bot$, then the corresponding parts in \autoref{fig:valuesOflAndr}
simply do not exist and we can omit the respective checks.

\paragraph*{Compute a Follow-Up Pair.}
Now that we know how to store the well-shaped, normalized pair $(l, r)$ efficiently, we need to find a way to
compute the normalized follow-up pair if we apply an element from $Z_\Sigma^D = \{ X_a^D, Y_a^D \mid
a \in \Sigma, D \in \{ L, R \} \}$ ($C_{a, b}$
needs to be handled a bit differently; see below). By symmetry, we
can restrict our considerations to elements from $X_\Sigma^D$, which means that we have to find the
first $a$-position for an $a \in \Sigma$ on the right of the position $l$ in $\subs{\gamma}{\omega
+ \omega^*}$ for the $\pi$-term $\gamma$. For this, we use the additional variables from
\autoref{tab:variablesForFollowUp}.
\begin{table}[h]
  \centering
  \begin{tabular}{|l|c|c|}
    \hline
    \textbf{Name} & \textbf{Values} & \textbf{Size} \\ \hhline{===}
    \texttt{CurrentPPosition} & a position in $\gamma$ & $ \landO(\log n) $\\ \hline
    \texttt{pBranchPosition} & a $\pi$-position or $+\infty$ & $ \landO(\log n) $\\ \hline
    \texttt{OpenParentheses} & a value to count open parentheses & $ \landO(\log n) $\\ \hline
  \end{tabular}
  \caption{The variables required to compute the follow-up pair and their sizes.}
  \label{tab:variablesForFollowUp}
\end{table}

\label{alg:computeFollowUpPair}
We start by setting \texttt{pBranchPosition} to $\bot$ and \texttt{OpenParentheses} to $0$.
\texttt{Cur\-rent\-P\-Pos\-ition} gets assigned the next position to the right of the value of
\texttt{lEndPosition}. The algorithm now iteratively moves \texttt{CurrentPPosition} to the right
one position at a time. If \texttt{CurrentPPosition} reaches an $a$-position, we are done. If
\texttt{CurrentPPosition} reaches the end of $\gamma$, then we know that there is no further $a$ and
we can stop. If \texttt{CurrentPPosition} reaches an opening parenthesis, then we increment
\texttt{OpenParentheses}. If \texttt{Cur\-rent\-P\-Pos\-ition} reaches a closing parenthesis and
\texttt{pBranchPosition} is $\bot$, we distinguish: if \texttt{Open\-Par\-en\-theses} is $> 0$, then we
decrement it and move on; if it is $= 0$, then we know that the corresponding $\pi$-position is
relevant for $l$. We check if the value of $l$ at that position is equal to $-1$. If that is the
case, we can simply move on. If the value is from $\Nat$ or $< -1$, then we set
\texttt{pBranchPosition} to the $\pi$-position which corresponds to the closing parenthesis and
move \texttt{CurrentPPosition} to the matching opening parentheses using \texttt{OpenParentheses} to
count the number of opening and closing parentheses. From that position we continue with the normal
algorithm. If we then reach a closing parenthesis and \texttt{pBranchPosition} is equal to
\texttt{CurrentPPosition}, then we set \texttt{pBranchPosition} to $\bot$ and \texttt{OpenPar\-en\-the\-sis}
to $0$ and continue normally. Any other closing parenthesis which is reached while
\texttt{pBranchPosition} is not equal to $\bot$ simply gets ignored.

\begin{example}
  The only interesting part of the algorithm is that when a closing parenthesis is reached.
  Therefore, we explain this part in more detail. Suppose we are in the following
  situation\footnote{The numbers here do not have a special meaning other than to tell the
  parenthesis pairs and their exponents apart from one another.} and want to move to the next $a$:
  \begin{center}
    \begin{tikzpicture}[every node/.style={text height=.8em, text depth=.2em}]
      \matrix [rectangle, draw, matrix of math nodes, ampersand replacement=\&,
      inner sep=2pt] (gamma) {
        a \& (_1 \& (_2 \& a \& )_2 \& {}^{\pi_2} \& b \& (_3 \& c \& b \& (_4 \& b \& )_4 \& {}^{\pi_4} \& )_3 \& {}^{\pi_3} \& )_1 \& {}^{\pi_1} \& a \& (_5 \& b \& )_5 \& {}^{\pi_5} \& c \\
      };
      \node[below=0.2cm of gamma-1-1, xshift=-0.5cm] {$l$:};
      \draw[-latex] ($(gamma-1-9.south) - (0, 0.75cm)$) -- ($(gamma-1-9.south) - (0, 0.2cm)$);
      \node[below=0.2cm of gamma-1-16] {$2$};
      \node[below=0.2cm of gamma-1-18] {$-1$};
    \end{tikzpicture}
  \end{center}
  The algorithm would move \texttt{CurrentPPosition} to the right until it reaches $(_4$ where
  \texttt{Open\-Parentheses} gets increased to $1$. When $)_4$ gets reached, \texttt{OpenParen\-theses}
  is set back to $0$, which triggers a different behavior when reaching $)_3$ in
  the next step: \texttt{pBranch\-Position} gets set to ${\pi_3}$ and \texttt{CurrentPPosition} gets
  moved to $(_3$ since $l$ has the value $2$ (and \emph{not} $-1$) at ${\pi_3}$. From there,
  \texttt{CurrentPPosition} moves back to the right until it reaches $)_3$ because $)_4$ simply
  gets ignored and because there is no $a$ between $(_3$ and $)_3$. At $)_3$, \texttt{OpenPrantheses}
  and \texttt{pBranchPosition} are reset to $0$ and $\bot$. At the next closing parenthesis $)_1$,
  the value of $l$ gets checked again. But this time it is equal to $-1$ and we continue to the
  right where we find the sought-after next $a$.

  If $l$ had had the value $-2$ at ${\pi_1}$, then the algorithm would have moved to $(_1$
  and the next $a$ would have been the one between $(_2$ and $)_2$. In that case,
  \texttt{pBranch\-Position} would have been set to ${\pi_1}$ when the final $a$ is reached.
\end{example}

When the algorithm stops (and \texttt{CurrentPPosition} is not moved beyond the end of $\gamma$)
then \texttt{CurrentPPosition} points to the position in $\gamma$ which corresponds to the next
$a$-position $p$ in $\subs{\gamma}{\omega + \omega^*}$. The values at the $\pi$-positions can be
reconstructed from the stored variables. Again, this can be done by dividing the $\pi$-positions
into parts for which the value is well-known. If a $\pi$-position is relevant for $p$ but not for
$l$, then the value of $p$ at that position has to be $1$.
\begin{example}
  In the last example ${\pi_2}$, ${\pi_4}$ and ${\pi_5}$ are not relevant for $l$. If
  \texttt{CurrentP\-Position} would end up within the pair of parentheses belonging to any of these positions,
  its value there would be $1$.
\end{example}
\noindent If \texttt{pBranchPosition} has a value other than $\bot$, then $p$ has the value of $l$
plus $1$ there. Note that the value of $l$ cannot be $-1$ in that case by the definition of the algorithm. For
$\pi$-positions which are relevant for $l$ and for $p$ and which are hierarchically higher than
the position in \texttt{pBranchPosition}, we know that the values of $l$ and $p$ are equal. At those
positions which are hierarchically lower than \texttt{pBranchPosition} and which are relevant for
$l$ and for $p$, the value of $p$ has to be $1$ just like in the first case. If
\texttt{pBranchPosition} is $\bot$, then the values of $p$ at positions relevant for $l$ and for $p$
are equal to the values of $l$.

Because -- as said -- it can be checked in logarithmic space whether a $\pi$-position is
relevant for $l$, $p$ or $r$, this yields a deterministic algorithm with logarithmic space bounds for
calculating the values for $p$ at all its relevant $\pi$-positions. It remains to test whether
$p$ is still strictly smaller than $r$ and to compute the new values of the variables in
\autoref{tab:variablesForLAndR} after an $X_a^D$-step and normalization. This test can be done by walking through all
$\pi$-positions in $\gamma$ from the right to the left and checking whether they are relevant for
$p$ and for $r$. If that is the case, we check their values. If the value of $p$ is strictly smaller
than that of $r$, we are done. If they are equal everywhere, then \texttt{CurrentPPosition} needs to
be to the left of \texttt{rEndPosition} in $\gamma$.

For calculating the new values for the variables,
a similar approach can be applied. Consider the (slightly more difficult) case of $X_a^R$ where
$p$ gets the new $l$. We walk through the $\pi$-positions in $\gamma$ from right to left and check their
relevance for $p$ and for $r$. If they are relevant for one or the other, we can compute the
values. As long as the values are equal (which means that the position is relevant for both, $p$
and $r$) we know that the $\pi$-positions belong to the part in \autoref{fig:valuesOflAndr}
where $l$ and $r$ share the value $1$. If there is a $\pi$-position where the values are
different, we can store that position in \texttt{BranchPosition} and update \texttt{BranchValues}
according to the value pair which normalization would yield. In that case, we can set
\texttt{ShapeType} to \enquote{well-$\pi$-shaped}. After the \texttt{BranchPosition}, we
have to check $p$ and $r$ individually for a change from values in $\Nat$ to values in $-\Nat$ or
vice versa and store the corresponding position. If the values at all positions have been equal so
far and we reach a position which is relevant for $r$ but not for $p$ (the other way round is not
possible because $p$ has to be smaller than $r$), we can set \texttt{ShapeType} to
\enquote{well-c-shaped} and \texttt{BranchPosition} to the last $\pi$-position where the values
were equal. \texttt{ShapeType} gets also set to \enquote{well-c-shaped} if we reach \texttt{rEndPos}
before there is a difference in the values.

\paragraph*{Special Handling of $C_{a, b}$.}
By \autoref{thm:relationsInTrotterWeil}, we have to test whether $\subs{\alpha}{\omega + \omega^*}
\allowbreak\equiv_m^\WI \subs{\beta}{\omega + \omega^*}$ holds in order to test whether $\alpha =
\beta$ holds in $\Rm[m + 1] \cap \Lm[m + 1]$. The definition of $\equiv_m^\WI$, however, does not
only rely on factorizations at the first or last $a$. It uses an additional special factorization
which behaves like first factorizing on the first $a$ and then factorizing on the last $b$ but
which is only possible if the first $a$ is to the right of the last $b$. As defined above, this
kind of factorization is represented by $C_{a, b}$.

Clearly, $(l, r) \cdot C_{a, b}$ can only be defined if $(l, r) \cdot X_a Y_b$ is defined. But the
other direction does not hold: it is possible that $(l, r) \cdot X_a Y_b$ is defined while $(l, r)
\cdot C_{a, b}$ is undefined.

We have described how we can compute the follow-up values of the variables in
\autoref{tab:variablesForLAndR} for a factorization at the first $a$ and, by symmetry, also for
factorizations at the last $b$ when we start with a factor which is given by a normalized pair
$(l, r)$. We can combine these two calculations into a single one for $C_{a, b}$. For this, we assume
that we have two instances of the variables in
\autoref{tab:variablesForFollowUp}, one instance for the factorization at the first $a$ and one
instance for the factorization at the last $b$. For both instances, we apply the normal algorithm
for the corresponding factorization. Afterwards, we check whether the position of the first $a$ is
to the left of $r$ and whether the position of the last $b$ is to the right of $l$ just like we did
before. For $C_{a, b}$, we simply add a third check which can be done in a similar manner to the
other checks: we check whether the position of the first $a$ is to the right of the position of the
last $b$. If this check fails, then we know that $(l, r) \cdot C_{a, b}$ is undefined, otherwise it
indeed \emph{is} defined.

\paragraph*{Decidability in Nondeterministic Logarithmic Space.}
Let us recapitulate the proof for decidability for the word problems for $\pi$-terms. First, we
saw that the actual problem we needed to decide was whether $\subs{\alpha}{\omega + \omega^*}$ and
$\subs{\beta}{\omega + \omega}$ are equivalent with respect to the relation belonging to the variety
in question. To decide this, we constructed a deterministic finite automata $\mathcal{A}_\alpha$ and
$\mathcal{A}_\beta$ for each of the input $\pi$-terms $\alpha$ and $\beta$. The automaton
$\mathcal{A}_\alpha$ accepted exactly those factorization sequences $F \in F_\Sigma^*$ for which
$\subs{\alpha}{\omega + \omega^*} \cdot F$ (or $\subs{\beta}{\omega + \omega^*} \cdot F$ in the case
of $\mathcal{A}_\beta$) is defined. Additionally, we also constructed a deterministic
finite automaton $\mathcal{B}$ which accepted those factorization sequences that need to be tested
for the respective relation. Finally, we created for $\mathcal{A}_\alpha$ and for $\mathcal{A}_\beta$
an automaton accepting the intersection with the sequences accepted by $\mathcal{B}$. For these
automata, we checked the symmetric difference for emptiness.

Checking the symmetric difference can be done in nondeterministic logarithmic space by a naïve (iterative) guess and check algorithm.
The automaton $\mathcal{B}$ is fixed for every variety in the Trotter-Weil Hierarchy and for $\DA$
(or can be constructed deterministically in space logarithmic in $m$ if one considers this as an
input). Creating an automaton for the intersection can be done in deterministic logarithmic space.
The only interesting operation is the construction of $\mathcal{A}_\alpha$ and $\mathcal{A}_\beta$.
With the ideas from this section, we can describe how to do this in deterministic logarithmic space.
\begin{description}
  \item[Input:] a $\pi$-term $\alpha$
  \item[Output:] $\mathcal{A}_\alpha$ (which accepts a subset of $F_\Sigma^*$)
  \item[Algorithm:]\mbox{}\\ \vspace{-\baselineskip}
      \begin{algorithmic}
        \ForAll{possible values $V$ of the variables in \autoref{tab:variablesForLAndR}}
          \ForAll{$F \in F_\Sigma$}
            \State
              \begin{varwidth}[t]{\linewidth}
                Compute the values $V'$ of the variables for the follow-up pair after\\
                application of $F$ and normalization.
              \end{varwidth}\vspace{0.1cm}
            \State If there is such a follow-up pair, then output the transition
              \begin{center}
                \begin{tikzpicture}[->, auto]
                  \node[state] (V) {$V$};
                  \node[state, right=of V] (V') {$V'$};
                  \path (V) edge node {$F$} (V');
                \end{tikzpicture}
              \end{center}
          \EndFor
        \EndFor
        \State Mark the state which represents $(-\infty, +\infty)$ as initial state.
        \State Mark any state as final.
      \end{algorithmic}
\end{description}
In difference to our previous construction, the states in this automaton are not the normalized
pairs of positions anymore. Instead, we are representing them by their corresponding values for the
variables in \autoref{tab:variablesForLAndR}. Of course, not all possible values of the variables
represent a valid normalized pair, but, since the corresponding states cannot be reached from the
initial state, these errors do not affect the result.

These considerations allow us to state the following theorem.
\begin{restatable}{theorem}{wordProblemsAreInNL}
  \label{thm:wordProblemsAreInNL}
  Each of the word problems for $\pi$-terms over $\Rm$, $\Lm$, $\Rm \vee \Lm$, $\Rm \cap \Lm$ and $\DA$
  can be solved by a nondeterministic Turing machine in logarithmic space (for every $m \in \Nat$).
\end{restatable}

\section{Deterministic Polynomial Time}

While $\NL$ is quite efficient from a complexity class perspective,
directly translating the algorithm to polynomial time does not result in a better running time than
the algorithm for $\DA$ given by Moura \cite{moura2011word}. However, with some additional tweaks,
the algorithm's efficiency can be improved.

\paragraph*{Encoding and Calculating Positions.}
Because we are not limited in space, we can encode a position $p \in \dom{w}$ with $w =
\subs{\gamma}{\omega + \omega^*}$ for a $\pi$-term $\gamma$ simply by storing the values at the
relevant $\pi$-positions and the corresponding position in $\gamma$. Clearly, we can obtain
the values of $p$ at a given $\pi$-position in constant time. Additionally, we can store a
(possibly normalized) pair $(l, r)$ by storing $l$ and $r$. If we want to normalized such a pair, we
walk through all relevant $\pi$-positions and update the values of $l$ and $r$ there, which requires at
most linear time. Similarly, we can test whether a position is strictly smaller (or larger) than an
other position in linear time.

Suppose we have stored a position $p \in \dom{w}$ and want to compute $X_a(w; p)$ for an $a \in
\Sigma$. We can re-use the algorithm which we used previously to solve the problem in logarithmic
space (see page \pageref{alg:computeFollowUpPair}). This algorithm moves a pointer, which belongs to a position
in $\gamma$, to the right in every step. The only time it is moved to the left is when it hits a
closing parenthesis and $p$ has a value at the corresponding $\pi$-position which is \emph{not}
$-1$. In that case, we move to the matching opening parenthesis and continue to move to the right
from there until we hit the closing parenthesis again. Note that there is no \enquote{back-setting
to the left} in that process. Therefore, we can compute the position of the next $a$  in at most
quadratic time.

\paragraph*{Computing the Automata and Equivalence Test.}
As we did before, we construct a (not necessarily complete) deterministic finite automaton for each
input $\pi$-term $\gamma$. The automaton accepts exactly those sequences of factorizations which
can be applied to $\subs{\gamma}{\omega + \omega^*}$ and which need to be tested for the variety
in question. We will only demonstrate the details of the construction for the variety $\Rm[m + 1]
\cap \Lm[m + 1]$ and the input term $\alpha$; the construction for the other varieties is similar.
We need the variables \var{core} and \var{fringe}, which contain subsets of $(\dom{u} \times
\dom{u}) \times \smallset{1, 2, \dots, m}$ where $u = \subs{\gamma}{\omega + \omega^*}$ such that
all pairs $(l, r) \in \dom{u} \times \dom{u}$ are normalized. The algorithm works as follows:
\begin{algorithmic}
  \State $\var{core} \gets \emptyset$
  \State $\var{fringe} \gets \smallset{((-\infty, +\infty), m)}$
  \While{$\var{fringe} \neq \emptyset$}
    \State Remove $((l, r), k)$ from \var{fringe}
    \If{$k > 0$}
      \ForAll{$a \in \Sigma$}
        \State Compute $(l', r') = (l, r) \cdot X_a^L$ if defined
        \Comment{time limited by $\landO(n^2)$}
        \If{$(l', r')$ is defined}
          \State $(\bar{l}', \bar{r}') \gets \overline{(l', r')}^\alpha$
          \Comment{time limited by $\landO(n)$}
          \State Save transition
          
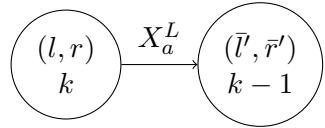
\begin{center}
            \begin{tikzpicture}[->, auto, align=center]
              \node[state] (lr) {$(l, r)$\\ $k$};
              \node[state, right=of lr] (l'r') {$(\bar{l}', \bar{r}')$\\ $k - 1$};
              \path (lr) edge node {$X_a^L$} (l'r');
            \end{tikzpicture}
          \end{center}
          \State Add $((\bar{l}', \bar{r}'), k - 1)$ to \var{fringe} unless it is in \var{core}
        \EndIf
        \State Compute $(l', r') = (l, r) \cdot X_a^R$ if it is defined
        \If{$(l', r')$ is defined}
          \State $(\bar{l}', \bar{r}') \gets \overline{(l', r')}^\alpha$
          \State Save transition
          
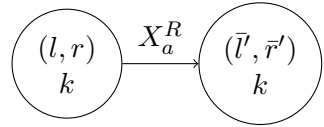
\begin{center}
            \begin{tikzpicture}[->, auto, align=center]
              \node[state] (lr) {$(l, r)$\\ $k$};
              \node[state, right=of lr] (l'r') {$(\bar{l}', \bar{r}')$\\ $k$};
              \path (lr) edge node {$X_a^R$} (l'r');
            \end{tikzpicture}
          \end{center}
          \State Add $((\bar{l}', \bar{r}'), k)$ to \var{fringe} unless it is in \var{core}
        \EndIf
        \State Handle $Y_a^L$ and $Y_a^R$ analogously
        \ForAll{$b \in \Sigma$} \Comment{Special case which is only required for $\equiv_m^\WI$}

          \State Compute $(l', r') = (l, r) \cdot C_{a, b}$ if it is defined
          \If{$(l', r')$ is defined}
            \State $(\bar{l}', \bar{r}') \gets \overline{(l', r')}^\alpha$
            \State Save transition
            
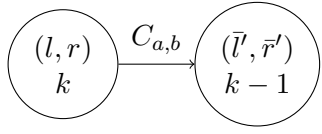
\begin{center}
              \begin{tikzpicture}[->, auto, align=center]
                \node[state] (lr) {$(l, r)$\\ $k$};
                \node[state, right=of lr] (l'r') {$(\bar{l}', \bar{r}')$\\ $k - 1$};
                \path (lr) edge node {$C_{a, b}$} (l'r');
              \end{tikzpicture}
            \end{center}
            \State Add $((\bar{l}', \bar{r}'), k - 1)$ to \var{fringe} unless it is in \var{core}
          \EndIf
        \EndFor
      \EndFor
      \State Add $((l, r), k)$ to \var{core}
    \EndIf
  \EndWhile
  \State Set $((-\infty, +\infty), m)$ as initial state
  \State Mark all states as final states
\end{algorithmic}
\noindent
Clearly, the resulting automaton accepts the desired factorization sequences. We know that any
normalized pair can be
encoded by the variables in \autoref{tab:variablesForLAndR}. Therefore, the number of such pairs
is limited by $\landO(n^5)$ when $n$ is the length of $\gamma$ seen as a finite word. This directly
yields that the constructed automaton has at most $\landO(n^5 m)$ states. Because the outer loop handles
any $((l, r), k)$ at most once the algorithm's running time is limited by $\landO(n^7 m)$.

After the construction of the two automata for $\alpha$ and $\beta$, we need to test
them for equivalence. They are equivalent if and only if $\subs{\alpha}{\omega + \omega^*}
\equiv_m^\WI \subs{\beta}{\omega + \omega^*}$ holds. The test for equivalence can be done by the
algorithm of Hopcroft and Karp \cite{hopcroft1971linear} in almost
linear time in the size of the automata. Since the number of states is bounded by
$\landO(n^5 m)$, the total running time of our complete algorithm is dominated by $\landO(n^7 m^2)$.

The automata for the varieties $\Rm$, $\Lm$ and $\Rm \vee \Lm$ need to store whether the last
factor was obtained by an element from $X_\Sigma^D$ or by an element from $Y_\Sigma^D$.
This information, however, is only of constant size and, therefore, does not change the asymptotic
running time of the overall algorithm. For $\DA$ we can omit the counting for $m$ in $k$ by
\autoref{fct:trotterWeilIsDA}, which even reduces the number of states.

This finally shows the following theorem.
\begin{restatable}{theorem}{wordProblemsInP}
  \label{thm:wordProblemsInP}
  The word problems for $\pi$-terms over $\Rm$, $\Lm$, $\Rm \vee \Lm$ and $\Rm \cap \Lm$
  can be solved by a deterministic algorithm with running time in $\landO(n^7 m^2)$ where $n$
  is the length of the input $\pi$-terms. Moreover, the word problem for $\pi$-terms over
  $\DA$ can be solved by a deterministic algorithms in time $\landO(n^7)$.
\end{restatable}

\section{Separability}\label{sec:separability}

Two languages $L_1, L_2 \subseteq \Sigma^*$ are \emph{separable} by a variety $\V$ if
there is a language $S \subseteq \Sigma^*$ with $L_1 \subseteq S$ and
$L_2 \cap S = \emptyset$ such that $S$ can be recognized by a monoid $M \in \V$. The
\emph{separation problem} of a variety $\V$ is the problem to decide whether two regular input
languages of finite words are separable by $\V$.

We are going to show the decidability of the separations problems of $\Rm$ for all $m \in \Nat$ as
well as for $\DA$ using the techniques presented in this paper\footnote{Decidability of the
separation problem of $\DA$ is already known \cite{place2013separating}. The proof, however, uses
a fixed point algorithm, which is
different from our approach.}. Note that, by symmetry, this also shows decidability for $\Lm$.

The general idea is as follows. If the input languages are separable, then we can find a separating
language $S$ which is recognized by a monoid in the variety in question. We can do this by
recursively enumerating all monoids and all languages in a suitable representation. For the other
direction, we show that, if the input languages are inseparable, then there are $\pi$-terms $\alpha$
and $\beta$ which witness their inseparability. Since we can also recursively enumerate these
$\pi$-terms, we have decidability.

To construct suitable $\pi$-terms we need an additional combinatorial property of the
$\equiv_{m, n}^X$ relations (which, in a slightly different form, can also be found in
\cite{kufleitner2012logical}).
\begin{restatable}{lemma}{XmForMTwo}
  \label{lem:XmForMTwo}
  Let $n, m \in \Nat$ with $m \geq 2$ and let $u \equiv_{m, n}^X v$ for two accessible words $u$ and $v$.
  Then, $u \cdot X_a^L \equiv_{m, n - 1}^X v \cdot X_a^L$ holds for all $a \in \alphabet(u) =
  \alphabet(v)$.
\end{restatable}
\begin{proof}
  We prove the lemma by induction over $n$. For $n = 1$, the assertion is satisfied by definition.
  Therefore, assume we have $u \equiv_{m, n + 1}^X v$ and we want to show $u_0 := u \cdot X_a^L
  \equiv_{m, n}^X v \cdot X_a^L =: v_0$. We already have $u_0 \equiv_{m - 1, n}^Y v_0$ by
  definition of $\equiv_{m, n + 1}^X$. This especially implies $\alphabet(u_0) = \alphabet(v_0)$
  since we have $m \geq 2$ and $n \geq 1$, as well as $u_0 \equiv_{m - 1, n - 1}^Y v_0$.
  Additionally, we have
  \[
    u_0 \cdot X_b^L = u \cdot X_b^L \equiv_{m - 1, n}^Y v \cdot X_b^L = v_0 \cdot X_b^L
  \]
  for all $b \in \alphabet(u_0) = \alphabet(v_0)$, which implies $u_0 \cdot X_b^L
  \equiv_{m - 1, n - 1}^Y v_0 \cdot X_b^L$. All that remains to
  be shown is that $u_0 \cdot X_b^R \equiv_{m, n - 1}^X v_0 \cdot X_b^R$ holds for all $b \in
  \alphabet(u_0) = \alphabet(v_0)$. Applying induction on
  $u \cdot X_b^R \equiv_{m, n}^X v \cdot X_b^R$ (for the same $a$) yields $u \cdot X_b^R X_a^L
  \equiv_{m, n - 1}^X v \cdot X_b^R X_a^L$. Since we have $u_0 \cdot X_b^R = u \cdot X_a^L X_b^R =
  u \cdot X_b^R X_a^L$ and $v_0 \cdot X_b^R = v \cdot X_a^L X_b^R = v \cdot X_b^R X_a^L$, we are
  done.
\end{proof}

Using this property, one can prove the following lemma about the $\pi$-term construction.
\begin{restatable}{lemma}{piTermConstruction}
  \label{lem:piTermConstruction}
  Let $M$ be a monoid, $\varphi: \Sigma^* \to M$ a homomorphism and $m \in \Nat_0$. Let $(u_n, v_n)_{n \in \Nat_0}$ be an infinite
  sequence of word pairs $(u_n, v_n)_{n \in \Nat_0}$ with
  \begin{multicols}{2}
    \begin{itemize}
      \item $u_n, v_n \in \Sigma^*$,
      \item $u_n \equiv_{m, n}^X v_n$,
      \item $\varphi(u_n) = m_u$ and
      \item $\varphi(v_n) = m_v$
    \end{itemize}
  \end{multicols}
  \noindent{}for fixed monoid elements
  $m_u, m_v \in M$ and all $n \in \Nat_0$. Then, the sequence yields $\pi$-terms
  $\alpha$ and $\beta$ (over $\Sigma$) such that $\varphi \left( \subs{\alpha}{M!} \right) = m_u$,
  $\varphi \left( \subs{\beta}{M!} \right) = m_v$ and $\subs{\alpha}{\omega + \omega^*}
  \equiv_{m}^X \subs{\beta}{\omega + \omega^*}$ hold.
\end{restatable}
\noindent{}%
Before we give a proof of the general case, we give a separate one for the case $m = 1$. It is basically
an adaption of the ideas from the proof showing decidability of the separation problem for the
variety $\V[J]$ of $\mathcal{J}$-trivial monoids given by van Rooijen and Zeitoun
\cite{roojien2013separation} to our setting.
\begin{proof}[Proof ($m = 1$)]
  This proof is based on Simon's Factorization Forest Theorem \cite{simon1990factorization}. For a finite word
  $w \in \Sigma^+$, a factorization tree is a rooted, finite, unranked, labeled ordered tree such
  that
  \begin{itemize}
    \item the tree's root is labeled with $w$,
    \item the leaves are labeled with letters (from $\Sigma$) and
    \item any internal node has at least two children and, if its children are labeled with $w_1,
    w_2, \dots,\allowbreak w_k \in \Sigma^+$, then the node is labeled with $w_1 w_2 \dots w_k$.
  \end{itemize}
  For every homomorphism $\psi: \Sigma^* \to N$ into a monoid $N$, Simon's Factorization Forest
  Theorem yields a factorization tree for every finite word $w \in \Sigma^+$ such that $\psi$ maps the
  labels of a node's children to the same idempotent in $N$ if the node has at least three
  children. Furthermore, the tree's height\footnote{A single node has height $0$.} is finite and
  limited by some constant that solely depends on $|N|$ (and, in particular, not on $w$).

  Before we begin with the actual proof, we note that, if we remove pairs from the sequence
  $(u_n, v_n)_{n \in \Nat_0}$ and still have an infinite sequence, then the resulting sequence
  still satisfies all conditions stated above, in particular $u_n \equiv_{1, n}^X v_n$.

  We extend $\varphi$ into a homomorphism $\Sigma^* \to M \times 2^{\Sigma}$ which maps a word $w$
  to its alphabet $\alphabet(w)$ for the second component.\footnote{$2^{\Sigma}$ is the monoid of
  all subsets of $\Sigma$ with taking union as the monoid's operation.} Then, we observe that there
  has to be an infinite subsequence such that all first components as well as all second components
  have the same alphabet. Indeed, these two alphabets have to coincide by the
  definition of $\equiv_{1, n}^X$! We remove all other words from the sequence. If the remaining words
  $u_n$ and $v_n$ are all empty (i.\,e.\ they have alphabet $\emptyset$), we can choose $\alpha =
  \beta = \varepsilon$ as well. Otherwise, we apply Simon's Factorization Forest Theorem to the
  remaining words $u_n$ and $v_n$, which yields a
  sequence of factorization tree pairs $(T_{u, n}, T_{v, n})$. We first construct $\alpha$ from
  $(T_{u, n})_{n \in \Nat_0}$ such that we have $\varphi(\subs{\alpha}{M!}) = m_u$ and the
  following conditions:
  \begin{itemize}
    \item If $w \in \Sigma^*$ is a subword\footnote{Recall that a finite word $u = a_1 a_2 \dots a_n$ with
    $a_i \in \Sigma$ is a subword of a (not necessarily finite) word $v$ if we can write
    $v = v_0 a_1 v_1 a_2 v_2 \dots a_n v_n$ for some words $v_0, v_1, \dots, v_n$.} of $u_n$ for
    an $n \in \Nat_0$, then $w$ is a subword of $\subs{\alpha}{\omega + \omega^*}$.
    \item If $w \in \Sigma^*$ is a subword of $\subs{\alpha}{\omega + \omega^*}$, then it is a
    subword of all $u_n$ with $n \geq n_0$ for an $n_0 \in \Nat_0$.
  \end{itemize}
  Afterwards, we proceed with $(T_{v, n})_{n \in \Nat_0}$ to construct $\beta$ in the same manner.

  We may assume that all trees $T_{u, n}$ have the same height $H$ as the height is bounded by
  a constant and we can remove all words $u_n$ from the underlying sequence which
  yield a tree not of height $H$. If $H$ is zero, all trees consist of a single leaf and all
  words $u_n$ consist of a single letter. Among these, one letter $a \in \Sigma$ has to appear
  infinitely often; we remove all other words from the sequence and choose $\alpha = a$. Clearly,
  all conditions for $\alpha$ are satisfied.

  For $H > 0$, we consider the situation at the root of each $T_{u, n}$. Let $u_{n, 1}, u_{n, 2},\allowbreak
  \dots, u_{n, K_n}$ be the labels of the root's children in $T_{u, n}$. If the sequence
  $(K_n)_{n \in \Nat_0}$ is bounded, there is an infinite subsequence such that $K_n$ is equal to
  a specific $K \geq 2$ for all indexes $n$ of the subsequence; we remove all words not belonging to this
  subsequence. In the result, there is an infinite subsequence such that, for each sequence
  $(u_{n, k})_{n \in \Nat_0}$ with $1 \leq k \leq K$, all $u_{n, k}$ get mapped to the same monoid
  element by $\varphi$; we remove all other words. As each child of the root yields a subtree,
  taking these subtrees gives $K$ infinite sequences of factorization trees of height $H - 1$.
  Applying induction on $H$ yields $\alpha_1, \alpha_2, \dots, \alpha_K$. We define $\alpha :=
  \alpha_1 \alpha_2 \dots \alpha_K$. Because $\alpha_1, \alpha_2, \dots, \alpha_K$ satisfy the
  conditions stated above for their respective subtree sequence, so does $\alpha$ for
  $\left( T_{u, n} \right)_{n \in \Nat_0}$.

  If the sequence $(K_n)_{n \in \Nat_0}$ is unbounded, we can, without loss of generality, assume
  $K_n \geq 3$ for all $n \in \Nat_0$ and that it is strictly increasing -- again taking the
  appropriate infinite subsequence. Also,
  we can assume that all $u_{n, 1}, u_{n, 2}, \dots, u_{n, K_n}$ get mapped to the same idempotent
  $e \in M \times 2^{\Sigma}$. Choose $w \in \varphi^{-1}(e)$ arbitrarily and define
  $\alpha := (w)^\pi$. Note that we now have $\alphabet(u_{n, 1}) = \alphabet(u_{n, 2}) = \dots =
  \alphabet(u_{n, K_n}) = \alphabet(u_n) = \alphabet(w)$ for all $n \in \Nat_0$. Therefore,
  $\alpha$ satisfies the conditions above.

  All which remains to be shown is that we now have $\subs{\alpha}{\omega + \omega^*} \equiv_1^X
  \subs{\beta}{\omega + \omega^*}$. The important observation here is that $w_1
  \equiv_{1, n}^X w_2$ with $n \in \Nat_0$ holds if and only if $w_1$ and $w_2$ have the same
  subwords of length $\leq n$. This means we have to show that
  $\subs{\alpha}{\omega + \omega^*}$ and $\subs{\beta}{\omega + \omega^*}$ have the same subwords
  (of arbitrary length). To show the subword equality, assume $w$ is a subword of
  $\subs{\alpha}{\omega + \omega^*}$ (without loss of generality). By the conditions above, $w$ is
  a subword of all $u_n$ with $n \geq n_0$ for an $n_0 \in \Nat_0$. Let $\tilde{n} = \max \{ n_0,
  |w| \}$. Since we have $u_{\tilde{n}} \equiv_{1, \tilde{n}}^X v_{\tilde{n}}$ and by applying our
  observation regarding subwords and $\equiv_{1, \tilde{n}}^X$, $w$ is a subword of $v_{\tilde{n}}$
  and, thus, a subword of $\subs{\beta}{\omega + \omega^*}$.
\end{proof}

Next, we give a proof for the general case $m \geq 1$.
\begin{proof}[Proof of \autoref{lem:piTermConstruction}]
  The assertion is trivial for $m = 0$. The case $m = 1$ has already been covered.
  For $m > 1$, we proceed by induction over $|\Sigma|$. For $\Sigma
  = \emptyset$, we set $\alpha = \beta = \varepsilon = u_n = v_n$. For $|\Sigma| > 0$, remember
  the observation from the previous proof: if we take an infinite subsequence $(u'_n, v'_n)_{n \in \Nat_0}$ of
  $(u_n, v_n)_{n \in \Nat_0}$, this sequence will still satisfy all conditions of the lemma.
  In particular, we will still have $u'_n \equiv_{m, n}^X v'_n$ for all $n \in \Nat_0$.

  Now, we factorize $u_n = w_{n, 0} a_{n, 0} w_{n, 1} a_{n, 1} \dots w_{n, K_n} a_{n, K_n}
  w_{n, K_n + 1}$ for all $n \in \Nat_0$ such that $\alphabet(w_{n, k}) = \alphabet(u_n) \setminus
  \{ a_{n, k} \}$ for all $k \in \{ 0, 1, \dots, K_n \}$ and $\alphabet(w_{n, K_n + 1})
  \subsetneq \alphabet(u_n)$. If the sequence $(K_n)_{n \in \Nat_0}$ is bounded, let $K$ be one of
  the numbers which appear infinitely often in it and restrict all further considerations to the
  corresponding subsequence of words. If $(K_n)_{n \in \Nat_0}$ is unbounded, let
  $K = |M|^2 + 1$ and remove all word pairs $(u_n, v_n)$ for which $K_n$ is smaller than $K$ from the
  sequence. For all $k \in \{ 0, 1, \dots, K \}$, a single letter $a_k \in \Sigma$ has to appear
  infinitely often in the sequence $(a_{n, k})_{n \in \Nat_0}$ because $\Sigma$ is of finite size.
  We restrict our consideration to the corresponding subsequence. Then, we define
  $x_{n, k} = u_n \cdot X_{a_{0}}^R X_{a_{1}}^R \dots X_{a_{k - 1}}^R X_{a_{k}}^L \text{ and }
  y_{n, k} = v_n \cdot X_{a_{0}}^R X_{a_{1}}^R \dots X_{a_{k - 1}}^R X_{a_{k}}^L$
  for $k \in \{ 0, 1, \dots, K \}$ as well as $x_{n, K + 1} = u_n \cdot X_{a_{0}}^R X_{a_{1}}^R
  \dots\allowbreak X_{a_{K}}^R$ and $y_{n, K + 1} = v_n \cdot X_{a_{0}}^R X_{a_{1}}^R \dots X_{a_{K}}^R$. We,
  thus, have $u_n = x_{n, 0} a_{0} x_{n, 1} a_{1} \dots x_{n, K}\allowbreak a_{K}\allowbreak x_{n, K + 1}$ and $v_n =
  y_{n, 0} a_{0} y_{n, 1} a_{1} \dots y_{n, K} a_{K} y_{n, K + 1}$
  for all $n \in \Nat_0$. Because $K$ is constant, we can safely assume that $\varphi$ maps all
  elements of the sequence $(x_{n, k})_{n \in \Nat_0}$ (for every $k \in \{ 0, 1, \dots, K + 1 \}$) to the
  same element $s_k \in M$: one element has to appear infinitely often and we take the corresponding
  subsequence. In the same way, we can ensure that $\varphi$ maps all element of
  $(y_{n, k})_{n \in \Nat_0}$ to the same element $t_k \in M$ (again, for all $k \in \{ 0, 1, \dots,
  K + 1 \}$). By removing the first $K + 2$ pairs of words, we can also ensure $u_n
  \equiv_{m, n + K + 2}^X v_n$ for all $n \in \Nat_0$. This implies $x_{n, k}
  \equiv_{m, n + K + 2 - k - 1}^X y_{n, k}$ for all $n \in \Nat_0$ and all $k \in \{ 0, 1, \dots,
  K \}$ by \autoref{lem:XmForMTwo}. Directly by the definition of $\equiv_{m, n}^X$, we already have
  $x_{n, K + 1} \equiv_{m, n + K + 2 - K - 1}^X y_{n, K + 1}$ and, therefore,
  $x_{n, k} \equiv_{m, n}^X y_{n, k}$ for all $k \in \{ 0, 1, \dots, K + 1 \}$.
  We can apply induction to $(x_{n, k}, y_{n, k})_{n \in \Nat_0}$ for $k \in \{ 0, 1, \dots, K \}$
  since we have $a_k \not\in \alphabet(x_{n, k})$ by construction. This yields $\pi$-terms
  $\alpha_0, \alpha_1, \dots, \alpha_K, \beta_0, \beta_1, \dots, \beta_K$. If
  $(K_n)_{n \in \Nat_0}$ was bounded, then $\alphabet(x_{n, K + 1}) =
  \alphabet(y_{n, K + 1}) \subsetneq \alphabet(u_n) = \alphabet(v_n)$ holds and we can apply induction as
  well, which yields $\pi$-terms $\alpha_{K + 1}$ and $\beta_{K + 1}$. Setting $\alpha = \alpha_0
  a_0 \alpha_1 a_1 \dots \alpha_K a_K \alpha_{K + 1}$ and $\beta = \beta_0 a_0 \beta_1 a_1 \dots
  \beta_K a_K \beta_{K + 1}$ satisfies $\subs{\alpha}{\omega + \omega^*} \equiv_{m}^X
  \subs{\beta}{\omega + \omega^*}$ since $\equiv_m^X$ is a congruence. If $(K_n)_{n \in \Nat_0}$
  was unbounded, we set $K = |M|^2 + 1$ and, by the pigeon hole principle, there are $i, j \in
  \{ 0, 1, \dots, K \}$ with $i < j$ and
  \begin{align*}
    s_0 \varphi(a_0) s_1 \varphi(a_1) \dots s_{i} \varphi(a_{i}) &=
      s_0 \varphi(a_0) s_1 \varphi(a_1) \dots s_{j} \varphi(a_{j}) \text{ and}\\
    t_0 \varphi(a_0) t_1 \varphi(a_1) \dots t_{i} \varphi(a_{i}) &=
      t_0 \varphi(a_0) t_1 \varphi(a_1) \dots t_{j} \varphi(a_{j}) \text{.}
  \end{align*}
  We define
  \begin{align*}
    \alpha &= \alpha_0 a_0 \alpha_1 a_1 \dots \alpha_i a_i \left( \alpha_{i + 1} a_{i + 1}
      \alpha_{i + 2} a_{i + 2} \dots \alpha_{j} a_{j} \right)^{\pi} \alpha_{K + 1} \text{ and}\\
    \beta &= \beta_0 a_0 \beta_1 a_1 \dots \beta_i a_i \left( \beta_{i + 1} a_{i + 1}
      \beta_{i + 2} a_{i + 2} \dots \beta_{j} a_{j} \right)^{\pi} \beta_{K + 1}
  \end{align*}
  where $\alpha_{K + 1}$ and $\beta_{K + 1}$ are obtained by using induction on $m$ (and symmetry)
  for the sequences $(x_{n, K + 1})_{n \in \Nat_0}$ and $(y_{n, K + 1})_{n \in \Nat_0}$, i.\,e.\ we
  have $\varphi(\subs{\alpha_{K + 1}}{M!}) = s_{K + 1}$, $\varphi(\subs{\beta_{K + 1}}{M!}) = t_{K +
  1}$ and $\subs{\alpha_{K + 1}}{\omega + \omega^*} \equiv_{m - 1}^Y \subs{\beta_{K + 1}}{\omega +
  \omega^*}$. Therefore, we have $\varphi\left( \subs{\alpha}{M!} \right) = m_u$ and $\varphi\left(
  \subs{\beta}{M!} \right) = m_v$ by construction. We also have $\subs{\alpha}{\omega + \omega^*}
  \equiv_m^X \subs{\beta}{\omega + \omega^*}$: for the part left up to and including the
  $(\cdot)^\pi$, we have equivalence by induction and because $\equiv_m^X$ is a congruence; the
  right part, we cannot reach by arbitrarily many $X_a^D$ factorizations since all letters appear
  infinitely often in the $(\cdot)^\pi$ part and, if we reach it by using at least one $Y_a^D$
  factorization, we are done since we have $\subs{\alpha}{\omega + \omega^*} \equiv_{m - 1}^Y
  \subs{\beta}{\omega + \omega^*}$.
\end{proof}

We can now plug everything together and prove the following theorem.
\begin{restatable}{theorem}{separationIsDecidable}
  \label{thm:separationIsDecidable}
  For each $m \in \Nat$, the separation problem for $\Rm$ is decidable and so is the one for $\Lm$.
\end{restatable}
\begin{proof}
  We only consider $\Rm$ as the case for $\Lm$ is symmetric. If the input languages are separable,
  we can find a separating language by enumerating all candidates. If the languages are
  inseparable, we have to apply the previous lemma. As regular languages, the input languages $L_1
  \subseteq \Sigma^*$ and $L_2 \subseteq \Sigma^*$ can be recognized by monoids $M_1$ and $M_2$ via
  the homomorphisms $\varphi_1$ and $\varphi_2$ and the homomorphism can be computed. Therefore,
  they are also recognized by $M := M_1 \times M_2$ via the homomorphism $\varphi$ which maps a
  finite word to a pair whose first component is determined by $\varphi_1$ and whose second component is
  determined by $\varphi_2$. Let $n \in \Nat_0$ be arbitrary. Since we have $\Sigma^* / {\equiv_{m,
  n}^X} \in \Rm$ and since $L_1$ and $L_2$ cannot be separated by $\Rm$, there have to be finite
  words $u_n, v_n \in \Sigma^*$ with $u_n \in L_1$, $v_n \in L_2$ and $u_n \equiv_{m, n}^X v_n$;
  otherwise, we could construct a separating language. The homomorphism $\varphi$ has to map
  infinitely many elements of the sequence $(u_n, v_n)_{n \in \Nat_0}$ to the same element in $M$
  since $M$ is finite. If we remove all other elements, we still have an infinite sequence $(u_n,
  v_n)_{n \in \Nat_0}$ with $u_n \equiv_{m, n}^X v_n$ for all $n \in \Nat_0$ which also satisfies
  all conditions of \autoref{lem:piTermConstruction}. Therefore, there are $\pi$-terms $\alpha$
  and $\beta$ with $\subs{\alpha}{\omega + \omega^*} \equiv_{m}^X \subs{\beta}{\omega + \omega^*}$,
  $\varphi \left( \subs{\alpha}{M!} \right) \in \varphi(L_1)$ and $\varphi \left( \subs{\beta}{M!}
  \right) \in \varphi(L_2)$. Since we can test whether $\subs{\alpha}{\omega + \omega^*}
  \equiv_{m}^X \subs{\beta}{\omega + \omega^*}$ holds for any two $\pi$-terms $\alpha$ and $\beta$
  by \autoref{thm:RelationsAreDecidable}, we can also recursively enumerate all possible
  $\pi$-term pairs and check whether the conditions above are met. We know that we can find such a
  pair if $L_1$ and $L_2$ are inseparable. On the other hand, suppose $L_1$ and $L_2$ can be
  separated by $S \subseteq \Sigma^*$ which is recognized by the monoid $N \in \Rm$ via a
  homomorphism $\psi: \Sigma^* \to N$ and we have found a pair $\alpha$ and $\beta$ with
  $\subs{\alpha}{\omega + \omega^*} \equiv_{m}^X \subs{\beta}{\omega + \omega^*}$, $\varphi \left(
  \subs{\alpha}{M!} \right) \in \varphi(L_1)$ and $\varphi \left( \subs{\beta}{M!} \right) \in
  \varphi(L_2)$. Then, we have $\varphi \left( \subs{\alpha}{N! \cdot M!} \right) = \varphi \left(
  \subs{\alpha}{M!} \right) \in \varphi(L_1)$ and, thus, $\subs{\alpha}{N! \cdot M!} \in L_1$ as
  well as $\subs{\beta}{N! \cdot M!} \in L_2$ (by a similar argument). Also, $\alpha = \beta$ holds
  in $\Rm$ by \autoref{thm:relationsAndEquations}, which implies $s := \psi \left(
  \subs{\alpha}{N! \cdot M!} \right) = \psi \left( \subs{\alpha}{N!} \right) =
  \psi \left( \subs{\beta}{N!} \right) = \psi \left( \subs{\beta}{N! \cdot M!} \right)$. If we
  have $s \in \psi(S)$, then we have $\subs{\beta}{N! \cdot M!} \in S \cap L_2$; otherwise,
  we have $\subs{\alpha}{N! \cdot M!} \in L_1$ but $\subs{\alpha}{N! \cdot M!} \not\in S$ and,
  thus, a contradiction in either case.
\end{proof}

Since two languages are separable by $\Rm$ for $m = | \Sigma | + 1$ if they are separable by $\DA$
\cite{weis2009structure}, we also get decidability of the separation problem of
$\DA$, which has already been shown by Place, van Rooijen and Zeitoun \cite{place2013separating}.
\begin{corollary}
  The separation problem for $\DA$ is decidable.
\end{corollary}

\bibliographystyle{plain}
\bibliography{citations}

\end{document}